\documentclass[a4paper, 10pt]{article}
 
\usepackage{amsmath}
\usepackage{amssymb}
\usepackage{a4wide}
\usepackage{soul}
\usepackage{graphicx}
\usepackage{color}  
\usepackage{graphics}
\usepackage{epsfig}
\usepackage{bbm}
\usepackage{amsmath}
\usepackage{amsfonts}
\usepackage{mathrsfs}
\usepackage{times}
 \usepackage{fourier}
 \usepackage{tikz}

{%
\stepcounter{equation}
\begin{equation*}}{%
\leqno(\arabic{equation})
\end{equation*}
}
 \usepackage{hyperref} 
\newtheorem{theorem}{Theorem}[section]
\newtheorem{lemma}[theorem]{Lemma}
\newtheorem{proposition}[theorem]{Proposition}
\newtheorem{definition}[theorem]{Definition}
\newtheorem{remark}[theorem]{Remark}
\newtheorem{corollary}[theorem]{Corollary}

\newcommand{\fact}[1]{#1\mathpunct{}!}
\numberwithin{equation}{section}
\newcommand{\h}{\hbar} 
  {\setlength{\baselineskip}{1.5\baselineskip}

\newenvironment{proof}
 {\textit{\textbf{Proof.}}}
 {\hfill $\square$ \vskip 8pt}

\title{Long time semiclassical Egorov theorem for $\h$-pseudodifferential systems}
\author{Marouane Assal \\ \vspace{1pt} 
\small{Institut de Math\'ematiques de Bordeaux}\\
\small{351 cours de la Lib\'eration, 33405 Talence cedex, France} \\ 
\small{\textit{E-mail address :} marouane.assal@math.u-bordeaux1.fr }}
\begin{document}

\maketitle

\begin{abstract}
In the Heisenberg picture, we study the semiclassical time evolution of a bounded quantum observable $Q^w(x,\h D_x;\h)$ associated to a $(m\times m)$ matrix-valued symbol $Q$ generated by a semiclassical matrix-valued Hamiltonian $H\sim H_0+\h H_1$. Under a non-crossing assumption on the eigenvalues of the principal symbol $H_0$ that ensures the existence of almost invariant subspaces of $L^{2}(\mathbb R^n)\otimes \mathbb C^m$, and for a class of observables that are semiclassically block-diagonal with respect to the projections onto these almost invariants subspaces, we establish a long time matrix-valued version for the semiclassical Egorov theorem valid in a large time interval of Ehrenfest type $T(\h)\simeq \log(\h^{-1})$. 
\end{abstract} 
   


%\begin{abstract}
%We study the semiclassical time evolution of a $\h$-pseudodifferential system $Q^w(x,\h D_x)$ associated to a $(m\times m)$ matrix-valued symbol $Q$ using the Heisenberg representation. For a certain class of observables that is blocks-diagonal with respect to the eigenprojectors associated to the principal symbol of the Hamiltonian $H$, we construct a semiclassical approximation for $Q(t):=e^{\frac{it}{\h}H^w(x,\h D_x)} Q^w(x,\h D_x)e^{-\frac{it}{\h}H^w(x,\h D_x)}$ in terms of $\h$-pseudodifferential operators and we recover the Ehrenfest time for the validity of this approximation. 
%\end{abstract}
\maketitle
 
\section{Introduction}\label{Introo}
Known as Bohr's correspondence principle in physics and Egorov theorem in mathematical literature, the semiclassical approximation ensures the transition between quantum and classical evolutions of observables. The relation between the classical and quantum evolutions may be considered as one of the oldest problems of semiclassical analysis. The starting point was the famous Bohr's correspondence principle proposed by Niels Bohr on 1923. According to this principle, the quantum evolution of an observable is closer and closer to its classical evolution as the Planck constant $\h$ becomes negligible. For many years, the semiclassical approximation has been the object of several investigations following two approaches. The first one uses semiclassical wave packets (or coherent states) as initial data and aims to approximate the evolved wave packet by a linear combination of coherent states (see \cite{Paul}, \cite{Joye}, \cite{Joye1}, \cite{Combescure} and \cite{Com} for a complete description of coherent states). The second approach, that will be our principal subject in this paper, considers the Heisenberg evolution of suitable bounded observables and seeks to construct an approximation for the corresponding observables in terms of $\h$-pseudodifferential operators (see \cite{Wang}, \cite{Rob1}, \cite{Rob}, \cite{iv}). Let us start by recalling the general setting of this approach. %In the following, many basic notions of semiclassical analysis are used. We refer to Appendix B for a short review of this notions and related results. 

Let $H \in C^{\infty}(\mathbb R^{2n};\mathbb R)$ be a classical Hamiltonian and $Q \in C^{\infty}(\mathbb R^{2n};\mathbb R)$ be a real observable. We consider the system of Hamilton equations 
\begin{equation}\label{Hamilton equation}
\frac{dx}{dt} = \partial_{\xi} H(x,\xi), \quad \frac{d\xi}{dt}=-\partial_xH(x,\xi).
\end{equation}
These equations generate a flow $\phi_H^t$ called the Hamiltonian flow associated to $H$ defined by $\phi_H^t(x=x(0),\xi=\xi(0))=(x(t),\xi(t))$ and $\phi_H^t{_{|t=0}}(x,\xi)=(x,\xi)$. In the sequel, additional assumption will be required on $H$ so that for all initial data $(x,\xi)\in \mathbb R^{2n}$, $\phi_H^t(x,\xi)$ is well defined for all $t\in \mathbb R$.

The time evolution of $Q$ under the flow $\phi_H^t$ given by $q_0(t):=Q\circ \phi_H^t$ is described by the equation
\begin{equation}\label{1}
\frac{d}{dt} q_0(t)=\{H,Q\}\circ \phi_H^t, \quad q_0(t){_{|t=0} }= Q,
\end{equation}
where $\{H,Q\}$ denotes the Poisson bracket of $H,Q$ defined by 
$$
\{H,Q\}:=\partial_{\xi}H.\partial_xQ-\partial_xH . \partial_{\xi}Q.
$$
Let $H^w:= H^w(x,\h D_x)$ and $Q^w:= Q^w(x,\h D_x)$ be the self-adjoint operators in $L^2(\mathbb R^n)$ associated to $H$ and $Q$, respectively. Here we use the $\h$-Weyl quantization (see formula \eqref{Weyl}). Let $e^{-\frac{it}{\h}H^w}$ be the unitary operator solution of the evolution equation
$$
i\h \partial_t u(t) = H^w  u(t).
$$
The time evolution of the quantum observable $Q^w$ under $e^{-\frac{it}{\h}H^w}$ given by $Q(t):= e^{\frac{it}{\h}H^w}Q^w e^{-\frac{it}{\h}H^w}$ is described by the quantum analogous of (\ref{1}) called the Heisenberg equation of motion 
\begin{equation}\label{2}
\frac{d}{dt} Q(t)= \frac{i}{\h} [H^w,Q(t)], \quad Q(t)_{|t=0}=Q^w,
\end{equation}
where $[H^w,Q(t)]:=H^w Q(t)-Q(t)H^w$ denotes the commutator of $H^w,Q(t)$.

Under suitable growth assumptions on $H$, the semiclassical Egorov theorem (see \cite[Theorem 1.2]{Rob} for the precise statement) states that for every fixed $t$, $Q(t)$ is an $\h$-pseudodifferential operator with principal symbol $q_0(t,x,\xi)=Q\circ \phi_H^{t}(x,\xi)$. More precisely, there exists a family of symbols $(q_j(t,.,.))_{j\geq 1}$ with $\text{supp}(q_j(t,.,.)) \subset \phi_H^{-t}(K)$, where $K$ is the union of the supports of the $q_j$, such that for every finite time $\bar{t}>0$ and for all $N\in \mathbb N$, the following estimate
\begin{equation}\label{caldd}
\bigg\Vert  Q(t) - \sum_{j=0}^N \h^j \big( q_j(t)\big)^w \bigg\Vert_{\mathcal{L}(L^2(\mathbb R^n))} \leq C_N \h^{N+1},  
\end{equation}
holds uniformly for $|t|\leq \bar{t}$. The time-dependent family of operators $\big((q_j(t))^w\big)_{j\geq 0}$ is called semiclassical approximation of $Q(t)$. The construction of such approximation is obtained in a formal way by solving the Heisenberg equation (\ref{2}) in the space of formal power series of $\h$. If one assumes that $Q(t)$ is an $\h$-pseudodifferential operator with Weyl symbol $q(t)$, where for now we look at $q(t)$ as a formal power series of $\h$, $q(t)\sim \sum_{j\geq 0} \h^j q_j(t)$, by (\ref{2}) it fulfils the equation
\begin{equation}\label{3}
\displaystyle{\frac{d}{dt}} q(t)=\frac{i}{\h}[H,q(t)]_{\#}, \quad q(t)_{|t=0}=Q,
\end{equation}
where $[H,q(t)]_{\#}:=H\#q(t)-q(t)\#H$, with $H\#q(t)$ denotes the Moyal product of $H,q(t)$ defined as the symbol of $H^w\circ Q(t)$ (see appendix \ref{semi-classical background}).

By expanding the symbol $\frac{i}{\h}[H,q(t)]_{\#}$ in powers of $\h$ and then equating equal powers of $\h$ in both sides, the Cauchy problem (\ref{3}) implies similar Cauchy problems for the symbols $q_j(t)$ given by
\begin{equation}\label{Cauchy problems} (\mathcal{C}_j)
\left\{   
\begin{array}{rcl}
\displaystyle{\frac{d}{dt}} q_j(t)&=& \bigg(  \frac{i}{\h}[H,q(t)]_{\#} \bigg)_j, \\
q_0(t)_{|t=0}&=& Q, \quad q_j(t)_{|t=0}= 0, \forall j\geq 1,
\end{array} 
\right.
\end{equation}
where $\bigg(  \frac{i}{\h}[H,q(t)]_{\#} \bigg)_j$ is the $j$-th term (i.e. the coefficient of $\h^j$) in the asymptotic expansion of $\frac{i}{\h}[H,q(t)]_{\#}$.

For $j=0$, the principal symbol of $[H,q(t)]_{\#}$ given by the commutator $[H,q_0(t)]$ vanishes. Therefore, the term of order $\h^0$ in the asymptotic expansion of $\frac{i}{\h}[H,q(t)]_{\#}$ coincides with the sub-principal symbol of $i[H,q(t)]_{\#}$ which is equal to $\{H_0,q_0(t)\}$. Consequently, $(\mathcal{C}_0)$ reads 
$$
\frac{d}{dt}q_0(t)= \{H, q_0(t)\}, \quad q_0(t)_{|t=0}=Q,
$$
and we get $q_0(t)=Q \circ \phi_H^t$. The symbols $q_j(t)$ for $j\geq 1$ are obtained in the same manner by following the above algorithm (see Remark \ref{les symboles dans le cas scalaire} for the general formula of the symbols $q_j(t)$).

The semiclassical approximation given by the semiclassical Egorov theorem is limited to the evolutions in finite time intervals. Several works was devoted to the investigation of its validity for large times which may depend on $\h$ (see \cite{Combescure} for the time evolution of coherent states, and \cite{Bamb}, \cite{Rob} for the time evolution of observables). These investigations was based on a conjecture going back to the physicists Chirikov and Zaslavski (\cite{Chiri},\cite{Zas}) which claim that the semiclassical approximation remains valid in a large time interval of length $T(\h)\simeq \log(\h^{-1})$ known as the Ehrenfest time. The optimal time for which this claim was proved has been obtained by Bouzouina and Robert \cite{Rob}. They showed that the $L^2$-operator norm of the remainder term in the asymptotic expansion of $Q(t)$ given by the left hand side of (\ref{caldd}) is uniformly dominated, at any order, by an exponential term whose argument is linear in time. In particular, this allows them to recover the Ehrenfest time for the validity of the semiclassical approximation. 

The purpose of this paper is to study the extension of the above results to the case of matrix-valued observables. Given a bounded quantum observable $Q^w(x,\h D_x;\h)$ associated to a ($m\times m$) matrix-valued symbol $Q(x,\xi;\h)\sim \sum_{j\geq 0}\h^j Q_j(x,\xi)$, we study the time evolution $Q(t):= e^{\frac{it}{\h}H^w} Q^w e^{-\frac{it}{\h}H^w}$ generated by a semiclassical $(m \times m)$ hermitian-valued Hamiltonian $H(x,\xi;\h)$. We establish a long time matrix-valued version for the semiclassical Egorov theorem by giving a semiclassical approximation for $Q(t)$ valid in a large time interval of Ehrenfest type. To the best of our knowledge, there is no results concerning the large time behaviour of the semiclassical approximation in the matricial case.

Matrix-valued version for Egorov's theorem has been discussed several times in the literature (\cite{Corde}, \cite{Pan}, \cite{Brum}, \cite{bolte}, \cite{iv}). Brummelhuis and Nourrigat \cite{Brum} have studied the particular case of matrix-valued Hamiltonian with scalar principal symbol and proved an extension of the semiclassical Egorov theorem valid for evolutions in finite time intervals. 
%This particular case was also treated by Ivrii \cite[ch 2.3]{iv} which established long time approximation for the Heisenberg evolution. 
This result has been extended to the general case by Bolte and Glaser \cite{bolte} under an assumption on the gap between the eigenvalues of the principal symbol of the Hamiltonian (see assumption \textbf{(A1)} in the next section). However, their result is again only valid for finite time. Here we require the same assumption as in \cite{bolte} and we are concerned with the large time behaviour of the approximation. Some ideas from \cite{bolte} are still present here.

Let us explain the main difficulties arising from the matrix structure of the problem. By going back to the Cauchy problem $(\mathcal{C}_0)$ satisfied by the principal symbol $q_0(t)$, one immediately sees that in the case where $H$ and $Q$ are matrix-valued functions, the principal symbol of the Moyal commutator $[H,q(t)]_{\#}$ which is equal to the matrix commutator $[H_0,q_0(t)]$ is no longer zero. Here $H_0$ denotes the principal symbol of $H$. Then, at leading semiclassical order, we have an equation of the type
\begin{equation}\label{leading semiclassical order}
\frac{d}{dt}q_0(t)= \frac{i}{\h}[H_0,q_0(t)]+\mathcal{O}(\h^0), \quad \h \searrow 0.
\end{equation}
In order to get a solvable equation for $q_0(t)$, the factor $\h^{-1}$ forces us to restrict ourself to a class of observables for which the commutativity between $H_0$ and $q_0(t)$ is preserved under the time evolution. For $t=0$, this is equivalent to a block-diagonal form of $Q_0$ with respect to the eigenprojectors of $H_0$. Under a non-crossing assumption on the eigenvalues of $H_0$, we use an idea due to Helffer and Sj\"ostrand \cite{Hel} which consists in decomposing the Hilbert space $L^2(\mathbb R^n)\otimes \mathbb C^m$ into almost invariant subspaces with respect to the time evolution generated by $H^w$. By considering a class of observables that are block-diagonal with respect to the projections onto these almost invariant subspaces, we reduce the study of $Q(t)$ to that of a family of block-diagonal Heisenberg observables for each of them we construct a formal asymptotic expansion in powers of $\h$ by solving the corresponding symbolic Heisenberg problem. This reduction is modulo $\mathcal{O}(\h^{\infty})$ in $\mathcal{L}(L^2(\mathbb R^n)\otimes \mathbb C^m)$ uniformly in time in large time intervals which cover the Ehrenfest time. Then, to justify this asymptotic expansion, we control the remainder term at any order by giving a uniform exponential estimate with linear argument in time. In particular, this estimate allows us to recover the Ehrenfest time for the validity of the semiclassical approximation. %We point out that the idea of using almost invariant subspaces in the study of spectral problems in the semiclassical limit associated to matrix-valued Hamiltonian is well known and appear many times in the literature in different context e.g. \cite{Hel}, \cite{We}, \cite{Sor}, \cite{Brum}, \cite{bolte}. We refer to \cite{Sor} for a detailed description of this approach.

Another difficulty related to the matrix structure of $H$ and $Q$ lies in the fact that the time evolution of the symbols of the constructed approximation will be governed not only by the Hamiltonian flows (generated by the eigenvalues of $H_0$), but also by a conjugation by a family of transport matrices that we have to control the behaviour of theirs derivatives uniformly in time.

%In this paper, we study the time evolution of a bounded quantum observable $Q^w(x,\h D_x)$ associated to a ($m\times m$) matrix-valued symbol $Q$ generated by a semiclassical $(m \times m)$ hermitian-valued Hamiltonian $H(x,\xi;\h)= H_0(x,\xi)+ \h H_1(x,\xi)+\h^2 S(1;\mathbb R^{2n},M_m(\mathbb C))$ (the class of matrix-valued functions bounded together with all theirs derivatives on $\mathbb R^{2n}$ uniformly with respect to $\h$). Under a "non-crossing" assumption on the eigenvalues of the principal symbol $H_0$ and for a class of observables $Q$ block-diagonal with respect to the eigenprojectors of $H_0$, we construct a semiclassical approximation for the Heisenberg observable $Q(t):=e^{\frac{it}{\h}H^w} Q^w(x,\h D_x) e^{-\frac{it}{\h}H^w}$ in terms of $\h$-pseudodifferential operators and we recover the Ehrenfest time. 
The paper is organised as follows. In section \ref{sec}, after introducing the classes of symbols that we shall use through the paper, we state our main results. Section \ref{Sc} will be devoted to the study of the particular case of Hamiltonian with scalar principal symbol (and matrix-valued sub-principal symbol). In section \ref{Generalization}, we generalize the results of section \ref{Sc} to the case of matrix-valued principal symbol without crossing eigenvalues. The appendix \ref{semi-classical background} contains a short background on some basic results of $\h$-pseudodifferential calculus in the context of operators with matrix-valued symbols. In appendices \ref{Ann Cauchy problem general} and \ref{projjj}, we give the proofs of some technical results.

%In the appendix \ref{Ann Cauchy problem general}, we solve the general Cauchy problem used in the construction of the formal asymptotic expansions in powers of $\h$ of $Q(t)$. Finally, for the convenience of the reader, we give in the appendix \ref{projjj} a proof of the construction of the semiclassical projections used in section \ref{Generalization}. 
%To our best knowledge, there is not many results concerning the study of the time evolution of matrix-valued observables in the Heisenberg representation. In \cite[Chapter 2]{iv}, Ivrii considers the case where $H_0$ is a scalar multiple of the identity (i.e. $H_0(x,\xi)=\lambda(x,\xi)I_m$, with $\lambda$ a scalar real-valued symbol) and restrict itself to observable $Q$ compactly supported in small balls in the phase space of size $\varepsilon$ satisfying the "microlocal uncertainly principle" ($\varepsilon>\h^{\frac{1-\delta}{2}}$, $\delta>0$). In \cite{bolte}, Bolte and Glaser studied the Heisenberg evolution problem in the matricial case and they constructed a semiclassical approximation valid for finite times. Many ideas of \cite{bolte} are still present here.

\textbf{Some notations :} Let $M_m(\mathbb C)$ be the space of ($m\times m$) complex-valued matrices endowed with the operator norm denoted by $\Vert \cdot \Vert$. We denote $I_m$ the corresponding identity matrix. 

In this paper, three types of commutators appear : for $P,Q$ two matrix-valued functions in some suitable classes of symbols, $[P,Q]:=P Q-Q P$ is the usual matrix commutator. We use the same notation for the standard operators commutator $[P^w,Q^w]:= P^w Q^w-Q^wP^w$. Finally, the symbol of $[P^w,Q^w]$ will be denoted $[P,Q]_{\#}:=P\#Q-Q\#P$ and called the Moyal commutator of $P$ and $Q$.
 
Through the paper smooth means $C^{\infty}$. For $A\in C^{\infty}(\mathbb R^{2n})\otimes M_m(\mathbb C)$ and $\alpha,\beta \in \mathbb N^n$, we introduce the notation
\begin{equation*}
A_{(\beta)}^{(\alpha)}(x,\xi):=\partial_{\xi}^{\beta} \partial_x ^{\alpha} A(x,\xi).
\end{equation*}
Given a function $f_{\h}$ depending on the semiclassical parameter $\h\in(0,1]$, the asymptotic relation $f_{\h}=\mathcal{O}(\h^{\infty})$ means that $f_{\h}=\mathcal{O}(\h^N)$, for all $N\in \mathbb N$.

The identity operator on $L^2(\mathbb R^n)\otimes \mathbb C^m$ will be denoted $\text{id}_{L^2(\mathbb R^n)\otimes \mathbb C^m)}$ . For $\zeta=(\zeta_1,\cdots,\zeta_{2n})\in \mathbb R^{2n}$, we use the standard notation $\langle \zeta \rangle:= (1+|\zeta|^2)^{\frac{1}{2}}=(1+|\zeta_1|^2+\cdots+|\zeta_{2n}|^2)^{\frac{1}{2}}$. Finally, our convention for the Poisson bracket of matrix-valued functions $A,B\in C^{\infty}(\mathbb R^{2n})\otimes M_m(\mathbb C)$ is 
$$
\{A,B\}:= \partial_{\xi}A \partial_x B- \partial_x A \partial_{\xi}B.
$$
Notice that in general $\{A,B\}\neq -\{B,A\}$.

\section{Assumptions and main results}\label{sec}

%Let us start by recalling some notions concerning the classes of symbols that we shall use through the paper. We refer the reader to \cite{} where these notions 

%These notions are known in the case of scalar-valued symbol

Let us begin by recalling some notions about semiclassical classes of symbols. We refer to \cite[ch. 7]{dim} and \cite[ch. 4]{Zwo} for more details. For the context of operators with matrix-valued symbols see \cite[ch. 1]{iv}.

In this paper we use the standard $\h$-Weyl quantization defined for $A \in \mathscr{S}(\mathbb R^{2n})\otimes M_m(\mathbb C)$ (the space of Schwartz functions on $\mathbb R^{2n}$ with values in $M_m(\mathbb C)$) by the formula 
\begin{equation}\label{Weyl}
A^w(x,\h D_x)u(x) := \frac{1}{(2\pi \h)^n} \int \int_{\mathbb R^{2n}} e^{\frac{i}{\h}\langle x-y,\xi\rangle} A\bigg(\frac{x+y}{2},\xi\bigg) u(y)dyd\xi, \quad u\in \mathscr{S}(\mathbb R^{n})\otimes \mathbb C^m.
\end{equation}
For short, we shall sometimes simply write $A^w$.
 
Let $g : \mathbb R^{2n}\rightarrow [1,+\infty[$ be an order function, i.e. $g$ satisfies : there exists $C,N>0$ such that 
\begin{equation*}
g(v)\leq C \langle v-w \rangle^N g(w), \quad \forall v,w \in \mathbb R^{2n}.
\end{equation*}
The typical example is $g(x,\xi)=\langle (x,\xi)\rangle^a$, $a\geq 0$.

\begin{definition}
\begin{itemize}
\item[(i)] We denote by $S(g;\mathbb R^{2n},M_m(\mathbb C))$ the set of smooth functions on $\mathbb R^{2n}$ with values in $M_m(\mathbb C)$ satisfying : for all multi-index $\gamma \in \mathbb N^{2n}$, there exists a constant $C_{\gamma}>0$ such that for all $(x,\xi)\in \mathbb R^{2n}$,
\begin{equation}\label{seminorms}
\Vert \partial_{(x,\xi)}^{\gamma} A(x,\xi)\Vert \leq C_{\gamma} g(x,\xi).
\end{equation}
We will write it simply $S(g)$ when no confusion can arise. Symbols in $S(g)$ may depend on the semiclassical parameter $\h\in (0,1]$. In this case, we say that $A\in S(g)$ if $A(.,.;\h)$ is uniformly bounded in $S(g)$ when $\h$ varies in $(0,1]$.

For $r\in \mathbb R$, we define the classes
$$
S^r(g):= \h^{-r}S(g), \quad S^{-\infty}(g):= \bigcap_{r\in \mathbb R} S^r(g).
$$
\item[(ii)] $A\in S(g)$ is said to be elliptic if $A^{-1}(x,\xi)$ exists for all $(x,\xi)\in \mathbb R^{2n}$ and belongs to $S(g^{-1})$.
\item[(iii)] We say that $A$ admits an asymptotic expansion in powers of $\h$ in $S(g)$ if there exists $\h_0 \in ]0,1]$ and a sequence of $\h$-independent symbols $(A_j=A_j(x,\xi))_{j\in \mathbb N}\subset S(g)$ such that $A$ is a map from $]0,\h_0]$ into $S(g)$ satisfying 
\begin{equation}\label{definition du dev}
\h^{-(N+1)}\bigg(A(x,\xi;\h) - \sum_{j=0}^N \h^j A_j(x,\xi) \bigg) \in S(g), \quad \forall N\in \mathbb N.
\end{equation}

If (\ref{definition du dev}) holds, we write $A(x,\xi;\h)\sim \sum_{j\geq 0}\h^j A_j(x,\xi)$ in $S(g)$. $A_0$ is called the principal symbol and $A_1$ is called the sub-principal symbol of $A$.

Elements of $S(g)$ which admit an asymptotic expansion in powers of $\h$ are called semiclassical symbols and the corresponding Weyl operators via formula (\ref{Weyl}) will be denoted $A^w(x,\h D_x;\h)$ and called $\h$-pseudodifferential operators. We denote $S_{\text{sc}}(g)$ the set of semiclassical symbols in $S(g)$.
\item[(iv)] Let $P$ and $Q$ be two symbols in some suitable classes of symbols. The Moyal bracket of $P,Q$ denoted $\{P,Q\}^*$ is defined as the Weyl symbol of the operator $i\h^{-1}[P^w,Q^w]$. In the following, when $\{P,Q\}^*$ admits an asymptotic expansion in powers of $\h$, the coefficient of $\h^j$ will be denoted $\{P,Q\}^*_j$. 

The notion of the Moyal bracket will play an important role in this paper. We refer to the appendix \ref{semi-classical background} for more details.
\end{itemize}
\end{definition}

Let $H(x,\xi;\h) \sim \sum_{j\geq 0} \h^j H_j(x,\xi)$ in $S(g)$ be a $(m\times m)$ semiclassical Hamiltonian. To simplify the presentation and without any loss of generality, we suppose that $H(x,\xi;\h)= H_0(x,\xi) + \h H_1(x,\xi)$. We assume that 

\textbf{(A0).} $H_0$ and $H_1$ are hermitian-valued and $(H_0+i)$ is elliptic, i.e. there exists a constant $C>0$ such that 
$$
\Vert H_0(x,\xi) + i \Vert \geq C g(x,\xi),  \quad \forall (x,\xi)\in \mathbb R^{2n}.
$$

Under this assumption, $H^w(x,\h D_x;\h)$ is essentially self-adjoint in $L^{2}(\mathbb R^n)\otimes \mathbb C^m$ for $\h$ small enough (see \cite[Proposition 8.5]{dim} for the case $m=1$). By Stone's theorem (see e.g. \cite[p. 74]{Nelson}), the corresponding Schr\"odinger equation 
$$
i \h \partial_t u(t) =  H^w(x,\h D_x;\h)u(t)
$$
generates a one parameter group of unitary operators $U_H(t):= e^{-\frac{it}{\h}H^w}$ defined for all $t\in \mathbb R$.

%A solution of (\ref{schro equation}) can be obtained by defining $u(t):=U_{H}(t)u_0$, where $U_H(t)$ is the one-parameter group of unitary operators given by $U_H(t):=e^{-\frac{it}{\h}H^w}$, $t\in \mathbb R$.

%Let $\lambda_1(x,\xi)\leq \lambda_2(x,\xi)\leq ...\leq \lambda_m(x,\xi)$ be the (real) eigenvalues of $H_0(x,\xi)$. We assume that 

%\textbf{(A1).} There exists $l\in \{1,...,m\}$ and $r_1,...,r_l\in \mathbb N^*$ with $r_1+...+r_l=m$ such that for all $(x,\xi)\in \mathbb R^{2n}$, $H_0(x,\xi)$ admits exactly $l$ distinct eigenvalues $\lambda_1(x,\xi)<...<\lambda_l(x,\xi)$ of constant multiplicities on $\mathbb R^{2n}$ given by $r_1,...,r_l$ respectively.

%For $\nu \in \{1,...,l\}$ and $(x,\xi)\in \mathbb R^{2n}$, we denote by $P_{\nu,0}(x,\xi)$ the eigenprojector of $H_0(x,\xi)$ associated to the eigenvalue $\lambda_\nu(x,\xi)$. According to \textbf{(A1)}, $(x,\xi)\mapsto \lambda_{\nu}(x,\xi)$ and $(x,\xi)\mapsto P_{\nu,0}(x,\xi)$ are smooth functions on $\mathbb R^{2n}$, for all $1\leq \nu \leq l$.

Let $Q(x,\xi;\h) \sim \sum_{j\geq 0}\h^j Q_j(x,\xi)$ in $S(1)$ be a $(m\times m)$ semiclassical observable and we consider the time evolution of $Q^w(x,\h D_x;\h)$ in the Heisenberg picture given by
$$
Q(t):= U_{H}(-t) Q^w(x,\h D_x;\h) U_H(t), \quad  t\in \mathbb R.
$$
By the Calder\'on-Vaillancourt theorem (Theorem \ref{Cal}), $Q^w(x,\h D_x;\h)$ is bounded on $L^2(\mathbb R^n)\otimes \mathbb C^m$ and then $Q(t)$ is uniformly bounded on $L^2(\mathbb R^n)\otimes \mathbb C^m$ with respect to $t\in \mathbb R$. Moreover, $Q(t)$ satisfies the following Heisenberg equation of motion 
\begin{equation}\label{Heisenberg1}
\frac{d}{dt} Q(t)= \frac{i}{\h} [H^w,Q(t)], \quad Q(t)_{|t=0}=Q^w(x,\h D_x;\h).
\end{equation}
As indicated in the introduction, the first step in the semiclassical approximation of $Q(t)$ consists in the construction of a formal asymptotic expansion in powers of $\h$ for $Q(t)$ by solving the following Cauchy problems arising from (\ref{Heisenberg1}) if one assumes that $Q(t)$ admits a Weyl symbol $q(t)\sim \sum_{j\geq 0}\h^j q_j(t)$
\begin{equation} (\mathcal{C}_j) \label{Cau inks}
\left\{   
\begin{array}{rcl}
\displaystyle{\frac{d}{dt}} q_j(t)&=&  \{H,q(t)\}^*_j, \\
q_j(t)_{|t=0}&=& Q_j.
\end{array} 
\right.
\end{equation}

For $j=0$, according to (\ref{leading semiclassical order}) it is necessary to ensure the following  commutativity property
\begin{equation}\label{commutativity}
[H_0,q_0(t)]=0, \quad \forall t\in \mathbb R.
\end{equation}
For $t=0$, since $q_0(t)_{|t=0}=Q_0$, (\ref{commutativity}) is equivalent to a block-diagonal form of $Q_0$ with respect to the eigenprojectors of $H_0$. However, if one restricts to such observable, nothing ensures that this block-diagonal form will be respected by the time evolution. 

%In the following, we shall study two cases : 

\subsection{Hamiltonian with scalar principal symbol}
We begin with a particular but an important case where the principal symbol $H_0$ is a scalar multiple of the identity, that is 

\textbf{(A1').} $H_0(x,\xi)=\lambda(x,\xi)I_m$, for a scalar real-valued symbol $\lambda$.

This case allows us to understand the contribution of the sub-principal symbol $H_1$ in the time evolution. It will be clear from Theorem \ref{main T1} below that this case is different from the scalar one studied by Bouzouina-Robert \cite{Rob}. In particular, this case cannot be deduced from the results of \cite{Rob}.

We assume that

%We require the following growth assumptions on $H_0$ and $H_1$

\textbf{(A2').} For all $\gamma\in \mathbb N^{2n}$ and $j\in \{0,1\}$,
$$
\partial_{(x,\xi)}^{\gamma} H_j \in L^{\infty}(\mathbb R^{2n}), \quad \text{for}\;\;|\gamma|+j\geq 2.
$$ 
%Notice that assumption \textbf{(A2')} is the same as the one required in the scalar case.
Let $\phi_{\lambda}^t$ be the Hamiltonian flow generated by $\lambda$. Under the above assumption, the correspondant vector field $\mathcal{X}_{\lambda}:=(\partial_{\xi}\lambda,-\partial_{x}\lambda)$ grows at most linearly at infinity. Therefore a trajectory $\phi^t_{\lambda}(x,\xi)$ cannot blow up at finite times so that, for all $(x,\xi)\in \mathbb R^{2n}$, $\phi_{\lambda}^t(x,\xi)$ exists for all $t\in \mathbb R$. 

%In paragraph \ref{Sc}, using assumption \textbf{(A1')}, we solve the Cauchy problems $(\mathcal{C}_j)_{j\geq 0}$. The constructed matrix-valued functions $(q_j(t,x,\xi))_{j\geq 0}$ are defined by forumlas \eqref{symbole principal multiple scalaire} and \eqref{symbole j multiple scalaire}. Since we are interested in the semiclassical approximation for $Q(t)$ up to times of Ehrenfest type, we give in Proposition \ref{u5} uniform (in time) estimates on the derivatives with respect to $(x,\xi)$ of the symbols $(q_j(t,x,\xi))_{j\geq 0}$. Then, using the method of Bouzouina and Robert \cite{Rob}, we establish an accurate uniform exponential estimate on the remainder term at any order in the asymptotic expansion of $Q(t)$ (Theorem \ref{main T1}). This allows us to recover the Ehrenfest time for the validity of the semiclassical approximation (Corollary \ref{first coro}). In order to state our estimate on the remainder term at leading semiclassical order, we introduce the upper bound 

Put
\begin{equation}\label{upper bound Gamma}
\Gamma:= \big\Vert J \nabla_{(x,\xi)}^{(2)}\lambda(x,\xi) \big\Vert_{L^{\infty}(\mathbb R^{2n})},
\end{equation}
where $\nabla_{(x,\xi)}^{(2)}\lambda$ is the Hessian matrix of $\lambda$ and $J$ is the $(2n\times 2n)$ matrix associated to the canonical symplectic form on $\mathbb R^{2n}$ (see (\ref{symplectic form})). 

\begin{theorem}\label{main T1}
Assume \textbf{(A0)}, \textbf{(A1')} and \textbf{(A2')}, and let $Q\in S_{\text{sc}}(1)$. There exists a sequence of $(m\times m)$ matrix-valued $\h$-pseudodifferential operators $\left(( q_j(t) )^w(x,\h D_x)\right)_{j\geq 0}$ such that for all $N\in \mathbb N$, there exists $C_N>0$ such that for all $t\in \mathbb R$, the following estimate holds
\begin{equation}\label{remainder estimate T1}
\bigg\Vert Q(t)-\sum_{j=0}^N \h^j \big(q_j(t)\big)^w(x,\h D_x)  \bigg\Vert_{\mathcal{L}(L^2(\mathbb R^{n})\otimes \mathbb C^m)} \leq C_{N} \h^{N+1} \exp\bigg( \big( 4N+\delta_n\big)\Gamma |t| \bigg),
\end{equation}
where $\delta_n$ is an integer depending only on the dimension $n$. The symbols $q_j(t)$, $j\geq 0$, are defined by formula (\ref{symbole j multiple scalaire}) and satisfy estimates \eqref{As0} and \eqref{Asj}. In particular, the principal symbol $q_0(t)$ is given by 
$$
q_0(t,x,\xi)=T^{-1}(t,x,\xi)Q_0\big(\phi_{\lambda}^t(x,\xi)\big) T(t,x,\xi), \quad t\in \mathbb R, (x,\xi)\in \mathbb R^{2n},
$$
where $T$ is the unitary $(m\times m)$ matrix-valued function solution of the system 
\begin{equation}
\frac{d}{dt} T(t,x,\xi)= -i  H_1\big(\phi_{\lambda}^t(x,\xi)\big)T(t,x,\xi) , \quad T(0,x,\xi)=I_m.
\end{equation}
\end{theorem}

As a consequence of estimate \eqref{remainder estimate T1}, we get the following corollary about the Ehrenfest time for the validity of the semiclassical approximation.

\begin{corollary}\label{first coro}
Under the assumptions of Theorem \ref{main T1}, for all $N \geq 1$, there exists $C_N>0$ such that for every $\varepsilon>0$, we have 
\begin{equation}
\bigg\Vert Q(t)-\sum_{j=0}^N \h^j \big(q_{j}(t)\big)^w(x,\h D_x)  \bigg\Vert_{\mathcal{L}(L^2(\mathbb R^{n})\otimes \mathbb C^m)}\leq C_N \h^{\varepsilon N+1} \h^{\frac{(\varepsilon-1)}{4}\delta_n} ,
\end{equation}
uniformly for $|t|\leq \displaystyle{\frac{(1-\varepsilon)}{4\Gamma}\log(\h^{-1})}$. 
\end{corollary}

\begin{remark}\label{première remarque}
\begin{itemize}
\item[(i)] The upper bound $\Gamma$ is used to control the exponential growth of the flow $\phi_{\lambda}^t$ at infinity (see Lemma \ref{flot1}).
\item[(ii)] The constant $\delta_n$ is related to the universal constant in the Calder\'on-Vaillancourt Theorem (Theorem \ref{Cal}). See the end of the proof of Theorem \ref{main T1}.
\item[(iii)] Notice that for $m\geq 2$, our estimate on the remainder term \eqref{remainder estimate T1} is different from the one proved in the scalar case (see \cite[Theorem 1.4]{Rob}) where the argument in the exponential term was $2N + \delta'_n$ with $\delta'_n$ a universal constant. In particular, the constant $\displaystyle{\frac{1}{4\Gamma}}$ in the Ehrenfest time up to which the semiclassical approximation remains valid is half of the one proved in \cite{Rob}. This is due to the matrix structure of the sub-principal symbol $H_1$ (see Remark \ref{explanation} for more details). 
\end{itemize}
\end{remark}
%\textcolor{blue}{ \begin{remark}
%Here the $\mathcal{O}(\h)$ occurs from the fact that 
%$$
%\bigg\Vert Q(0) - \big(q_0(0)\big)^w(x,\h D_x) \bigg\Vert_{\mathcal{L}(L^2(\mathbb R^n)\otimes \mathbb C^m)} = \bigg\Vert Q^w(x,\h D_x;\h) - Q_0^w(x,\h D_x) \bigg\Vert_{\mathcal{L}(L^2(\mathbb R^n)\otimes \mathbb C^m)}= \mathcal{O}(\h).
%$$
%So in the case where the observable $Q$ is independnt on $\h$, i.e. $Q=Q(x,\xi)$, we have not this extra term. The same remark is true for the $\mathcal{O}(\h^{N+1})$ appearing in \eqref{remainder estimate T1}.
%\end{remark}}

%\textcolor{blue}{The main idea is to reduce the study of $Q(t)$ to that of a family of block diagonal Heisenberg obsevables for which the block diagonal form of the initial obsevable is preserved under the time evolution. Moreover, this reduction must be uniform up to times which at least cover Ehrenfest type times.}

%A particular case for which this property is always satisfied is when $H_0$ is a scalar multiple of the identity in $M_m(\mathbb C)$. More precisely

%\textbf{(A1).} $H_0(x,\xi)=\lambda(x,\xi)I_m$, for a scalar real-valued symbol $\lambda$.

\subsection{General case}
Now we drop the assumption \textbf{(A1')}. We assume that

{\textbf{(A1).} There exists $l\in \{1,...,m\}$ and $r_1,...,r_l\in \mathbb N^*$ with $r_1+\cdots+r_l=m$ such that $H_0(x,\xi)$ admits exactly $l$ distinct eigenvalues $\lambda_1(x,\xi)<\cdots<\lambda_l(x,\xi)$ with constant multiplicities on $\mathbb R^{2n}$ given by $r_1,...,r_l$ respectively, satisfying : there exists a constant $\rho >0$ such that for all $1\leq \mu \neq \nu \leq l$,
\begin{equation}\label{Gap}
|\lambda_{\mu}(x,\xi) - \lambda_{\nu}(x,\xi)| \geq \rho g(x,\xi), \quad \text{for}\;\; |x|+|\xi| \geq c>0.
\end{equation}

\textbf{(A2).} For all $\gamma\in \mathbb N^{2n}$ and $j\in \{0,1\}$,
$$
\partial_{(x,\xi)}^{\gamma} H_j \in L^{\infty}(\mathbb R^{2n}), \quad \text{for}\;\;|\gamma|+j\geq 1.
$$ 

For $\nu \in \{1,...,l\}$, let $P_{\nu,0}(x,\xi)$ be the eigenprojector associated to the eigenvalue $\lambda_{\nu}(x,\xi)$. The assumption \textbf{(A1)} ensures that the functions $(x,\xi)\mapsto \lambda_{\nu}(x,\xi)$ and $(x,\xi)\mapsto P_{\nu,0}(x,\xi)$ are smooth in $\mathbb R^{2n}$. Moreover, in Lemma \ref{smoothness}, we show that $P_{\nu,0}\in S(1)$ and $\lambda_{\nu}\in S(g)$, for all $1\leq \nu \leq l$.

As in  \cite{bolte}  see also \cite{Hel, Sor} and  Theorem \ref{projections}, 
we construct  $l$ $\h$-pseudodifferential operators $P_1^w(x,\h D_x;\h)$, ..., $P_l^w(x,\h D_x;\h)$ satisfying  
\begin{equation*}
\big(P_{\nu}^w\big)^2 = \big(P_{\nu}^w\big)^* = P_{\nu}^w,
\end{equation*}
and 
\begin{equation*}\label{properties}
[H^w,P_{\nu}^w] = 0, \quad \sum_{\nu=1}^l P_{\nu}^w = \text{id}_{L^2(\mathbb R^n)\otimes \mathbb C^m}, \quad P_{\nu}^w P_{\mu}^w = 0 , \quad \forall 1 \leq \nu \neq \mu \leq l,
\end{equation*}
modulo $\mathcal{O}(\h^{\infty})$ in norm $\mathcal{L}(L^2(\mathbb R^n)\otimes \mathbb C^m)$. For $\nu\in \{1,...,l\}$, the principal symbol of $P_{\nu}^w(x,\h D_x;\h)$ coincides with the eigenprojector $P_{\nu,0}$. The operators $\big(P_{\nu}^w(x,\h D_x;\h)\big)_{1\leq \nu \leq l}$ are called semiclassical projections associated to $H^w(x,\h D_x;\h)$.

%In view of the assumptions needed on the observable $Q$, it is more convenient in this case to consider a semiclassical observable $Q\in S_{\text{sc}}(1)$ with an asymptotic expansion in powers of $\h$, $Q(x,\xi;\h)\sim \sum_{j\geq 0}\h^j Q_j(x,\xi)$.
As indicated in \cite[Proposition 3.2]{bolte} (see also Remark \ref{la classe Q de 1}), to construct a complete asympotic expansion in powers of $\h$ for $Q(t)$, some restrictions on the initial observable $Q$ are necessary. We introduce the class $\mathcal{Q}(1)$ of observables $Q\in S_{\text{sc}}(1)$ that are "semiclassically" block-diagonal with respect to the semiclassical projections $P_{\nu}$, $1\leq \nu \leq l$, i.e.
$$
\mathcal{Q}(1) := \left\{ Q \in S_{\text{sc}}(1) \: ;\; Q \sim \sum_{\nu=1}^l P_{\nu} \# Q \# P_{\nu}\;\; \text{in} \;\; S(1) \right\}.
$$
In particular, using formula \eqref{deduc2}, one sees that if $Q\in \mathcal{Q}(1)$ then $Q_0$ is block diagonal with respect to the eigenprojectors $P_{\nu,0}$, i.e. 
$$
Q_0(x,\xi) = \sum_{\nu=1}^l P_{\nu,0}(x,\xi) Q_0(x,\xi) P_{\nu,0}(x,\xi), \quad \forall (x,\xi)\in \mathbb R^{2n}.
$$
Let $\phi_{\nu}^t$ be the Hamiltonian flow generated by the eigenvalue $\lambda_{\nu}$. The assumption \textbf{(A2)} ensures that $\phi_{\nu}^t(x,\xi)$ exists globally on $\mathbb R$, for all $(x,\xi)\in \mathbb R^{2n}$, $1\leq \nu \leq l$.

Put
\begin{equation}\label{upper bounds}
\Gamma_{\nu}:= \Vert J \nabla_{(x,\xi)}^{(2)} \lambda_{\nu}(x,\xi) \Vert_{L^{\infty}(\mathbb R^{2n})}, \quad \Gamma_{\text{max}} := \max_{1\leq \nu \leq l} \Gamma_{\nu} ,
\end{equation}
where $\nabla_{(x,\xi)}^{(2)}\lambda_{\nu}$ denotes the Hessian matrix of $\lambda_{\nu}$, $1\leq \nu \leq l$.

Our main result of this paper is the following 
\begin{theorem}\label{main T}
Assume \textbf{(A0-2)} and let $Q\in \mathcal{Q}(1)$. There exists a sequence $\big( (q_j(t))^w(x,\h D_x) \big)_{j\geq 0}$ of $(m\times m)$ matrix-valued $\h$-pseudodifferential operators such that for all $N\in \mathbb N$, there exists $C_{N}>0$ such that for all $t\in \mathbb R$, the following estimate holds
\begin{align}\label{semmm}
\bigg\Vert Q(t)-\sum_{j=0}^N \h^j \big(q_{j}(t)\big)^w(x,\h D_x)  \bigg\Vert_{\mathcal{L}(L^2(\mathbb R^{n})\otimes \mathbb C^m)} & \leq C_{N} \h^{N+1} \exp\bigg( (4N+\tilde{\delta}_n)\Gamma_{\text{max}} |t| \bigg) ,
\end{align}
where $\tilde{\delta}_n$ is an integer depending only on the dimension $n$. The symbols $q_{j}(t,x,\xi)$ are defined for $t\in \mathbb R$ and $(x,\xi)\in \mathbb R^{2n}$ by 
$$
q_j(t,x,\xi) := \sum_{\nu=1}^l q_{\nu,j}(t,x,\xi), \quad j\geq 0,
$$
where $q_{\nu,j}(t)$ are given by the general formula (\ref{general formula for the solution}) and satisfy estimates \eqref{AG0} and \eqref{AGj}. In particular, the principal symbol $q_0(t)$ is given by
\begin{equation}\label{symbole principal th}
q_0(t,x,\xi) = \sum_{\nu=1}^l T_{\nu}^{-1}(t,x,\xi) \big(P_{\nu,0} Q_0 P_{\nu,0}\big)\big( \phi_{\nu}^t(x,\xi) \big) T_{\nu}(t,x,\xi),
\end{equation}
where $T_{\nu}$ is the unitary $(m\times m)$ matrix-valued function solution of the system

\begin{equation}
\frac{d}{dt} T_{\nu}(t,x,\xi) = -i \tilde{H}_{\nu,1}\big( \phi_{\nu}^t(x,\xi) \big) T_{\nu}(t,x,\xi) \quad T_{\nu}(0,x,\xi) = I_m.
\end{equation}
Here $\tilde{H}_{\nu,1}$ is the $(m\times m)$ hermitian-valued function defined by 
\begin{equation}
\tilde{H}_{\nu,1}= \frac{1}{2i} P_{\nu,0} \big\{P_{\nu,0},H_0 \big\}P_{\nu,0}  - i \big[ P_{\nu,0},\{\lambda_{\nu},P_{\nu,0}\} \big] + P_{\nu,0}H_1 P_{\nu,0}.
\end{equation}
\end{theorem}

As a consequence we get the following corollary.
\begin{corollary}\label{second coro}
Under the assumptions of Theorem \ref{main T}, for all $N\geq 1$ there exists $C_N>0$ such that for all $\varepsilon>0$, we have 
\begin{equation}\label{uniform limit2}
\bigg\Vert Q(t)-\sum_{j=0}^N \h^j (q_j(t))^w(x,\h D_x) \bigg\Vert_{\mathcal{L}(L^2(\mathbb R^{n})\otimes \mathbb C^m)} \leq C_N \h^{\varepsilon N+1} \h^{\frac{(\varepsilon-1)}{4} \tilde{\delta}_n},
\end{equation}
uniformly for $|t|\leq \displaystyle{\frac{(1-\varepsilon)}{4 \Gamma_{\text{max}}}} \log( \h^{-1})$.
\end{corollary}

If we only look for the principal symbol of $Q(t)$, the assumption $Q\in \mathcal{Q}(1)$ can be relaxed and we have the following result.

\begin{corollary}\label{sec coro}
Let $H$ be a semiclassical Hamiltonian satisfying the assumptions of Theorem \ref{main T} and let $Q\in S_{\text{sc}}(1)$. We assume that $Q_0(x,\xi) = \sum_{\nu = 1}^l P_{\nu,0}(x,\xi) \tilde{Q}(x,\xi) P_{\nu,0}(x,\xi)$ for some $\tilde{Q}\in S(1)$. There exists $C>0$ such that for all $t\in \mathbb R$, the following estimate holds
\begin{align*}
\bigg\Vert Q(t)- \big( q_{0}(t) \big)^w(x,\h D_x)  \bigg\Vert_{\mathcal{L}(L^2(\mathbb R^{n})\otimes \mathbb C^m)} & \leq C \h \exp\big( \tilde{\delta}_n\Gamma_{\text{max}} |t| \big),
\end{align*} 
where $q_0(t)$ is given by \eqref{symbole principal th}.
\end{corollary}

\begin{remark}\label{la classe Q de 1}
In \cite[Proposition 3.4]{bolte}, it was shown that the class $\mathcal{Q}(1)$ exhausts all symbols $Q\in S_{\text{sc}}(1)$ such that the corresponding Heisenberg observable $Q(t)$ is an $\h$-pseudodifferential operator with symbol $q(t)\in S_{\text{sc}}(1)$, for all finite time $t$. More explicitly, if $H$ satisfies the assumptions of Theorem \ref{main T}, then we have 
$$
\bigg\{ Q\in S_{\text{sc}}(1) \: ; \;\; \forall |t|\leq \overline{t}< \infty, \; Q(t) =  \big(q(t)\big)^w(x,\h D_x;\h), \;\; \text{with}\; q(t)\in S_{\text{sc}}(1)  \bigg\} = \mathcal{Q}(1).
$$
\end{remark}

\section{Hamiltonian with scalar principal symbol}\label{Sc}

In this section, we study the particular case where the principal symbol $H_0$ is a scalar multiple of the identity in $M_m(\mathbb C)$. The proof of Theorem \ref{main T1} relies essentially on the following steps. In the next paragraph, using assumption \textbf{(A1')}, we construct a formal asymptotic expansion in powers of $\h$ for $Q(t)$ by solving the Cauchy problems $(\mathcal{C}_j)_{j\geq 0}$ (see \eqref{Cau inks}). The constructed matrix-valued functions $(q_j(t,x,\xi))_{j\geq 0}$ are defined by forumla \eqref{symbole j multiple scalaire}. Since we are interested in the semiclassical approximation for $Q(t)$ up to times of Ehrenfest type, we give in Proposition \ref{u5} uniform (in time) estimates on the derivatives with respect to $(x,\xi)$ of the symbols $(q_j(t,x,\xi))_{j\geq 0}$. Then, using these estimates, we prove \eqref{remainder estimate T1} by following the method of Bouzouina-Robert \cite{Rob}.

\subsection{Formal asymptotic expansion}\label{formal Sc}
Let $H(x,\xi;\h) = H_0(x,\xi) + \h H_1(x,\xi)$ be a semiclassical Hamiltonian and suppose that $H_0$ satisfies \textbf{(A1')}. According to this assumption, the principal symbol of $[H,q(t)]_{\#}$ given by the commutator $[H_0,q_0(t)]$ vanishes for all $t\in \mathbb R$. Consequently, using the rule of asymptotic expansion of the product of symbols (formula (\ref{développement asymptotique})), the symbol $\{H,q(t)\}^*$ can be expended in a power series of $\h$ (see formula (\ref{rMp})) and then the Cauchy problems $(\mathcal{C}_j)_{j\geq 0}$ become
\begin{align}(\mathcal{C}_j)\label{Cj}
\left\{   
\begin{array}{rcl}
\displaystyle{\frac{d}{dt}q_j(t)} &=& \sum_{|\alpha|+|\beta|+k+p=j+1} \tilde{\gamma}(\alpha,\beta) \bigg( {H_k}_{(\alpha)}^{(\beta)}{q_p(t)}_{(\beta)}^{(\alpha)}-(-1)^{|\alpha|-|\beta|}{q_p(t)}_{(\beta)}^{(\alpha)}{H_k}_{(\alpha)}^{(\beta)} \bigg)\\
\\
{q_j(t)}_{|t=0}&=& Q_j, 
\end{array} 
\right.
\end{align}
with $\tilde{\gamma}(\alpha,\beta):= \displaystyle{\frac{i(-1)^{| \beta |}}{(2i)^{| \alpha | + | \beta |}\fact{\alpha}\fact{\beta}} }$.

Thanks to assumption \textbf{(A1')} again, for $p=j+1$, the right hand side of \eqref{Cj} is equal to $i[H_0,q_{j+1}(t)]$ which vanishes for all $t\in \mathbb R$. Then, $(\mathcal{C}_j)$ can be rewritten in the following form
\begin{align}
\displaystyle{\frac{d}{dt}q_j(t)} &=\sum\limits_{\underset{0\leq p \leq j}{|\alpha|+|\beta|+k+p=j+1}} \tilde{\gamma}(\alpha,\beta) \bigg( {H_k}_{(\alpha)}^{(\beta)}{q_p(t)}_{(\beta)}^{(\alpha)}-(-1)^{|\alpha|-|\beta|}{q_p(t)}_{(\beta)}^{(\alpha)}{H_k}_{(\alpha)}^{(\beta)} \bigg) \nonumber \\
&= \{\lambda,q_j(t)\}+i[H_1,q_j(t)]+ \sum\limits_{\underset{0\leq p \leq j-1}{|\alpha|+|\beta|+k+p=j+1}} \tilde{\gamma}(\alpha,\beta) \bigg( {H_k}_{(\alpha)}^{(\beta)}{q_p(t)}_{(\beta)}^{(\alpha)}-(-1)^{|\alpha|-|\beta|}{q_p(t)}_{(\beta)}^{(\alpha)}{H_k}_{(\alpha)}^{(\beta)} \bigg). \label{rewritten}
\end{align}
For $j\geq 0$, we set 
\begin{equation}\label{expression de B1}
B_j(t,x,\xi):= \sum\limits_{\underset{0\leq p \leq j-1}{|\alpha|+|\beta|+k+p=j+1}} \tilde{\gamma}(\alpha,\beta) \bigg( {H_k}_{(\alpha)}^{(\beta)}{q_p(t)}_{(\beta)}^{(\alpha)}-(-1)^{|\alpha|-|\beta|}{q_p(t)}_{(\beta)}^{(\alpha)}{H_k}_{(\alpha)}^{(\beta)} \bigg)(x,\xi),
\end{equation}
with the convention $B_0=0$ since the sum is empty. Before giving the solution of (\ref{rewritten}), let us make the following remark concerning the case where the sub-principal symbol $H_1$ is also scalar-valued.
\begin{remark}\label{les symboles dans le cas scalaire} Suppose that $H_1$ is a scalar real-valued symbol. Then, $[H_1,q_j(t)]$ vanishes and one can easily verify that equation \eqref{rewritten} is equivalent to the following one 
$$
\frac{d}{dt} \bigg( q_j\big(t,\phi_{\lambda}^{-t}(x,\xi)\big) \bigg) = B_j\big(t,\phi_{\lambda}^{-t}(x,\xi)\big), \quad j\geq 0,
$$
where $B_j$ can be rewritten in the simpler form 
\begin{equation}\label{Les Bj}
B_j(t,x,\xi) = \sum\limits_{\underset{0\leq p \leq j-1}{|\alpha|+|\beta|+k+p=j+1}} \frac{i((-1)^{|\beta|}-(-1)^{|\alpha|})}{(2i)^{|\alpha|+|\beta|}\fact{\alpha}\fact{\beta}}  \left({H_k}_{(\alpha)}^{(\beta)}(x,\xi) {q_p(t)}_{(\beta)}^{(\alpha)}(x,\xi)\right).
\end{equation}
Consequently, the solutions $q_{j,\text{sca}}(t)$, $j\geq 0$, are given by
\begin{equation}\label{symboles dans le cas scalaire}
q_{j,\text{sca}}(t,x,\xi)= q_{j,\text{sca}}\big(0,\phi_{\lambda}^t(x,\xi)\big) + \int_0^t B_j\big(s,\phi_{\lambda}^{t-s}(x,\xi)\big) ds, \quad  t\in \mathbb R, (x,\xi)\in \mathbb R^{2n},
\end{equation}
where we introduced the index "sca" to precise that we are in the case where $H_0$ and $H_1$ are scalar-valued. In particular, %since $q_0(t)_{|t=0}=Q$ and $B_0=0$, the principal symbol is simply given by 
$$
q_{0,\text{sca}}(t,x,\xi)=Q_0 \circ \phi_{\lambda}^t(x,\xi) \quad \text{and} \quad q_{1,\text{sca}}(t,x,\xi) = Q_1\big(\phi_{\lambda}^t(x,\xi)\big) + \int_0^t \big\{H_1, Q_0(\phi^s_{\lambda}) \big\}\circ \phi_{\lambda}^{t-s}(x,\xi) \; ds.
$$
\end{remark}

\begin{flushright}
$\square$
\end{flushright}

Turn now to the resolution of \eqref{rewritten}. Applying the results of Appendix \ref{Ann Cauchy problem general} with $\Lambda=\lambda$, $A=H_1$ and $B(t)= B_j(t)$, we get the solution for all $j\geq 0$
\begin{equation}\label{symbole j multiple scalaire}
q_j(t,x,\xi)=T^{-1}(t,x,\xi)\bigg( Q_j\big(\phi_{\lambda}^t(x,\xi) \big) +  \int_0^t T^{-1}\big(-s,\phi_{\lambda}^t(x,\xi)\big) B_j\big(s,\phi_{\lambda}^{t-s}(x,\xi)\big) T\big(-s,\phi_{\lambda}^t(x,\xi)\big) \; ds \bigg) T(t,x,\xi), 
\end{equation}
defined for all $t\in \mathbb R$ and $(x,\xi)\in \mathbb R^{2n}$, where $T$ and $T^{-1}$ are the unitary $(m\times m)$ matrix-valued functions solutions of the following systems
\begin{equation}\label{transport}
\frac{d}{dt} T(t,x,\xi)= -i  H_1\big(\phi_{\lambda}^t(x,\xi)\big)T(t,x,\xi) , \quad T(0,x,\xi)=I_m,
\end{equation}
\begin{equation}\label{transport inverse}
\frac{d}{dt} T^{-1}(t,x,\xi)= i T^{-1}(t,x,\xi) H_1\big(\phi_{\lambda}^t(x,\xi)\big), \quad T^{-1}(0,x,\xi)=I_m.
\end{equation}
In particular, the principal symbol $q_0(t)$ is given by 
\begin{equation}\label{symbole principal multiple scalaire}
q_0(t,x,\xi)=T^{-1}(t,x,\xi)Q_0\big(\phi_{\lambda}^t(x,\xi)\big) T(t,x,\xi), \quad \forall t\in \mathbb R, (x,\xi)\in \mathbb R^{2n}.
\end{equation}
\subsection{Uniform estimates}\label{results} 
Let $\Gamma$ be the upper bound defined by \eqref{upper bound Gamma}.
\begin{proposition}\label{u5}
Assume \textbf{(A2')}. For all $\gamma \in \mathbb N^{2n}$ and all $j\geq 0$, there exists $C_{\gamma,j}>0$ such that for all $t\in \mathbb R$ and all $(x,\xi)\in \mathbb R^{2n}$, we have 
\begin{equation}\label{As0}
\big\Vert \partial_{(x,\xi)}^{\gamma} q_0(t,x,\xi) \big\Vert \leq C_{\gamma,0} \exp\bigg( |\gamma|\Gamma |t|\bigg),
\end{equation}
and for $j\geq 1$,
\begin{equation}\label{Asj}
\big\Vert \partial_{(x,\xi)}^{\gamma} q_j(t,x,\xi) \big\Vert \leq C_{\gamma,j} \exp\bigg( \big(2|\gamma|+4j-3 \big)\Gamma |t|\bigg).
\end{equation}
\end{proposition}

To prove this proposition we need to recall the multivariate Fa\'a Di Bruno formula used for computing arbitrary partial derivatives of a function composition. In the following, this formula will be used wherever we have to estimate the derivatives of observables moving along the Hamiltonian flow. In the literature, one can found several forms to this formula (see for instance \cite{Leip}, \cite{Const}). As in \cite{Rob}, we use the following one :
\begin{lemma}
Let $F = (F_{ij})_{1\leq i,j \leq m}: \mathbb R^{2n}\rightarrow M_m(\mathbb C)$ and $G = (G_1,...,G_{2n}) : \mathbb R^{2n}\rightarrow \mathbb R^{2n}$ be smooth functions. For all $\gamma \in \mathbb N^{2n}$, we have 
\begin{equation*}
\partial_{(x,\xi)}^{\gamma} \big(F\circ G \big) = \bigg(\partial_{(x,\xi)}^{\gamma} \big(F_{ij}\circ G \big) \bigg)_{1\leq i,j \leq m}
\end{equation*}
where
\begin{equation}\label{Faa de Bruno}
\partial_{(x,\xi)}^{\gamma} \big(F_{ij}\circ G \big) = \sum_{\substack{\beta\in \mathbb N^{2n}\\ 0 \neq \beta \leq \gamma}} \big( \partial_{(x,\xi)}^{\beta}F_{ij}\big)\circ G \;.\; \mathcal{A}_{\gamma,\beta}(G),
\end{equation}
with
\begin{equation}\label{expl. Faa}
\mathcal{A}_{\gamma,\beta}(G) = \fact{\gamma} \sum\limits_{\underset{\sum_{\alpha}\alpha|\eta|=\gamma}{\sum_{\alpha} \alpha= \beta}} \prod_{\alpha \in \mathbb N^{2n}\setminus \{0\} } \frac{1}{\fact{\eta}} \bigg( \frac{\partial_{(x,\xi)}^{\alpha}G_1}{\fact{\alpha}} \bigg)^{\eta_1}...\;\bigg( \frac{\partial_{(x,\xi)}^{\alpha}G_{2n}}{\fact{\alpha}} \bigg)^{\eta_{2n}}.
\end{equation}
Here we use the usual rules for multi-indices (see \cite{Const}).
%The sum in the above formula is over all multi-indices $\alpha^{(1)},...,\alpha^{(|\gamma|)}\in \mathbb N^{2n}$ with $\alpha^{(1)}+ \cdots + \alpha^{(|\gamma|)} = \beta$ and $\eta^{(1)},..., \eta^{(|\gamma|)}\in \mathbb N^{2n}$ with $|\eta^{(1)}| \alpha^{(1)} + \cdots + |\eta^{(|\gamma|)}|\alpha^{(|\gamma|)} = \gamma$.
\end{lemma}

The proof of Proposition \ref{u5} is based on the three following lemmas. The first one gives exponential estimate on the derivatives (with respect to $(x,\xi)$) of the Hamiltonian flow associated to $\lambda$. This result can be proved by induction on $|\gamma|$ using the Jacobi stability equation 
\begin{equation}\label{Jacobi}
\frac{d}{dt}\nabla_{(x,\xi)}\phi_{\lambda}^t(x,\xi)=J\nabla_{(x,\xi)}^{(2)}\lambda(\phi_{\lambda}^t(x,\xi))\nabla_{(x,\xi)}\phi_{\lambda}^t(x,\xi),
\end{equation} 

where $\nabla_{(x,\xi)}\phi_{\lambda}^t:=(\partial_x \phi_{\lambda}^t, \partial_{\xi}\phi_{\lambda}^t)$. 

The following lemma is proved in  \cite[Lemma 2.2]{Rob}.
\begin{lemma}\label{flot1}
We assume that 
$$ \partial_{(x,\xi)}^{\gamma}\lambda \in L^{\infty}(\mathbb R^{2n}), \quad \forall \gamma\in \mathbb N^{2n};|\gamma|\geq 2 .
$$ 
Then, for all $\gamma \in \mathbb N^{2n}\setminus \{0\}$, there exists $C_{\gamma}> 0$ such that for all $t\in \mathbb R$ and all $(x,\xi)\in \mathbb R^{2n}$,
\begin{equation}\label{floww1}
\big\Vert\partial_{(x,\xi)}^{\gamma}\phi_{\lambda}^t(x,\xi)\big\Vert\leq C_{\gamma} \exp\big(|\gamma|\Gamma|t|\big).
\end{equation}
\end{lemma}

In the next lemma, we prove similar estimate on the derivatives of the matrix-valued function $T$ (see (\ref{transport})).
\begin{lemma}\label{Transport1}
Assume \textbf{(A2')}. For all $\gamma \in \mathbb N^{2n}\setminus \{0\}$, there exists $C_{\gamma}>0$ (independent of $t\in \mathbb R$ and $(x,\xi)\in \mathbb R^{2n}$) such that
\begin{equation}\label{estimation de T1}
\big\Vert\partial_{(x,\xi)}^{\gamma}T(t,x,\xi)\big\Vert\leq C_{\gamma} \exp \big(|\gamma|\Gamma |t|\big).
\end{equation}
\end{lemma}
\begin{proof} 
Without any loss of generality, we assume that $t\geq 0$ (the proof for $t\leq 0$ is similar). We proceed by induction on $|\gamma|$. Let us check (\ref{estimation de T1}) for the first order derivative of $T$ with respect to $x_1$. A straightforward computation using equations (\ref{transport}) and (\ref{transport inverse}) yields 
\begin{eqnarray*}
\frac{d}{dt} \big( T^{-1}(t,x,\xi) \partial_{x_1} T(t,x,\xi) \big) &=& \partial_t T^{-1}(t,x,\xi) \partial_{x_1} T(t,x,\xi) + T^{-1}(t,x,\xi) \partial_t \partial_{x_1} T(t,x,\xi) \\
&=& - i T^{-1}(t,x,\xi) (\partial_{x_1} H_1)(\phi_{\lambda}^t(x,\xi))\partial_{x_1} \phi_{\lambda}^t(x,\xi)  T(t,x,\xi).
\end{eqnarray*}
Therefore
$$
T^{-1}(t,x,\xi) \partial_{x_1} T(t,x,\xi) = -i \int_0^t T^{-1}(s,x,\xi) (\partial_{x_1} H_1)(\phi_{\lambda}^s(x,\xi)) \partial_{x_1} \phi_{\lambda}^s(x,\xi) T(s,x,\xi) \; ds,
$$
since $\partial_{x_1}T(0,x,\xi)=0$ (we recall that $T(0,x,\xi)=I_m$). Taking into account the fact that $T^{-1}$ is unitary and using estimate (\ref{floww1}), we obtain 
\begin{equation*}
\big\Vert \partial_{x_1} T(t,x,\xi) \big\Vert \leq  \int_0^{t}  \big\Vert \partial_{x_1} \phi_{\lambda}^s(x,\xi) \big\Vert . \big\Vert  \partial_{x_1} H_1 \big\Vert_{L^{\infty}(\mathbb R^{2n})} \; ds \leq C \exp(\Gamma t),
\end{equation*}
uniformly for $t\geq 0$ and $(x,\xi)\in \mathbb R^{2n}$. This gives the proof for $\gamma=(1,0,...,0)$. The same proof holds for $|\gamma|=1$.

Let us now assume that (\ref{estimation de T1}) holds for all $\gamma \in \mathbb N^{2n}$ with $|\gamma|<r$, $r\geq 2$, and take $|\gamma|=r$. Computing derivatives with respect to $(x,\xi)$ in (\ref{transport}) using Leibniz formula, we get
$$
\frac{d}{dt} \partial_{(x,\xi)}^{\gamma} T(t,x,\xi) = -i H_{1}\big(\phi_{\lambda}^{t}(x,\xi)\big) \partial_{(x,\xi)}^{\gamma}T(t,x,\xi) -i \sum_{1\leq |\beta|\leq r} \binom{\beta}{\gamma}\partial_{(x,\xi)}^{\beta}\big( H_{1}(\phi_{\lambda}^{t}(x,\xi)) \big)\partial_{(x,\xi)}^{\gamma-\beta}T(t,x,\xi).
$$
Therefore
\begin{eqnarray*}
\frac{d}{dt} \big(  T^{-1}(t,x,\xi) \partial_{(x,\xi)}^{\gamma}T(t,x,\xi)\big) = -i T^{-1}(t,x,\xi) \sum_{1\leq |\beta|\leq r} \binom{\beta}{\gamma}\partial_{(x,\xi)}^{\beta}\big( H_{1}(\phi_{\lambda}^{t}(x,\xi)) \big)\partial_{(x,\xi)}^{\gamma-\beta}T(t,x,\xi).
\end{eqnarray*}
According to assumption \textbf{(A2')}, for all $\beta\in \mathbb N^{2n}$ with $|\beta|\geq 1$, we have $\partial_{(x,\xi)}^{\beta}H_1 \in L^{\infty}(\mathbb R^{2n})$. Consequently, using Fa\'a Di Bruno's formula (\ref{Faa de Bruno}) and estimate (\ref{floww1}), we obtain 
\begin{equation}\label{hij1}
\big\Vert \partial_{(x,\xi)}^{\beta}\big( H_{1} \circ \phi_{\lambda}^{t}(x,\xi) \big) \big\Vert \leq C_{\beta} \exp \big(|\beta|\Gamma t \big),
\end{equation}
uniformly with respect to $t\geq 0$ and $(x,\xi)\in \mathbb R^{2n}$.

On the other hand, by the induction hypothesis, there exists $C_{\gamma,\beta}>0$ such that for all $t\geq 0$ and all $(x,\xi)\in \mathbb R^{2n}$, we have 
\begin{equation}\label{hij2}
\big\Vert \partial_{(x,\xi)}^{\gamma-\beta}T(t,x,\xi) \big\Vert \leq C_{\gamma,\beta} \exp \big( (r-|\beta|)\Gamma t\big).
\end{equation}

Putting together (\ref{hij1}) and (\ref{hij2}) and taking into account the fact that $\partial_{(x,\xi)}^{\gamma} T(0,x,\xi)=0$, we get
\begin{eqnarray*}
\big\Vert \partial_{(x,\xi)}^{\gamma}T(t,x,\xi) \big\Vert &\leq &  \sum_{1\leq |\beta|\leq r} C_{\gamma,\beta} \int_0^{t}  \big\Vert \partial_{(x,\xi)}^{\beta}\big( H_{1}\big(\phi_{\lambda}^{s}(x,\xi)\big) \big)\big\Vert . \big\Vert \partial_{(x,\xi)}^{\gamma-\beta}T(s,x,\xi) \big\Vert \; ds \\
&\leq & \sum_{1\leq |\beta|\leq r} C'_{\gamma,\beta} \int_0^{t}  \exp\big(|\beta|\Gamma s\big) \exp\big((r-|\beta|)\Gamma s\big) \; ds \\
&\leq & C_{\gamma} \exp(r \Gamma t).
\end{eqnarray*}
Hence (\ref{estimation de T1}) holds for $|\gamma|=r$. This ends the proof.

Notice that the same proof can be repeated for $T^{-1}$ and then estimate (\ref{estimation de T1}) remains valid for the derivatives of $T^{-1}$.
\end{proof}
The following lemma is a consequence of the two previous lemmas and the Fa\'a Di Bruno formula (\ref{Faa de Bruno}).
\begin{lemma}\label{transport avec le flot}
Under assumption \textbf{(A2')}, for all $\gamma\in \mathbb N^{2n}\setminus \{0\}$, there exists $C_{\gamma}>0$ such that for all $(x,\xi)\in \mathbb R^{2n}$ and all $t,s\in \mathbb R$, we have 
\begin{equation}\label{est transport avec le flot}
\big\Vert \partial_{(x,\xi)}^{\gamma} \big( T(s,\phi_{\lambda}^t(x,\xi)) \big) \big\Vert \leq C_{\gamma} \exp\bigg( |\gamma|\Gamma(|t|+|s|)  \bigg).
\end{equation}
Furthermore, the same estimate holds for the derivatives of $T^{-1}(s,\phi_{\lambda}^t(x,\xi))$.
\end{lemma}

With Lemmas \ref{flot1}, \ref{Transport1} and \ref{transport avec le flot} at hand, we are now ready to prove Proposition \ref{u5}.

\begin{proof}
\begin{itemize}
\item[(i)] We start by proving estimate (\ref{As0}). Using formula (\ref{Faa de Bruno}) and estimate (\ref{floww1}), one can easily verify that for all $\gamma \in \mathbb N^{2n}$, there exists $C_{\gamma}>0$ such that for all $t\in \mathbb R$ and all $(x,\xi)\in \mathbb R^{2n}$, 
\begin{equation}\label{toute pre est}
\big\Vert \partial_{(x,\xi)}^{\gamma} \big( Q_0\circ \phi_{\lambda}^t(x,\xi)\big) \big\Vert \leq C_{\gamma} \exp\big(|\gamma|\Gamma |t| \big).
\end{equation}
Consequently, by differentiating $q_0(t)$ $|\gamma|$-times with respect to $(x,\xi)$ using the Leibniz formula, we obtain \begin{eqnarray*}
\big\Vert \partial_{(x,\xi)}^{\gamma}q_{0}(t,x,\xi) \big\Vert &\leq &\sum_{\beta\leq \gamma, \alpha\leq \beta}\binom{\beta}{\gamma}\binom{\alpha}{\beta} \big\Vert \partial_{(x,\xi)}^{\alpha}T^{-1}(t,x,\xi)\big\Vert \big\Vert \partial_{(x,\xi)}^{\beta-\alpha}\bigg(Q_0\big(\phi_{\lambda}^t(x,\xi) \big) \bigg)\big\Vert \big\Vert\partial_{(x,\xi)}^{\gamma-\beta}T(t,x,\xi)\big\Vert \\
&\leq & \sum_{\beta\leq \gamma, \alpha\leq \beta} C_{\alpha,\beta,\gamma} \exp\big( (|\gamma|+|\alpha|-|\beta|)\Gamma |t|\big) \exp\big( (|\beta|-|\alpha|)\Gamma |t| \big) \\
&\leq & C_{\gamma} \exp\big(|\gamma|\Gamma |t| \big),
\end{eqnarray*} 
uniformly for $(t,x,\xi)\in \mathbb R \times \mathbb R^{2n}$. Hence (\ref{As0}) holds.
\item[(ii)] We shall prove (\ref{Asj}) by induction with respect to $j\geq 1$. We give the proof only for $t\geq 0$, the case $t\leq 0$ is similar. Recall the expression of $q_j(t,x,\xi)$ 
$$
q_j(t,x,\xi)=T^{-1}(t,x,\xi)\bigg( Q_j\big( \phi_{\lambda}^t(x,\xi) \big) + \int_0^t T^{-1}\big(-s,\phi_{\lambda}^t(x,\xi)\big) B_j\big(s,\phi_{\lambda}^{t-s}(x,\xi)\big) T\big(-s,\phi_{\lambda}^t(x,\xi)\big) \; ds \bigg) T(t,x,\xi), 
$$
with 
\begin{equation*}
B_j(t,x,\xi):= \sum\limits_{\underset{0\leq p \leq j-1}{|\alpha|+|\beta|+k+p=j+1}} \tilde{\gamma}(\alpha,\beta) \bigg( {H_k}_{(\alpha)}^{(\beta)}{q_p(t)}_{(\beta)}^{(\alpha)}-(-1)^{|\alpha|-|\beta|}{q_p(t)}_{(\beta)}^{(\alpha)}{H_k}_{(\alpha)}^{(\beta)} \bigg)(x,\xi).
\end{equation*}
For $j=1$, we have 
\begin{equation*}
B_1(t,x,\xi)= \sum_{|\alpha|+|\beta|+k=2} \tilde{\gamma}(\alpha,\beta) \bigg( {H_k}_{(\alpha)}^{(\beta)}{q_0(t)}_{(\beta)}^{(\alpha)}-(-1)^{|\alpha|-|\beta|}{q_0(t)}_{(\beta)}^{(\alpha)}{H_k}_{(\alpha)}^{(\beta)} \bigg)(x,\xi).
\end{equation*}
Since $H_0$ is scalar according to assumption \textbf{(A1')}, for $k=0$ the previous sum vanishes and then $B_1$ can be rewritten as
\begin{eqnarray}\label{commutativité à l'ordre 1}
B_1(s,x,\xi) &=&  \sum_{|\alpha|+|\beta|=1} \tilde{\gamma}(\alpha,\beta) \bigg( {H_1}_{(\alpha)}^{(\beta)}{q_0(s)}_{(\beta)}^{(\alpha)}-(-1)^{|\alpha|-|\beta|}{q_0(s)}_{(\beta)}^{(\alpha)}{H_1}_{(\alpha)}^{(\beta)} \bigg)(x,\xi) \nonumber\\
&=& \frac{1}{2}\bigg( \{H_1,q_0(s)\}(x,\xi)-\{q_0(s),H_1\}(x,\xi)  \bigg).
\end{eqnarray}
Let $\gamma \in \mathbb N^{2n}$. Using assumption \textbf{(A2')} with $j=1$ and estimate (\ref{As0}), we get
\begin{equation*}
\big\Vert \partial_{(x,\xi)}^{\gamma}B_1(s,x,\xi) \big\Vert \leq C_{\gamma} \exp\bigg( (|\gamma|+1)\Gamma s \bigg),
\end{equation*}
uniformly for $s\geq 0$ and $(x,\xi)\in \mathbb R^{2n}$. Now, computing the derivatives of $B_1\circ \phi_{\lambda}^{t-s}$ by means of the Fa\'a Di Bruno formula (\ref{Faa de Bruno}) and combining the above estimate with (\ref{floww1}), we obtain 
\begin{equation}\label{huy1}
\big\Vert \partial_{(x,\xi)}^{\gamma}B_1\big(s,\phi_{\lambda}^{t-s}(x,\xi)\big) \big\Vert \leq C_{\gamma} \exp\bigg( \big( |\gamma|t+s\big) \Gamma \bigg),
\end{equation}
uniformly for $0\leq s \leq t$ and $(x,\xi)\in \mathbb R^{2n}$.

Put 
$$
A_1(t,s,x,\xi):= T^{-1}\big(-s,\phi_{\lambda}^t(x,\xi)\big) B_1\big(s,\phi_{\lambda}^{t-s}(x,\xi)\big) T\big(-s,\phi_{\lambda}^t(x,\xi)\big),
$$
and 
$$
\tilde{A}_1(t,x,\xi) := Q_1\big( \phi_{\lambda}^t(x,\xi) \big) + \int_0^t A_1(t,s,x,\xi) ds.
$$

Using Leibniz formula, estimates (\ref{huy1}) and (\ref{est transport avec le flot}) imply
\begin{align*}
\big\Vert \partial_{(x,\xi)}^{\gamma} A_1(t,s,x,\xi) \big\Vert \leq  \sum_{\beta \leq \gamma,\alpha\leq \beta} \binom{\beta}{\gamma}\binom{\alpha}{\beta} \big\Vert &\partial_{(x,\xi)}^{\alpha} \bigg( T^{-1}\big(-s,\phi_{\lambda}^t(x,\xi)\big) \bigg) \big\Vert  \\ 
&\times \big\Vert \partial_{(x,\xi)}^{\gamma -\beta} \bigg( T\big(-s,\phi_{\lambda}^t(x,\xi)\big) \bigg) \big\Vert  \big\Vert \partial_{(x,\xi)}^{\beta-\alpha} \bigg( B_1\big(s,\phi_{\lambda}^{t-s}(x,\xi)\big) \bigg) \big\Vert   
\end{align*}
\begin{eqnarray}
&\leq & \sum_{\beta \leq \gamma,\alpha\leq \beta}C_{\alpha,\beta,\gamma} \exp \bigg( \big( |\gamma|+|\alpha|-|\beta|\big)\big( t+s \big)  \Gamma \bigg) \exp\bigg( \big( (|\beta|-|\alpha|) t +s \big) \Gamma \bigg) \nonumber\\
&\leq & C_{\gamma} \exp\bigg( \big( |\gamma|t + (|\gamma|+1)s\big) \Gamma \bigg) \label{explain},
\end{eqnarray}
uniformly for $0\leq s \leq t$ and $(x,\xi)\in \mathbb R^{2n}$. Therefore, 
\begin{equation}\label{tbm}
\bigg\Vert \int_0^t \partial_{(x,\xi)}^{\gamma} A_1(t,s,x,\xi) ds \bigg\Vert \leq  C_{\gamma} \exp \big( \Gamma |\gamma| t \big) \int_0^{t} \exp\big( \Gamma (|\gamma|+1) s\big)  ds \leq C'_{\gamma} \exp \big( (2|\gamma|+1)\Gamma t\big).
\end{equation}
Combining this estimate with the fact that $Q_1\circ \phi_{\lambda}^t$ satisfies estimate \eqref{toute pre est}, we get 
$$
\big\Vert \partial_{(x,\xi)}^{\gamma} \tilde{A}_1(t,x,\xi) \big\Vert \leq C_{\gamma} \exp \big( (2|\gamma|+1)\Gamma t\big).
$$
Finally, we use Leibniz formula again to compute derivatives with respect to $(x,\xi)$ of $q_1(t,x,\xi)$. The above estimate together with estimate (\ref{estimation de T1}) give 
\begin{eqnarray*}
\big\Vert \partial_{(x,\xi)}^{\gamma} q_1(t,x,\xi) \big\Vert &\leq & \sum_{\beta \leq \gamma,\alpha\leq \beta} \binom{\beta}{\gamma}\binom{\alpha}{\beta} \big\Vert \partial_{(x,\xi)}^{\alpha} T^{-1}(t,x,\xi) \big\Vert \big\Vert \partial_{(x,\xi)}^{\gamma -\beta} T(t,x,\xi) \big\Vert  \big\Vert  \partial_{(x,\xi)}^{\beta-\alpha} \tilde{A}_1(t,x,\xi) \big\Vert \\
&\leq & C_{\gamma} \exp \big( (2|\gamma|+1)\Gamma t\big),
\end{eqnarray*}
uniformly for $t\geq 0$ and $(x,\xi)\in \mathbb R^{2n}$. Thus we proved (\ref{Asj}) for $j=1$.

Now, suppose that (\ref{Asj}) holds for all $j<r$. For $\gamma \in \mathbb N^{2n}$, we have 
\begin{equation}\label{hutn}
\partial_{(x,\xi)}^{\gamma} B_r(s,x,\xi) =  \sum\limits_{\underset{0\leq p \leq r-1}{|\alpha|+|\beta|+k+p=r+1}} \tilde{\gamma}(\alpha,\beta) \partial_{(x,\xi)}^{\gamma}  \bigg(  {H_k}_{(\alpha)}^{(\beta)}{q_p(s)}_{(\beta)}^{(\alpha)}  -(-1)^{|\alpha|-|\beta|} {q_p(s)}_{(\beta)}^{(\alpha)}{H_k}_{(\alpha)}^{(\beta)}\bigg)(x,\xi).
\end{equation}
We shall only focus on the first term of the above difference since the other term can be estimated similarly. Applying Leibniz formula, we get 
$$
 \partial_{(x,\xi)}^{\gamma}  \bigg( {H_k}_{(\alpha)}^{(\beta)}{q_p(s)}_{(\beta)}^{(\alpha)} \bigg)(x,\xi) = \sum_{\eta \leq \gamma} \binom{\eta}{\gamma} \partial_{(x,\xi)}^{\eta} {H_k}_{(\alpha)}^{(\beta)}(x,\xi) \partial_{(x,\xi)}^{\gamma-\eta}{q_p(s)}_{(\beta)}^{(\alpha)}(x,\xi).
$$
Firstly, since the sum in (\ref{hutn}) is over $((\alpha,\beta),k)\in \mathbb N^{2n}\times \{0,1\}$ such that $|\alpha|+|\beta|+k\geq 2$, then by assumption \textbf{(A2')} we have  $\partial_{(x,\xi)}^{\eta} {H_k}_{(\alpha)}^{(\beta)}\in L^{\infty}(\mathbb R^{2n})$, for all $\eta \in \mathbb N^{2n}$.

On the other hand, by the induction hypothesis, there exists a constant $C=C(\gamma,\eta,\alpha,\beta)>0$ such that for all $s\geq 0$ and $(x,\xi) \in \mathbb R^{2n}$, we have 
\begin{equation}
\big\Vert \partial_{(x,\xi)}^{\gamma-\eta} {q_p(s)}_{(\beta)}^{(\alpha)}(x,\xi) \big\Vert \leq C \exp\bigg( \big(2(|\gamma|-|\eta|+|\alpha|+|\beta|)+4p-3 \big)\Gamma s \bigg).
\end{equation}
Thus, taking the supremum over $0\leq |\eta|\leq |\gamma|$ and $|\alpha|+|\beta|=r+1-p$ with $0\leq p \leq r-1$, we obtain 
$$
\big\Vert \partial_{(x,\xi)}^{\gamma} B_r(s,x,\xi) \big\Vert \leq C_{\gamma} \exp\big( (2|\gamma|+4r-3)\Gamma s \big),
$$
uniformly with respect to $s\geq 0$ and $(x,\xi)\in \mathbb R^{2n}$.

Consequently, by applying Fa\'a Di Bruno's formula (\ref{Faa de Bruno}) and using estimate the flow (\ref{floww1}), we get 
\begin{equation}\label{huy2}
\big\Vert \partial_{(x,\xi)}^{\gamma}B_r\big(s,\phi_{\lambda}^{t-s}(x,\xi)\big) \big\Vert \leq C_{\gamma} \exp\bigg( (2|\gamma|+4r-3)\Gamma s + |\gamma|\Gamma (t-s) \bigg),
\end{equation}
uniformly for $0\leq s \leq t$ and $(x,\xi)\in \mathbb R^{2n}$. Put 
$$
A_r(t,s,x,\xi):= T^{-1}\big(-s,\phi_{\lambda}^t(x,\xi)\big) B_r\big(s,\phi_{\lambda}^{t-s}(x,\xi)\big) T\big(-s,\phi_{\lambda}^t(x,\xi)\big) ,
$$
and 
$$
\tilde{A}_r(t,x,\xi) := Q_r\big( \phi_{\lambda}^t(x,\xi) \big) + \int_0^t A_r(t,s,x,\xi) ds .
$$
Performing a similar computation as for $A_1$ and using estimates \eqref{huy2} and (\ref{transport avec le flot}), we obtain 
$$
\big\Vert \int_0^t \partial_{(x,\xi)}^{\gamma}A_r(t,s,x,\xi) ds \big\Vert \leq C_{\gamma} \exp\bigg((2|\gamma|+4r-3 ) \Gamma t \bigg),
$$
uniformly for $t\geq 0$ and $(x,\xi)\in \mathbb R^{2n}$. Consequently, using the fact that $Q_r \circ \phi_{\lambda}^t$ satisfies the estimate \eqref{toute pre est}, we get 
$$
\big\Vert \partial_{(x,\xi)}^{\gamma} \tilde{A}_r(t,x,\xi) \big\Vert \leq C_{\gamma} \exp\bigg((2|\gamma|+4r-3 ) \Gamma t \bigg).
$$

Finally, using the Leibniz formula and (\ref{estimation de T1}), we conclude 
\begin{eqnarray*}
\big\Vert \partial_{(x,\xi)}^{\gamma} q_r(t,x,\xi) \big\Vert &\leq & \sum_{\beta \leq \gamma,\alpha\leq \beta} \binom{\beta}{\gamma}\binom{\alpha}{\beta} \big\Vert \partial_{(x,\xi)}^{\alpha} T^{-1}(t,x,\xi) \big\Vert \big\Vert \partial_{(x,\xi)}^{\gamma -\beta} T(t,x,\xi) \big\Vert  \big\Vert \partial_{(x,\xi)}^{\beta - \alpha} \tilde{A}_r(t,x,\xi) \big\Vert \\
&\leq & C_{\gamma} \exp \bigg( (2|\gamma|+4r-3 ) \Gamma t \bigg),
\end{eqnarray*}
uniformly for $t\geq 0$ and $(x,\xi)\in \mathbb R^{2n}$. Hence (\ref{Asj}) holds for $j=r$. This ends the proof of Proposition \ref{u5}.
\end{itemize}
\end{proof}

\begin{remark}\label{explanation}
Notice that estimate \eqref{Asj} on the derivatives of the symbols $q_j(t,x,\xi)$, $j\geq 1$, is different from the one proved in the scalar case (see \cite[Theorem 1.4]{Rob}). This is caused by the derivatives of the term $T(-s,\phi_{\lambda}^t(x,\xi))$ appearing in the expression \eqref{symbole j multiple scalaire} of $q_j(t)$ which does not exist in the scalar case. Assume that $Q$ is classical, i.e. $Q(x,\xi)= Q_0(x,\xi)$ (as in \cite{Rob}) and let us explain this difference at the level of sub-principal symbols, i.e. for $j=1$. We have shown in Remark \ref{les symboles dans le cas scalaire} that in the case where $H_1$ is also scalar-valued, the sub-principal symbol $q_{1,\text{sca}}(t,x,\xi)$ is given by 
$$
q_{1,\text{sca}}(t,x,\xi) =  \int_0^t B_1\big(s,\phi_{\lambda}^{t-s}(x,\xi)\big) ds ,
$$
where $B_1$ is defined by \eqref{Les Bj}. Using estimate \eqref{huy1} on the derivatives of $B_1\big(s,\phi_{\lambda}^{t-s}(x,\xi)\big)$, one obtains
\begin{equation}\label{estim scalaire}
\big| \partial_{(x,\xi)}^{\gamma}q_{1,\text{sca}}(t,x,\xi) \big| \leq C_{\gamma} \exp\big( (|\gamma|+1)\Gamma |t| \big), \quad \forall \gamma\in \mathbb N^{2n}, t\in \mathbb R, (x,\xi)\in \mathbb R^{2n},
\end{equation}
which is the estimate proved in \cite[Theorem 1.4]{Rob}. Now, in the case where $H_1$ is matrix-valued, according to \eqref{symbole j multiple scalaire} (taking into account the fact that $Q_1=0$), $q_1(t,x,\xi)$ is given by 
$$
q_1(t,x,\xi) = T^{-1}(t,x,\xi)\bigg( \int_0^t T^{-1}\big(-s,\phi_{\lambda}^t(x,\xi)\big) B_1\big(s,\phi_{\lambda}^{t-s}(x,\xi)\big) T\big(-s,\phi_{\lambda}^t(x,\xi)\big) \; ds \bigg) T(t,x,\xi).
$$
By going back to estimate \eqref{explain}, one sees that due to the term $T\big(-s,\phi_{\lambda}^t(x,\xi)\big)$, when differentiating $q_1(t,x,\xi)$ $|\gamma|$- times with respect to $(x,\xi)$, there is a loss of $ \exp\big(|\gamma|\Gamma|t|\big)$ compared to \eqref{estim scalaire}, i.e. we have 
$$
\big\Vert\partial_{(x,\xi)}^{\gamma}q_1(t,x,\xi) \big\Vert \leq C_{\gamma} \exp\big( (2|\gamma|+1)\Gamma |t| \big), \quad \forall \gamma\in \mathbb N^{2n}, t\in \mathbb R, (x,\xi)\in \mathbb R^{2n}.
$$
As pointed out in (iii) of Remark \ref{première remarque}, this explain the fact that our estimate on the remainder term \eqref{remainder estimate T1} is different from the one obtained in the scalar case.

%and then the constant $\displaystyle{\frac{1}{4\Gamma}}$ in the Ehrenfest time up to which the semiclassical approximation still valid is half of the one obtained in the scalar case.

\end{remark}

\subsection{Proof of Theorem \ref{main T1}}\label{prf T1} 

The proof of estimate \eqref{remainder estimate T1} is based on estimates \eqref{As0} and \eqref{Asj} and the control of the remainder terms in the composition formula of $\h$-pseudodifferential operators \eqref{dév asy annexe}. We follow the method of Bouzouina-Robert \cite{Rob}.

For $A,B$ two semiclassical symbols in suitable classes of symbols and $k\in \mathbb N$, we define
\begin{equation}\label{la fonction tilde R}
\tilde{R}_{k}(A,B):=i\h^{-(k+1)}\big(R_k(A,B)-R_k(B,A) \big),
\end{equation}
where $R_k(A,B,x,\xi;\h):= A \# B(x,\xi;\h) - \sum_{j=0}^k \h^j (A\#B)_j$ denotes the remainder term of order $k$ in the asymptotic expansion of the symbol $A\#B$ (see appendix \ref{semi-classical background}). 

%\textcolor{blue}{In terms of the Moyal bracket, we have 
%$$
%\tilde{R}_k(A,B) = \h^{-k} R_{k-1}(\{A,B\}^*),
%$$
%where $R_{k-1}(\{A,B\}^*)$ denotes the remainder term of order $k-1$ in the asymptotic expansion of $\{A,B\}^*$.}

For $N\in \mathbb N$, we set 
$$
Q_N(t) := Q(t) - \sum_{j=0}^N \h^j \big(q_j(t)\big)^w(x,\h D_x).
$$
The first step in the proof of estimate (\ref{remainder estimate T1}) is the following lemma.
\begin{lemma}
For all $N\in \mathbb N$, the following estimate holds
\begin{align}\label{main est T1}
\big\Vert  Q_N(t) \big\Vert_{\mathcal{L}(L^2(\mathbb R^n)\otimes \mathbb C^m)} \leq \h^{N+1} \bigg\Vert \int_0^t U_H(-s) \big(R^{(N+1)}(t-s)\big)^w U_H(s) ds\bigg\Vert_{\mathcal{L}(L^2(\mathbb R^n)\otimes \mathbb C^m)} +\mathcal{O}(\h^{N+1}),
\end{align}
uniformly for $t\in \mathbb R$, with 
\begin{eqnarray}\label{writn}
R^{(N+1)}(t):= \tilde{R}_{N+1}(H,q_0(t)) + \tilde{R}_{N}(H,q_1(t)) + \cdots + \tilde{R}_1(H,q_N(t)) = \sum_{j=0}^N \tilde{R}_{N+1-j}(H,q_j(t)).
\end{eqnarray}
\end{lemma}
\begin{proof}
Let $N\in \mathbb N$ and define 
$$
q^{(N)}(t,x,\xi;\h):= \sum_{j=0}^N \h^j q_j(t,x,\xi).
$$
According to the Cauchy problems $(\mathcal{C}_j)_{j\geq 0}$ satisfied by the symbols $q_j(t)$, for all $j\geq 0$, we have 
\begin{equation}
\frac{d}{dt} q_j(t) = \{H,q_0(t)\}^*_j + \{H,q_1(t)\}^*_{j-1} + \{H,q_2(t)\}^*_{j-2} + \cdots + \{H,q_j(t)\}^*_0, 
\end{equation}
where we recall that for $0\leq k \leq j$, $\{H,q_k(t)\}^*_{j-k}$ denotes the coefficient of $\h^{j-k}$ in the asymptotic expansion of the Moyal bracket $\{H,q_k(t)\}^*$ (see appendix \ref{semi-classical background}). Then
$$
\frac{d}{dt} q^{(N)}(t) = \sum_{j=0}^N \h^j \frac{d}{dt} q_j(t) = \sum_{j=0}^N \h^j \{H,q_0(t)\}^*_j + \h \sum_{j=0}^{N-1} \h^j \{H,q_1(t)\}^*_j + \cdots + \h^N \{H,q_N(t)\}^*_0.
$$
Using the formula of asymptotic expansion of the Moyal bracket (\ref{rMp}), we obtain 
\begin{equation}\label{canbe}
\{H,q^{(N)}(t)\}^* = \frac{d}{dt} q^{(N)}(t) + \h^{N+1} R^{(N+1)}(t),
\end{equation}
with $R^{(N+1)}(t)$ defined by (\ref{writn}).
%Set 
%$$
%Q_N(t):=Q(t)- \big(q^{(N)}(t)\big)^w.
%$$
A simple computation using (\ref{canbe}) yields 
\begin{eqnarray*}
\frac{d}{ds}\bigg( U_H(-s) Q_{N}(t-s) U_H(s) \bigg) &=& U_H(-s) \bigg( \frac{i}{h}[H^w,Q_{N}(t-s)]-\frac{d}{dt} Q_{N}(t-s)\bigg)U_H(s)    \\
&=& U_H(-s) \bigg(  \frac{d}{dt} \left( q^{(N)}(t-s) \right)^w - \frac{i}{\h} \big[ H^w, \left( q^{(N)}(t-s) \right)^w \big]  \bigg) U_H(s) \\
&=& - \h^{N+1}U_H(-s) \big(R^{(N+1)}(t-s)\big)^w U_H(s).  
\end{eqnarray*}
Therefore, by integrating in $s$ and using the fact that 
$$
\big\Vert Q_N(0) \big\Vert_{\mathcal{L}(L^2(\mathbb R^n)\otimes \mathbb C^m)} = \big\Vert Q^w(x,\h D_x;\h) - \sum_{j=0}^N \h^j Q_j^w(x,\h D_x) \big\Vert_{\mathcal{L}(L^2(\mathbb R^n)\otimes \mathbb C^m)} = \mathcal{O}(\h^{N+1}),
$$ 
we get
$$
\big\Vert Q_N(t)  \big\Vert_{\mathcal{L}(L^2(\mathbb R^n)\otimes \mathbb C^m)} \leq \h^{N+1} \bigg\Vert \int_0^t U_H(-s) \big( R^{(N+1)}(t-s)\big)^w U_H(s) ds \bigg\Vert_{\mathcal{L}(L^2(\mathbb R^n)\otimes \mathbb C^m)}  + \mathcal{O}(\h^{N+1}),
$$
uniformly for $t\in \mathbb R$. This ends the proof of the lemma.
\end{proof}

\textbf{End of the proof of Theorem \ref{main T1}.}

It remains now to estimate the $\mathcal{L}(L^{2}(\mathbb R^n)\otimes \mathbb C^m)$-norm of the operator $\big(R^{(N+1)}(t)\big)^{w}$. For that, we shall employ the Calder\'on-Vaillancourt theorem (Theorem \ref{Cal}). We shall therefore need estimates on the derivatives with respect to $(x,\xi)$ of the symbol $R^{(N+1)}(t,x,\xi;\h)$. %This estimation is based on the estimates (\ref{As0}), (\ref{Asj}) and theorem \ref{est ress}.

Let $N\in \mathbb N$ and $0\leq j \leq N$. We have 
\begin{equation}\label{rema point es}
\tilde{R}_{N+1-j}(H,q_j(t)) = \tilde{R}_{N+1-j} (H_0,q_j(t))  +   \tilde{R}_{N-j} (H_1,q_j(t)).
\end{equation}

Let $k\in \{0,1\}$. According to Theorem \ref{est ress} (combined with Remark \ref{extension}), for all $\gamma \in \mathbb N^{2n}$, there exists a constant $C=C(n,N,j,\gamma,k)>0$ such that for all $t\in \mathbb R$ and all $u\in \mathbb R^{2n}$, we have 
\begin{align}
\big\Vert\partial_{u}^{\gamma}\tilde{R}_{N+1-j-k}(H_k,q_{j}(t);u) \big\Vert \leq C \sup_{(\ast)} \bigg( \big\Vert\partial_{v}^{(\alpha,\beta)+\eta}H_{k}(v+u)\big\Vert  \big\Vert\partial_{w}^{(\beta,\alpha)+\kappa}q_{j}(t,w+u)\big\Vert \bigg),
\end{align}
where $\sup_{(\ast)}$ is the supremum under the conditions
$$
(\ast): \quad v,w \in \mathbb R^{2n}, \quad \eta,\kappa\in \mathbb N^{2n};|\eta|+|\kappa|\leq 4n+1+|\gamma|, \quad \alpha,\beta \in \mathbb N^n;|\alpha|+|\beta|=N+2-j-k.
$$
Observe first that by assumption \textbf{(A2')}, for $k\in \{0,1\}$, for all multi-indices $((\alpha,\beta),\eta) \in \mathbb N^{2n}\times \mathbb N^{2n}$ and all $0\leq j\leq N$ with $|\alpha|+|\beta|=N+2-j-k$, we have 
$$
\partial_{(x,\xi)}^{(\alpha,\beta)+\eta}H_{k}\in L^{\infty}(\mathbb R^{2n}).
$$

On the other hand, using the estimates given by Proposition \ref{u5}, for all $((\alpha,\beta),\kappa)\in \mathbb N^{2n}\times \mathbb N^{2n}$ and all $j\geq 0$, there exists $C_j=C(\alpha,\beta,\kappa,j)>0$ such that 
$$
\big\Vert \partial_{(x,\xi)}^{(\beta,\alpha)+\kappa}q_0(t,x,\xi) \big\Vert \leq C_0 \exp\bigg( (|\alpha|+|\beta|+|\kappa|)\Gamma |t| \bigg)
$$
and for $j\geq 1$
$$
\big\Vert \partial_{(x,\xi)}^{(\beta,\alpha)+\kappa}q_j(t,x,\xi) \big\Vert \leq C_j \exp\bigg( (2 (|\alpha|+|\beta|+|\kappa|)+4j-3)\Gamma |t|\bigg),
$$ 
uniformly for $(t,x,\xi)\in \mathbb R \times \mathbb R^{2n}$.

Therefore, taking the supremum over ($\ast$), there exists $C=C(\gamma,N,n,j,k)>0$ such that for all $t\in \mathbb R$ and all $(x,\xi)\in \mathbb R^{2n}$, we have 
$$
\big\Vert\partial_{(x,\xi)}^{\gamma}\tilde{R}_{N+1-j-k}(H_k,q_{j}(t);x,\xi) \big\Vert \leq C \exp\bigg( \big(2|\gamma|+2N+8n+3+2j-2k \big)\Gamma |t|\bigg).
$$
Now, summing over $j=0,...,N$, we get 
\begin{equation}
\big\Vert  \partial_{(x,\xi)}^{\gamma}R^{(N+1)}(t,x,\xi) \big\Vert \leq C_{n,N,\gamma} \exp\bigg( \big(2|\gamma|+4N+8n+3 \big)\Gamma |t| \bigg)
\end{equation}
uniformly for $t\in \mathbb R$ and $(x,\xi)\in \mathbb R^{2n}$.

Consequently, using the Calder\'on-Vaillancourt theorem (Theorem \ref{Cal}), we deduce
$$
\big\Vert \big(R^{(N+1)}(t)\big)^w(x,\h D_x;\h) \big\Vert_{\mathcal{L}(L^2(\mathbb R^n)\otimes \mathbb C^m)} \leq C_{n,N} \exp\bigg( \big( 4N+\delta_n\big)\Gamma |t| \bigg),
$$
uniformly for $t\in \mathbb R$, where $\delta_n$ is an integer depending only on the dimension $n$. 

By going back to \eqref{main est T1}, we obtain 
\begin{eqnarray*}
\Vert Q_N(t) \Vert_{\mathcal{L}(L^2(\mathbb R^n)\otimes \mathbb C^m)} &\leq & \h^{N+1} \int_0^t \big\Vert \big( R^{N+1}(t-s) \big)^w \big\Vert_{\mathcal{L}(L^2(\mathbb R^n)\otimes \mathbb C^m)} ds + \mathcal{O}(\h^{N+1}) \\
&\leq & C_{N} \: \h^{N+1} \int_0^t \exp \bigg( \big( 4N + \delta_n \big) \Gamma (t-s) \bigg) ds + \mathcal{O}(\h^{N+1}) \\
&\leq & C'_{N} \: \h^{N+1} \exp \bigg( \big( 4N + \delta_n \big) \Gamma t \bigg), 
\end{eqnarray*}
uniformly for $t\geq 0$. Analogously, we prove the estimate for $t\leq 0$. This ends the proof of Theorem \ref{main T1}.

\begin{flushright}
$\square$
\end{flushright}

\section{General case}\label{Generalization}
We now turn to the study of the general case where the principal symbol of the Hamiltonian $H$ which generates the time evolution is no longer a scalar multiple of the identity in $M_m(\mathbb C)$. 

Let $H \in S(g)$ be an hermitian-valued semiclassical Hamiltonian satisfying \textbf{(A0)} and suppose that \textbf{(A1)} is fulfilled. We consider the time evolution of a bounded quantum observable $Q^w(x,\h D_x;\h)$ associated to a semiclassical observable $Q\in S_{\text{sc}}(1)$, given by 
$$
Q(t):= U_H(-t)Q^w(x,\h D_x;\h)U_H(t), \quad t\in \mathbb R.
$$
%As already mentioned in section \ref{sec}, we shall reduce the study of $Q(t)$ to that of the block-diagonal family of Heisenberg observables $Q_{\nu}(t)$ defined by (\ref{blok}). 
 
\subsection{Semiclassical projections}
As already mentioned in section \ref{sec}, the first step in the study of $Q(t)$ consists in the construction of the semiclassical projections associated to $H^w(x,\h D_x;\h)$.
%The main ingredient in this reduction is the following result 
\begin{theorem}\label{projections}
Let assumption \textbf{(A1)} be satisfied. For every $1\leq \nu \leq l$, there exists a semiclassical symbol $\tilde{P}_{\nu}(x,\xi;\h)\sim \sum_{j\geq 0} \h^j \tilde{P}_{\nu,j}(x,\xi)$ in $S(1, \mathbb R^{2n}, M_m(\mathbb C))$ (unique modulo $S^{-\infty}(1;\mathbb R^{2n},M_m(\mathbb C))$) such that modulo $\mathcal{O}(\h^{\infty})$ in $\mathcal{L}(L^2(\mathbb R^n)\otimes \mathbb C^m)$, we have  
 \begin{equation}\label{diag}
\tilde{P}_{\nu}^w \tilde{P}_{\nu}^w = \tilde{P}_{\nu}^w
\end{equation}
\begin{equation}\label{diag2}
 \tilde{P}_{\nu}^w = \big( \tilde{P}_{\nu}^w\big)^* ,
\end{equation}
\begin{equation}\label{diag3}
[\tilde{P}_{\nu}^w,H^w] = 0.
\end{equation}
\begin{equation}\label{différents indices}
\tilde{P}_{\mu}^w \tilde{P}_{\nu}^w= \tilde{P}_{\nu}^w \tilde{P}_{\mu}^w = 0, \quad \forall 1\leq \mu\neq \nu \leq l,
\end{equation}
\begin{equation}\label{identité}
\sum_{\nu=1}^l \tilde{P}_{\nu}^w = \text{id}_{L^2(\mathbb R^n) \otimes \mathbb C^m}. 
\end{equation}
In particular, $\tilde{P}_{\nu,0}(x,\xi) = P_{\nu,0}(x,\xi)$ is the orthogonal projector onto $\text{Ker}\: \big( H_0(x,\xi)-\lambda_{\nu}(x,\xi) \big)$.
\end{theorem}
There are at least two methods of proof for this result. The first one followed by Brummelhuis and Nourrigat \cite{Brum} (see also \cite{Nenciu}) consists in solving, in the space of formal power series of $\h$, the symbolic equations corresponding to (\ref{diag}), (\ref{diag2}) and (\ref{diag3}),
$$
\tilde{P}_{\nu}(x,\xi;\h)\# \tilde{P}_{\nu}(x,\xi;\h) \sim \tilde{P}_{\nu}(x,\xi;\h) \sim \big( \tilde{P}_{\nu}(x,\xi;\h)\big)^*, \quad [ \tilde{P}_{\nu}(x,\xi;\h),H(x,\xi;\h)]_{\#}\sim 0.
$$
The second method due to Helffer and Sj\"ostrand \cite{Hel} uses the Riesz projectors and the symbolic calculus of $\h$-pseudodifferential operators (see \cite[ch. 8]{dim}). For the reader's convenience, we give in the appendix \ref{projjj} an outline of the 
 proof of Theorem \ref{projections}  following the method in \cite{Hel}. %using the second method since it gives explicit formulas for the symbols $\tilde{P}_{\nu}(x,\xi;\h)$.

%\begin{remark}
%Notice that the construction of the semiclassical projections $\tilde{P}_{\nu}$, $1\leq \nu \leq l$, remains valid under the following weaker assumption on the gap between the eigenvalues of $H_0$
%\begin{equation}\label{weaken gap}
%\min_{1\leq \mu \neq \nu \leq l} |\lambda_{\mu}(x,\xi) - \lambda_{\nu}(x,\xi)| \geq \rho, 
%\end{equation}
%where $\rho>0$ is independent of $(x,\xi)\in \mathbb R^{2n}$. In this case, we have also to assume \textbf{(A2)}. See Remark \ref{affaiblir gap} for more details.
%\end{remark}

%\begin{remark}
%Notice that in \cite[Proposition 2.1]{bolte}, the construction of the semiclassical projections $\tilde{P}_{\nu}(x,\xi;\h)$ was made in the case of a general hermitian-valued Hamiltonian $H(x,\xi;\h) \sim \sum_{j\geq 0}\h^j H_j(x,\xi)$ in $S(g)$, under the assumption : $\exists \rho> 0$ such that for all $(x,\xi)\in \mathbb R^{2n}$,
%$$
%\min_{1\leq \mu \neq \nu \leq l} |\lambda_{\mu}(x,\xi) - \lambda_{\nu}(x,\xi)| \geq \rho g(x,\xi).
%$$
%In our case, because we require assumptions \textbf{(A2)} and \eqref{dév asymptotique tildeH} on the derivatives of the symbols $H_j$, $j\geq 0$, \eqref{Gap} is sufficient.
%\end{remark}

For our next purposes, it is more convenient to work with exact projections, i.e. with operators which satisfy \eqref{diag} exactly, not only modulo $\mathcal{O}(\h^{\infty})$ in norm $\mathcal{L}(L^2(\mathbb R^n)\otimes \mathbb C^m)$. To do this, we follow an idea from \cite{Nenciu1} (see also \cite{Nenciu, Sor}) which consists in introducing the operators
$$
\mathcal{P}_{\nu} := \frac{i}{2\pi} \int_{|z-1| = \frac{1}{2}} (\tilde{P}_{\nu}^w - z)^{-1} dz, \quad 1\leq \nu \leq l.
$$
For $\h$ small enough, \eqref{diag} implies that the spectrum of $\tilde{P}^w_{\nu}$ is concentrated near $0$ and $1$ (see \cite{Nenciu1}), then $\mathcal{P}_{\nu}$ is well defined and satisfies
\begin{equation}\label{tildediag}
\mathcal{P}_{\nu} \mathcal{P}_{\nu} = \mathcal{P}_{\nu} = \mathcal{P}_{\nu}^*. 
\end{equation}
By a similar computation as in \cite[sec. III]{Nenciu1}, one gets  
\begin{equation}\label{rtuip}
\big\Vert \mathcal{P}_{\nu} - \tilde{P}_{\nu}^w \big\Vert_{\mathcal{L}(L^2(\mathbb R^n)\otimes \mathbb C^m)}= \mathcal{O}(\h^{\infty}).
\end{equation}
Then, by Beals's characterization of $\h$-pseudodifferential operators (see \cite[Proposition 8.3]{dim}), $\mathcal{P}_{\nu}$ is an $\h$-pseudodifferential operator with symbol $P_{\nu}(x,\xi;\h)\in S_{\text{sc}}(1)$, i.e. $\mathcal{P}_{\nu} = P_{\nu}^w(x,\h D_x;\h)$. Moreover, we have (see \cite{Nenciu1})
\begin{equation}\label{tildediag1}
\big\Vert [\mathcal{P}_{\nu},H^w] \big\Vert_{\mathcal{L}(L^2(\mathbb R^n)\otimes \mathbb C^m)} = \mathcal{O}\left( \big\Vert [\tilde{P}^w_{\nu},H^w] \big\Vert_{\mathcal{L}(L^2(\mathbb R^n)\otimes \mathbb C^m)} \right) = \mathcal{O}(\h^{\infty}).
\end{equation}

\subsection{Block-diagonalization}
In what follows, we shall use the notation $P_{\nu}^w$ for $\mathcal{P}_{\nu}$. We introduce the family of Heisenberg operators $Q_{\nu}(t)$ defined by  
\begin{equation}\label{les blocks Qnu}
Q_{\nu}(t) := e^{\frac{it}{\h}P_{\nu}^w H^w P_{\nu}^w} P_{\nu}^w Q^w P_{\nu}^w e^{-\frac{it}{\h}P_{\nu}^w H^w P_{\nu}^w}, \quad 1\leq \nu \leq l.
\end{equation}
The main result of this paragraph is the following.

\begin{proposition}\label{proposition réduction} Assume that $H$ satisfies the assumption of Theorem \ref{projections}. 
\begin{itemize}
\item[(i)] If $Q\in \mathcal{Q}(1)$, then the following estimate holds
\begin{equation}\label{réduction}
\big\Vert Q(t)-\sum_{\nu=1}^l Q_{\nu}(t) \big\Vert_{\mathcal{L}(L^2(\mathbb R^n)\otimes \mathbb C^m)} = \mathcal{O}\big( (1+|t|) \h^{\infty} \big),
\end{equation}
uniformly for $t\in \mathbb R$. 
\item[(ii)] Assume that $Q_0(x,\xi) = \sum_{\nu=1}^l P_{\nu,0}(x,\xi) \tilde{Q}(x,\xi) P_{\nu,0}(x,\xi)$ for some $\tilde{Q}\in S(1)$. Then, we have 
\begin{equation}\label{réduction modulo h}
\big\Vert Q(t) - \sum_{\nu=1}^l Q_{\nu}(t) \big\Vert_{\mathcal{L}(L^2(\mathbb R^n)\otimes \mathbb C^m)} = \mathcal{O}\big( (1+|t|) \h \big),
\end{equation}
uniformly for $t\in \mathbb R$.
\end{itemize}
\end{proposition}

The following lemma is the first step in the proof of the above proposition.

\begin{lemma}\label{jus}
For all $1\leq \nu \leq l$, we have 
\begin{equation}\label{le}
e^{\frac{it}{\h}H^w}P_{\nu}^w = e^{\frac{it}{\h} P_{\nu}^wH^wP_{\nu}^w}P_{\nu}^w+\mathcal{O}(|t| \h^{\infty}),
\end{equation}
uniformly for $t\in \mathbb R$.
\end{lemma}
\begin{proof} Fix $\nu \in \{1,...,l\}$ and set 
$$
U(t):=e^{\frac{it}{\h}H^w}P_{\nu}^w, \quad V(t):= e^{\frac{it}{\h} P_{\nu}^wH^wP_{\nu}^w}P_{\nu}^w, \quad t\in \mathbb R.
$$
Obviously, $U(t)$ satisfies
\begin{equation}\label{equ1}
\left\{   
\begin{array}{rcl}
(\h D_t-H^w ) U(t) &=& 0 \\ 
U(0) &=& P_{\nu}^w.
\end{array} 
\right.
\end{equation}
Here we use the standard notation $D_t:=\frac{1}{i} \partial_t$. Let us prove that $V(t)$ satisfies 
\begin{equation}\label{equ2}
\left\{   
\begin{array}{rcl}
(\h D_t-H^w ) V(t)&=& I(t) \\
V(0)&=& P_{\nu}^w,
\end{array} 
\right.
\end{equation}
with $\Vert I(t) \Vert_{\mathcal{L}(L^2(\mathbb R^n)\otimes \mathbb C^m)} = \mathcal{O}(\h^{\infty})$, uniformly for $t\in \mathbb R$. Put
$$
R(t):=V(t)-P_{\nu}^w e^{\frac{it}{\h}P_{\nu}^w H^w P_{\nu}^w}.
$$
Using \eqref{tildediag}, we get  
\begin{eqnarray}
\h D_t R(t) &=& e^{\frac{it}{\h}P_{\nu}^w H^w P_{\nu}^w} P_{\nu}^w H^w (P_{\nu}^w)^2 - (P_{\nu}^w)^2 H^w P_{\nu}^w e^{\frac{it}{\h} P_{\nu}^w H^w P_{\nu}^w} \nonumber \\
&=& e^{\frac{it}{\h}P_{\nu}^w H^w P_{\nu}^w} P_{\nu}^w H^w P_{\nu}^w - P_{\nu}^w H^w P_{\nu}^w e^{\frac{it}{\h} P_{\nu}^w H^w P_{\nu}^w} \nonumber \\
&=& 0,
\end{eqnarray}
which together with $R(0)=0$ yields 
\begin{equation}
R(t)=0, \quad \forall t\in \mathbb R.
\end{equation}
Now, a simple computation gives 
\begin{eqnarray*}
(\h D_t - H^w) V(t) &=& (P_{\nu}^w H^w P_{\nu}^w - H^w)e^{\frac{it}{\h}P_{\nu}^w H^w P_{\nu}^w}P_{\nu}^w\\
&=& (P_{\nu}^w H^w P_{\nu}^w - H^w) P_{\nu}^w e^{\frac{it}{\h}P_{\nu}^w H^w P_{\nu}^w} +(P_{\nu}^w H^w P_{\nu}^w - H^w) R(t)\\
&=& (P_{\nu}^w H^w P_{\nu}^w - H^w) P_{\nu}^w e^{\frac{it}{\h}P_{\nu}^w H^w P_{\nu}^w} \\
&=:& I(t).
\end{eqnarray*}
According to \eqref{tildediag}, we have 
\begin{eqnarray}
I(t) = \bigg( P_{\nu}^w H^w (P_{\nu}^w)^2 - H^w P_{\nu}^w \bigg) e^{\frac{it}{\h}P_{\nu}^w H^w P_{\nu}^w} = \bigg( P_{\nu}^w H^w P_{\nu}^w - H^w P_{\nu}^w\bigg) e^{\frac{it}{\h}P_{\nu}^w H^w P_{\nu}^w}. \label{ma9}
\end{eqnarray}
From \eqref{tildediag1}, we have $ P_{\nu}^w H^w = H^w P_{\nu}^w + \mathcal{O}(\h^{\infty})$ which together with the fact that $\Vert P_{\nu}^w \Vert = \mathcal{O}(1)$ yields 
\begin{equation}\label{à la pl}
P_{\nu}^w H^w P_{\nu}^w =  H^w (P_{\nu}^w)^2 + \mathcal{O}(\h^{\infty}) =  H^w P_{\nu}^w + \mathcal{O}(\h^{\infty}),
\end{equation}
where in the last step, we used \eqref{tildediag} again. Putting together \eqref{à la pl} and \eqref{ma9}, we obtain
\begin{equation}
\Vert I(t) \Vert_{\mathcal{L}(L^2(\mathbb R^n)\otimes \mathbb C^m)} = \mathcal{O}(\h^{\infty}), \quad \text{uniformly for}\; t\in \mathbb R.
\end{equation}

Now, according to Duhamel's principle, we have
$$
V(t)-U(t)=\frac{1}{\h} \int_0^t U(t-s)\mathcal{O}(\h^{\infty})ds,
$$ 
which yields
\begin{equation}\label{ut}
U(t)-V(t)=\mathcal{O}(|t|\h^{\infty}), \quad \text{uniformly for}\; t\in \mathbb R.
\end{equation}
This ends the proof of the lemma.
\end{proof}
Turn now to the proof of Proposition \ref{proposition réduction}.

\noindent
\textbf{Proof of Proposition \ref{proposition réduction} :}

\noindent
By conjugating $Q(t)$ with $\sum_{\nu=1}^l P_{\nu}^w(x,\h D_x;\h)= \text{id}_{L^2(\mathbb R^n)\otimes \mathbb C^m}+ \mathcal{O}(\h^{\infty})$ and using the above lemma, we get
\begin{eqnarray}
Q(t)&=&\sum_{\mu,\nu=1}^l e^{\frac{it}{\h}P_{\mu}^wH^w P_{\mu}^w} P_{\mu}^w Q^w P_{\nu}^w e^{-\frac{it}{\h}P_{\nu}^wH^w P_{\nu}^w} +\mathcal{O}\big( (1+|t|) \h^{\infty} \big) \nonumber\\
&=& \sum_{\nu=1}^l Q_{\nu}(t) + \sum_{\mu\neq\nu=1}^l e^{\frac{it}{\h}P_{\mu}^wH^w P_{\mu}^w} P_{\mu}^w Q^w P_{\nu}^w e^{-\frac{it}{\h}P_{\nu}^wH^w P_{\nu}^w} +\mathcal{O}\big( (1+|t|)\h^{\infty} \big), \label{r}
\end{eqnarray}
uniformly for $t\in \mathbb R$. 

Passing from symbols to operators, the assumption $Q\in \mathcal{Q}(1)$ implies 
$$
Q^w=\sum_{\nu=1}^l P_{\nu}^w Q^w P_{\nu}^w + \mathcal{O}(\h^{\infty}).
$$
Therefore, using that $P_{\mu}^w P_{\nu}^w = \mathcal{O}(\h^{\infty})$ for $\mu \neq \nu$ (which follows from \eqref{différents indices} and \eqref{rtuip}), we deduce  
\begin{equation}\label{en terme d'opérateurs}
P_{\mu}^w Q^w P_{\nu}^w = \mathcal{O}(\h^{\infty}), \quad \forall 1\leq \mu \neq \nu \leq l.
\end{equation}
Consequently, the norm $\mathcal{L}(L^2(\mathbb R^n)\otimes \mathbb C^m)$ of the second term in the right hand side of \eqref{r} is equal to $\mathcal{O}(\h^{\infty})$ uniformly for $t\in \mathbb R$. Thus \eqref{réduction} holds. 

Now, assume that $Q_0(x,\xi) = \sum_{\nu=1}^l P_{\nu,0}(x,\xi) \tilde{Q}(x,\xi) P_{\nu,0}(x,\xi)$, for some arbitrary $\tilde{Q}\in S(1)$. This implies that 
$$
[Q_0(x,\xi),P_{\nu,0}(x,\xi)] =0 , \quad \forall 1\leq \nu \leq l, \forall (x,\xi)\in \mathbb R^{2n}.
$$
Combining this with the fact that $Q^w$ and $P_{\nu}^w$ are bounded in $L^2(\mathbb R^n)\otimes \mathbb C^m$, we get 
$$
[Q^w,P_{\nu}^w] = \mathcal{O}(\h) , \quad \forall 1\leq \nu \leq l,
$$
which by using that $P_{\mu}^w P_{\nu}^w = \mathcal{O}(\h^{\infty})$ for $\mu \neq \nu$ again, implies $P_{\nu}^w Q^w P_{\mu}^w = \mathcal{O}(\h)$, for all $1\leq \mu \neq \nu \leq l$. Thus \eqref{réduction modulo h} holds immediately from \eqref{r}.
\begin{flushright}
$\square$
\end{flushright}

\begin{remark}
According to estimate \eqref{réduction}, the study of $Q(t)$ is reduced modulo $\mathcal{O}(\h^{\infty})$ to that of the blocks $Q_{\nu}(t)$ defined by \eqref{les blocks Qnu}. The main property of this reduction lies in the fact that it is preserved up to times of order $\h^{-\infty}$ (i.e. of order $\h^{-k}$, for all $k\in \mathbb N$) which in particular cover Ehrenfest type times. Thus, the problem of the construction of an asymptotic expansion in powers of $\h$ for $Q(t)$ is reduced to the construction of an asymptotic expansion for each block $Q_{\nu}(t)$ defined by \eqref{les blocks Qnu}. This will be the object of the following paragraph.  
\end{remark}

\subsection{Formal asymptotic expansion for $Q_{\nu}(t)$}\label{formal sc nu}
From now on $\nu$ will be fixed in $\{1,...,l\}$. We introduce the following notations for the symbols of the operators $P_{\nu}^w H^w P_{\nu}^w$ and $P_{\nu}^w Q^w P_{\nu}^w$ respectively,
\begin{equation}\label{le symbole Hnu}
H_{\nu}:=P_{\nu}\#H \#P_{\nu} \sim \sum_{j\geq 0} \h^j H_{\nu,j}
\end{equation}
\begin{equation}\label{le symbole Qnu}
Q_{\nu}:= P_{\nu} \# Q \# P_{\nu}\sim \sum_{j\geq 0} \h^j Q_{\nu,j}.
\end{equation}

Recall that by definition of $Q_{\nu}(t)$ (see \eqref{les blocks Qnu}) and the fact that $\big(P_{\nu}^w\big)^j = P_{\nu}^w$, $\forall j\in \mathbb N$ (according to \eqref{tildediag}), we have
$$
(P_{\nu}^w \big)^j Q_{\nu}(t) \big(P_{\nu}^w \big)^j = Q_{\nu}(t), \quad \forall j\in \mathbb N, \forall t\in \mathbb R.
$$ 
In the following this property will play an important role in the construction of an asymptotic expansion in powers of $\h$ for $Q_{\nu}(t)$.
%In the previous paragraph, we proved that the study of the time evolution of $Q^w$ under the unitary group $U_H(t):=e^{-\frac{it}{\h}H^w}$ is reduced modulo $\mathcal{O}_{\mathcal{L}(L^2)}(\h^{\infty})$ uniformly in large time intervals which recover the Ehrenfest time, to the study of the time evolution of the block diagonal observables $Q_{\nu}^w$ under $U_{\nu}(t):=e^{-\frac{it}{\h}H_{\nu}^w}$. 

The starting point is the following Heisenberg problem 
\begin{equation}\label{Heisenberg nu}
\left\{   
\begin{array}{rcl}
\displaystyle{\frac{d}{dt}} Q_{\nu}(t)&=& \frac{i}{\h}[H_{\nu}^w,Q_{\nu}(t)], \quad (t \in \mathbb R)\\ 
Q_{\nu}(t)_{|t=0}&=&Q_{\nu}^w(x,\h D_x;\h),
\end{array} 
\right.
\end{equation}
which we rewrite at the level of symbols as 
\begin{equation}\label{hix}
\left\{   
\begin{array}{rcl}
\displaystyle{\frac{d}{dt}} q_{\nu}(t)&=& \{H_{\nu},q_{\nu}(t)\}^*, \quad (t \in \mathbb R)\\ 
q_{\nu}(t)_{|t=0}&=&Q_{\nu}.
\end{array} 
\right.
\end{equation}

As in section \ref{Sc}, considering $q_{\nu}(t)$ as a formal power series of $\h$ of the form $q_{\nu}(t)\sim \sum_{j\geq 0}\h^j q_{\nu,j}(t)$ and then equating equal powers of $\h$ in both sides of \eqref{hix}, we derive the following Cauchy problems  

\begin{equation} (\mathcal{C}_{\nu,j})
\left\{   
\begin{array}{rcl}
\displaystyle{\frac{d}{dt}} q_{\nu,j}(t)&=& \{H_{\nu}, q_{\nu}(t)\}^*_j \\ 
q_{\nu,j}(t)_{|t=0}&=&Q_{\nu,j},
\end{array} 
\right.
\end{equation}
where we recall that $\{H_{\nu}, q_{\nu}(t)\}^*_j = i \big([ H_{\nu}, q_{\nu}(t)]_{\#}\big)_{j+1}$ denotes the coefficient of $\h^j$ in the asymptotic expansion of the Moyal bracket $\{H_{\nu},q_{\nu}(t)\}^*$.

Our objective consists in looking for a solution of \eqref{hix} of the form 
\begin{equation}\label{forme de la solution}
q_{\nu}(t) \sim P_{\nu}\# \sum_{k\geq 0} \h^k \tilde{q}_{\nu,k}(t) \#P_{\nu}.
\end{equation}

%\begin{remark}
%\textcolor{blue}{The main steps in the resolution of $(\mathcal{C}_{\nu,j})_{j\geq 0}$ are the following. Firstly, we fix the initial conditions $\tilde{q}_{\nu,j}(t)_{|t=0}$ compatible with \eqref{forme de la solution} by constructing a complete asymptotic expansion in powers of $\h$ for $q_{\nu}(t)_{|t=0}=Q_{\nu}$ of the form $ Q_{\nu}\sim P_{\nu}\# \sum_{k\geq 0} \h^k \tilde{Q}_{\nu,k} \#P_{\nu}$ (see Lemma \ref{conditions initiales} below). Then, using the general form of the solution \eqref{forme de la solution}, we derive recursive problems for the $\tilde{q}_{\nu,j}(t)$. For all $j\geq 0$, from \eqref{forme de la solution}, we have 
%$$
%q_{\nu,j}(t) = P_{\nu,0} \tilde{q}_{\nu,j}(t) P_{\nu,0} + P_{\nu}\# \sum_{k=0}^{j-1} \h^k \tilde{q}_{\nu,k}(t) \# P_{\nu}
%$$
%Then, the equations on the symbols $\tilde{q}_{\nu,j}(t)$ will be of the form 
%$$
%\frac{d}{dt} \big( P_{\nu,0} \tilde{q}_{\nu,j}(t) P_{\nu,0} \big) = F\big(  P_{\nu,0} \tilde{q}_{\nu,j}(t) P_{\nu,0} \big) + \mathcal{R}_{\nu,j-1}(t),
%$$
%where $\mathcal{R}_{\nu,j-1}(t)$ is a remainder term depending only on the preceding symbols $\tilde{q}_{\nu,k}(t)$ with $k\in \{0,...,j-1\}$. } 
%\end{remark}

More explicitly, using this general form of the solution, we shall derive recursive problems for the $\tilde{q}_{\nu,j}(t)$. Once these problems are derived and solved, the solution $q_{\nu,j}(t)$ of $(\mathcal{C}_{\nu,j})$ can then be computed using the composition formula \eqref{développement asymptotique} from the general formula 
\begin{equation}\label{general formula for the solution}
q_{\nu,j}(t,x,\xi) = \bigg( P_{\nu}\#  \sum_{k\geq 0} \h^k \tilde{q}_{\nu,k}(t) \#P_{\nu}  \bigg)_j(x,\xi)= \bigg( P_{\nu}\#  \sum_{k=0}^j \h^k \tilde{q}_{\nu,k}(t) \#P_{\nu}  \bigg)_j(x,\xi) , \quad \forall j\geq 0.
\end{equation}
In particular, 
\begin{equation}\label{forme du symbole principal}
q_{\nu,0}(t,x,\xi) = P_{\nu,0}(x,\xi) \tilde{q}_{\nu,0}(t,x,\xi) P_{\nu,0}(x,\xi).
\end{equation}

%\textcolor{blue}{\begin{remark}
%As we will see below, due to \eqref{forme de la solution}, the equation that we shall obtain on the $\tilde{q}_{\nu,j}(t)$ will be of the form
%$$
%\frac{d}{dt} \big( P_{\nu,0} \tilde{q}_{\nu,j}(t) P_{\nu,0} \big) = F\big(  P_{\nu,0} \tilde{q}_{\nu,j}(t) P_{\nu,0} \big) + \mathcal{R}_{j-1}(t),
%$$
%where $\mathcal{R}_{j-1}(t)$ is a remainder term depending only on the preceding symbols $\tilde{q}_{\nu,k}(t)$ with $k\in \{0,...,j-1\}$. The main idea in the resolution of equation of this type consists in 
%\end{remark}}

Let us start by fixing the initial conditions $\tilde{q}_{\nu,j}(t)_{|t=0}$.
\begin{lemma}\label{Lemma initial condition}
There exists a sequence of symbols $(\tilde{Q}_{\nu,k})_{k\geq 0}$ in $S(1)$ such that 
$$
Q_{\nu} \sim P_{\nu}\# \sum_{ k\geq 0} \h^k \tilde{Q}_{\nu,k} \# P_{\nu}
$$
and 
\begin{equation}\label{kix mak}
P_{\nu,0}\tilde{Q}_{\nu,k} P_{\nu,0} = \tilde{Q}_{\nu,k}, \quad \forall k\geq 0.
\end{equation}
In particular, $\tilde{Q}_{\nu,0} = Q_{\nu,0} = P_{\nu,0} Q_0 P_{\nu,0}$.
\end{lemma}
\begin{proof}
Using the fact that $P_{\nu}\#P_{\nu} = P_{\nu}$ according to \eqref{tildediag}, we have 
\begin{eqnarray*}
Q_{\nu}&=& P_{\nu}\# Q \# P_{\nu}\\ 
&=& P_{\nu}\#P_{\nu}\# Q\# P_{\nu} \#P_{\nu} \\
&\sim& P_{\nu}\# (P_{\nu,0}Q_{0}P_{\nu,0}) \# P_{\nu} +  P_{\nu}\#\big(  \sum_{j\geq 1}\h^{j} Q_{\nu,j} \big) \# P_{\nu}.
\end{eqnarray*}
Put 
$$
\tilde{Q}_{\nu,0}:= P_{\nu,0} Q_{0} P_{\nu,0} = Q_{\nu,0} \quad \text{and} \quad  R_{\nu,0}:= P_{\nu}\#\big(  \sum_{j\geq 1}\h^{j} Q_{\nu,j} \big) \# P_{\nu} \sim \sum_{j\geq 1}\h^j (R_{\nu,0})_j.
$$
We have 
\begin{eqnarray*}
R_{\nu,0} &= & P_{\nu}\# P_{\nu}\#R_{\nu,0}\#P_{\nu} \# P_{\nu}\\
&\sim & \h\; P_{\nu}\# P_{\nu}\# (R_{\nu,0})_1 \# P_{\nu}\# P_{\nu} + P_{\nu} \# \sum_{j\geq 2} \h^j (R_{\nu,0})_j \# P_{\nu} \\
&\sim & \h \; P_{\nu}\# (P_{\nu,0} (R_{\nu,0})_1 P_{\nu,0}) \# P_{\nu} + R_{\nu,1} \\
\end{eqnarray*}
where $R_{\nu,1} := P_{\nu}\# \sum_{j\geq 2}\h^j \bigg((P_{\nu}\# (R_{\nu,0})_1 \#P_{\nu})_{j-1} + (R_{\nu,0})_j\bigg) \# P_{\nu}$. We define 
$$
\tilde{Q}_{\nu,1}:= P_{\nu,0} (R_{\nu,0})_1 P_{\nu,0} .
$$ 
One can iterate this procedure using at each step the fact that $ R_{\nu,j} = P_{\nu} \# P_{\nu}\# R_{\nu,j} \#P_{\nu}\# P_{\nu}$ to construct the symbols $\tilde{Q}_{\nu,k}$ satisfying the property \eqref{kix mak}. The constructed symbols $\tilde{Q}_{\nu,k}$ are clearly in $S(1)$ since $P_{\nu}, Q\in S_{\text{sc}}(1)$.

\end{proof}
In view of \eqref{forme de la solution} and the above lemma, it is thus natural to impose the following initial conditions for the $\tilde{q}_{\nu,j}(t)$
\begin{equation}\label{conditions initiales}
\tilde{q}_{\nu,j}(t)_{|t=0} = \tilde{Q}_{\nu,j}, \quad \forall j\geq 0.
\end{equation}

%in the resolution of $(\mathcal{C}_{\nu,j})$ are the following. For $j\geq 0$, we start by expressing $q_{\nu,j}(t)$ and $\{H_{\nu}, q_{\nu}(t)\}^*_j$ with respect to $\tilde{q}_{\nu,j}(t)$ using the form of the solution \eqref{forme de la solution}.

%The resolution of $(\mathcal{C}_{\nu,j})$ is as follows. We shall derive the equations on the symbol $\tilde{q}_{\nu,j}(t)$ arising from $(\mathcal{C}_{\nu,j})$ and solve them. Once the symbols $(\tilde{q}_{\nu,j}(t))_{k\geq 0}$ are obtained, the solution $q_{\nu,j}(t)$ of $(\mathcal{C}_{\nu,j})$ can then be computed using the composition formula \eqref{développement asymptotique} from the general formula 
%\begin{equation}\label{general formula for the solution}
%q_{\nu,j}(t,x,\xi) = \bigg( P_{\nu}\#  \sum_{k\geq 0} \h^k \tilde{q}_{\nu,k}(t) \#P_{\nu}  \bigg)_j(x,\xi)= \bigg( P_{\nu}\#  \sum_{k=0}^j \h^k \tilde{q}_{\nu,k}(t) \#P_{\nu}  \bigg)_j(x,\xi) , \quad \forall j\geq 0.
%\end{equation}
%In particular, 
%$$
%q_{\nu,0}(t,x,\xi) = P_{\nu,0}(x,\xi) \tilde{q}_{\nu,0}(t,x,\xi) P_{\nu,0}(x,\xi).
%$$

%Let us derive the equations on the symbols $\tilde{q}_{\nu,j}(t)$ resulting from the Cauchy problems $(\mathcal{C}_{\nu,j})_{j\geq 0}$.
\noindent
Now, to derive the equations on the $\tilde{q}_{\nu,j}(t)$ arising from the Cauchy problems $(\mathcal{C}_{\nu,j})_{j\geq 0}$, we express $q_{\nu,j}(t)$ and $\{H_{\nu},q_{\nu}(t)\}^{*}_j$ with respect to $\tilde{q}_{\nu,j}(t)$. For $j\geq 0$, we define
\begin{equation}\label{Anuj-1}
A_{\nu,j-1}(t) := P_{\nu} \# \sum_{k=0}^{j-1} \h^k \tilde{q}_{\nu,k}(t) \# P_{\nu}
\end{equation}
with the convention $A_{\nu,-1}(t)=0$. From \eqref{forme de la solution}, we clearly have 
\begin{equation}\label{transf1}
q_{\nu,j}(t) = P_{\nu,0} \tilde{q}_{\nu,j}(t) P_{\nu,0} + \big( A_{\nu,j-1}(t) \big)_j,
\end{equation}
where $\big( A_{\nu,j-1}(t) \big)_j$ denotes the coefficient of $\h^j$ in the asymptotic expansion of $A_{\nu,j-1}(t)$. On the other hand, we have
\begin{equation}\label{transf2}
\bigg([H_{\nu},q_{\nu}(t)]_{\#}\bigg)_{j+1} = [H_{\nu,0}, P_{\nu,0} \tilde{q}_{\nu,j+1}(t) P_{\nu,0}] + \bigg([H_{\nu}, P_{\nu}\# \tilde{q}_{\nu,j}(t) \# P_{\nu}]_{\#}\bigg)_1 + \bigg( [H_{\nu}, A_{\nu,j-1}(t)]_{\#}  \bigg)_{j+1}.
\end{equation}
The first term in the right hand side of the above equation vanishes since $H_{\nu,0}= \lambda_{\nu} P_{\nu,0}$. Then, putting together \eqref{transf1} and \eqref{transf2}, we deduce the equation on the symbol $\tilde{q}_{\nu,j}(t)$ arising from $(\mathcal{C}_{\nu,j})$ which reads
\begin{equation}\label{equation on tilde j}
\frac{d}{dt}  P_{\nu,0} \tilde{q}_{\nu,j}(t) P_{\nu,0} =  i \bigg([H_{\nu}, P_{\nu}\# \tilde{q}_{\nu,j}(t) \# P_{\nu}]_{\#}\bigg)_1 + K_{\nu,j-1}(t),
\end{equation}
where 
\begin{equation}\label{Knuj}
K_{\nu,j-1}(t) := i \bigg( [H_{\nu}, A_{\nu,j-1}(t)]_{\#}  \bigg)_{j+1} - \frac{d}{dt} \big( A_{\nu,j-1}(t) \big)_j.
\end{equation}
\noindent
%Notice that $K_{\nu,j-1}(t)$ depends only on the symbols $\tilde{q}_{\nu,k}(t)$ with $0\leq k \leq j-1$. ($K_{\nu,-1}(t)=0$).

Taking into account the initial conditions \eqref{conditions initiales}, we get the following Cauchy problems for $\tilde{q}_{\nu,j}(t)$
\begin{equation} \label{pbm Cauchy tildej}
\left\{   
\begin{array}{rcl}
\displaystyle{\frac{d}{dt} }   P_{\nu,0} \tilde{q}_{\nu,j}(t) P_{\nu,0} &=&  i \bigg([H_{\nu}, P_{\nu}\# \tilde{q}_{\nu,j}(t) \# P_{\nu}]_{\#}\bigg)_1 + K_{\nu,j-1}(t) \\ 
\tilde{q}_{\nu,j}(t)_{|t=0}&=& \tilde{Q}_{\nu,j}. 
\end{array} 
\right.
\end{equation}

Notice that $K_{\nu,j-1}(t)$ depends only on the symbols $\tilde{q}_{\nu,k}(t)$ with $0\leq k \leq j-1$. ($K_{\nu,-1}(t)=0$). 

%\textcolor{blue}{In the following proposition, we transform \eqref{pbm Cauchy tildej} in a suitable form that will be the key point in the recursive construction of the $\tilde{q}_{\nu,j}(t)$.}

%The following proposition is the first step in the resolution of \eqref{pbm Cauchy tildej}.

%The following proposition is the main step in the construction of the symbols $\tilde{q}_{\nu,j}(t)$, $j\geq 0$.
\begin{proposition}\label{main proposition construction}
Let $j\in \mathbb N$. The Cauchy problem \eqref{pbm Cauchy tildej} is equivalent to the following one 
\begin{equation}\label{main construction transformation}
\left\{   
\begin{array}{rcl}
\displaystyle{\frac{d}{dt} }  P_{\nu,0} \tilde{q}_{\nu,j}(t) P_{\nu,0} &=&  \{\lambda_{\nu}, P_{\nu,0} \tilde{q}_{\nu,j}(t) P_{\nu,0}\} + i [\tilde{H}_{\nu,1}, P_{\nu,0} \tilde{q}_{\nu,j}(t) P_{\nu,0}] + K_{\nu,j-1}(t) \\ 
\tilde{q}_{\nu,j}(t)_{|t=0}&=& \tilde{Q}_{\nu,j},
\end{array} 
\right.
\end{equation}

where $\tilde{H}_{\nu,1}$ is the $(m\times m)$ hermitian-valued function defined by
\begin{equation}\label{equ79}
\tilde{H}_{\nu,1} := \frac{\lambda_{\nu}}{2i} P_{\nu,0} \{P_{\nu,0},P_{\nu,0}\} P_{\nu,0} -i \big[ P_{\nu,0}, \{\lambda_{\nu},P_{\nu,0}\} \big] +  P_{\nu,0} H_{\nu,1} P_{\nu,0}.
\end{equation}
%Furthermore, if $\tilde{q}_{\nu,j}(t)$ is a solution of the following problem 
%\begin{equation} \label{Cauchy problems remove}
%\left\{   
%\begin{array}{rcl}
%\displaystyle{\frac{d}{dt} }  \tilde{q}_{\nu,j}(t) &=&  \{\lambda_{\nu}, \tilde{q}_{\nu,j}(t)\} + i [\tilde{H}_{\nu,1}, \tilde{q}_{\nu,j}(t) ] + K_{\nu,j-1}(t) \\ 
%\tilde{q}_{\nu,j}(t)_{|t=0}&=& \tilde{Q}_{\nu,j},
%\end{array} 
%\right.
%\end{equation}
%then 
%\begin{equation}\label{en tt temps}
%\tilde{q}_{\nu,j}(t) = P_{\nu,0} \tilde{q}_{\nu,j}(t) P_{\nu,0}, \quad \forall t\in \mathbb R.
%\end{equation}
\end{proposition}

To prove this proposition, we recall the following result from the appendix of \cite{Spo}.

\begin{lemma}\label{Lemma Spohn}
Let $W : \mathbb R^{2n} \rightarrow M_m(\mathbb C)$ be such that $[W, P_{\nu,0}] = 0$. We have 
\begin{align*}
\frac{1}{2} P_{\nu,0}\bigg( \{ \lambda_{\nu}P_{\nu,0}, W\} - \{ W , \lambda_{\nu}P_{\nu,0}\}   \bigg) P_{\nu,0} &= \big\{ \lambda_{\nu}, P_{\nu,0} W P_{\nu,0}  \big\} \\
& -  \bigg[ P_{\nu,0} W P_{\nu,0}, \frac{\lambda_{\nu}}{2} P_{\nu,0} \{P_{\nu,0},P_{\nu,0}\} P_{\nu,0} + \big[ P_{\nu,0}, \{\lambda_{\nu},P_{\nu,0}\} \big]   \bigg].
\end{align*}
\end{lemma}

\noindent
\textbf{Proof of proposition \eqref{main proposition construction} :}

Let us start by computing $\bigg([H_{\nu}, P_{\nu}\# \tilde{q}_{\nu,j}(t) \# P_{\nu}]_{\#}\bigg)_1$. We have 
\begin{equation}\label{égalité symb sub}
P_{\nu,0} \bigg([H_{\nu}, P_{\nu}\# \tilde{q}_{\nu,j}(t) \# P_{\nu}]_{\#}\bigg)_1 P_{\nu,0} = \bigg([H_{\nu}, P_{\nu}\# \tilde{q}_{\nu,j}(t) \# P_{\nu}]_{\#}\bigg)_1.
\end{equation}
Indeed, the fact that $P_{\nu}\# H_{\nu} = H_{\nu}\# P_{\nu}=H_{\nu}$ (according to \eqref{tildediag}) implies
\begin{equation}\label{chmopu}
P_{\nu} \# [H_{\nu}, P_{\nu}\# \tilde{q}_{\nu,j}(t) \# P_{\nu}]_{\#} \# P_{\nu} = [H_{\nu}, P_{\nu}\# \tilde{q}_{\nu,j}(t) \# P_{\nu}]_{\#}.
\end{equation}
Consequently, the sub-principal symbols of the two terms in the above equation coincide. Since 
$$
\bigg([H_{\nu}, P_{\nu}\# \tilde{q}_{\nu,j}(t) \# P_{\nu}]_{\#} \bigg)_0 =  [H_{\nu,0}, P_{\nu,0} \tilde{q}_{\nu,j}(t) P_{\nu,0}] = 0,
$$
it follows that the sub-principal symbol of the left hand side of \eqref{chmopu} is equal to 
$$
P_{\nu,0} \bigg([H_{\nu}, P_{\nu}\# \tilde{q}_{\nu,j}(t) \# P_{\nu}]_{\#}\bigg)_1 P_{\nu,0}.
$$ 
Thus we get \eqref{égalité symb sub}. Using this property and formulas \eqref{deduc1} and \eqref{deduc2}, we obtain 
\begin{align}\label{chedjou1}
\bigg([H_{\nu}, P_{\nu}\# \tilde{q}_{\nu,j}(t) \# P_{\nu}]_{\#}\bigg)_1 &= P_{\nu,0}  \bigg([H_{\nu}, P_{\nu}\# \tilde{q}_{\nu,j}(t) \# P_{\nu}]_{\#}\bigg)_1 P_{\nu,0} \nonumber \\
&=  \frac{1}{2i}  P_{\nu,0} \bigg( \{H_{\nu,0}, P_{\nu,0}\tilde{q}_{\nu,j}(t) P_{\nu,0}\} - \{P_{\nu,0}\tilde{q}_{\nu,j}(t) P_{\nu,0} , H_{\nu,0}\}   \bigg)  P_{\nu,0} \nonumber \\
& + P_{\nu,0} \bigg( \big[ H_{\nu,0}, \big( P_{\nu}\# \tilde{q}_{\nu,j}(t) \# P_{\nu}\big)_1 \big] + \big[ H_{\nu,1}, P_{\nu,0} \tilde{q}_{\nu,j}(t) P_{\nu,0} \big] \bigg) P_{\nu,0}  \nonumber \\
&= \frac{1}{2i} P_{\nu,0}\bigg( \{H_{\nu,0},  P_{\nu,0}\tilde{q}_{\nu,j}(t) P_{\nu,0}\} - \{P_{\nu,0}\tilde{q}_{\nu,j}(t) P_{\nu,0} , H_{\nu,0}\}   \bigg)  P_{\nu,0} \nonumber  \nonumber \\
& +   \big[ P_{\nu,0} H_{\nu,1} P_{\nu,0}, P_{\nu,0} \tilde{q}_{\nu,j}(t) P_{\nu,0} \big],
\end{align}
where in the last step we used the fact that $P_{\nu,0}  \big[ H_{\nu,0}, \big( P_{\nu}\# \tilde{q}_{\nu,j}(t) \# P_{\nu}\big)_1 \big] P_{\nu,0} =0$ which can be easily verified using formula \eqref{deduc2}. Applying Lemma \ref{Lemma Spohn} with $W = P_{\nu,0} \tilde{q}_{\nu,j}(t) P_{\nu,0}$, we get  
\begin{align*}
i \bigg([H_{\nu}, P_{\nu}\# \tilde{q}_{\nu,j}(t) \# P_{\nu}]_{\#}\bigg)_1 =  \{\lambda_{\nu}, P_{\nu,0} \tilde{q}_{\nu,j}(t) P_{\nu,0}\} + i [\tilde{H}_{\nu,1}, P_{\nu,0} \tilde{q}_{\nu,j}(t) P_{\nu,0}],
\end{align*}
where $\tilde{H}_{\nu,1}$ is defined by \eqref{equ79}. This ends the proof of the proposition.
\begin{flushright}
$\square$
\end{flushright}

The resolution of the Cauchy problems \eqref{main construction transformation} will be made by induction on $j\geq 0$. Let us start with $j=0$. Since $K_{\nu,-1}(t)=0$, we have 
\begin{equation*}
\left\{   
\begin{array}{rcl}
\displaystyle{\frac{d}{dt} }  P_{\nu,0} \tilde{q}_{\nu,0}(t) P_{\nu,0} &=&  \{\lambda_{\nu}, P_{\nu,0} \tilde{q}_{\nu,0}(t) P_{\nu,0}\} + i [\tilde{H}_{\nu,1}, P_{\nu,0} \tilde{q}_{\nu,0}(t) P_{\nu,0}]  \\ 
\tilde{q}_{\nu,0}(t)_{|t=0}&=& \tilde{Q}_{\nu,0}.
\end{array} 
\right.
\end{equation*}  

In Lemma \ref{Lemma remove}, taking into account the fact that $\tilde{Q}_{\nu,0} = P_{\nu,0} \tilde{Q}_{\nu,0} P_{\nu,0}$ (see \eqref{kix mak}), we have shown that if $\tilde{q}_{\nu,0}(t)$ is a solution of the following problem
\begin{equation}\label{main construction transformation0}
\left\{   
\begin{array}{rcl}
\displaystyle{\frac{d}{dt} } \tilde{q}_{\nu,0}(t) &=&  \{\lambda_{\nu}, \tilde{q}_{\nu,0}(t)\} + i [\tilde{H}_{\nu,1}, \tilde{q}_{\nu,0}(t) ]  \\ 
\tilde{q}_{\nu,0}(t)_{|t=0}&=& \tilde{Q}_{\nu,0},
\end{array} 
\right.
\end{equation} 
then at any time $t$,
\begin{equation}\label{en tt temps}
\tilde{q}_{\nu,0}(t) = P_{\nu,0} \tilde{q}_{\nu,0}(t) P_{\nu,0}.
\end{equation}
Applying the result of Appendix \ref{Ann Cauchy problem general} with $\Lambda= \lambda_{\nu}$ and $A= \tilde{H}_{\nu,1}$, we obtain the solution of \eqref{main construction transformation0} which reads  
\begin{equation}\label{sol tilde0}
\tilde{q}_{\nu,0}(t,x,\xi)= T_{\nu}^{-1}(t,x,\xi) Q_{\nu,0}\big(\phi_{\nu}^t(x,\xi) \big) T_{\nu}(t,x,\xi),
\end{equation}
where $T_{\nu}$ in the unitary $(m\times m)$ matrix-valued function solution of the system
\begin{equation}\label{equation de transport pour Tnu}
\frac{d}{dt} T_{\nu}(t,x,\xi) = - i \tilde{H}_{\nu,1} \big(\phi_{\nu}^t(x,\xi)\big) T_{\nu}(t,x,\xi), \quad T_{\nu}(0,x,\xi) = I_m.
\end{equation}

Let us now assume that we have solved \eqref{main construction transformation} until the order $j-1$, i.e. we have constructed the symbols $\tilde{q}_{\nu,k}(t)$ for $k\in \{0,...,j-1\}$ and that they satisfy 
$$
\tilde{q}_{\nu,k}(t) = P_{\nu,0} \tilde{q}_{\nu,k}(t) P_{\nu,0}, \quad \forall k\in \{0,...,j-1\}.
$$
We are going to solve \eqref{main construction transformation} at the order $j$ and check that the solution $\tilde{q}_{\nu,j}(t)$ satisfies 
$$
\tilde{q}_{\nu,j}(t) = P_{\nu,0} \tilde{q}_{\nu,j}(t) P_{\nu,0}.
$$ 
To apply Lemma \ref{Lemma remove}, we have to prove that 
\begin{equation}\label{kix mak1}
P_{\nu,0} K_{\nu,j-1}(t) P_{\nu,0} = K_{\nu,j-1}(t).
\end{equation}
Recall that $K_{\nu,j-1}(t)$ defined by \eqref{Knuj} is the $j$-th term (i.e. the coefficient of $\h^j$) of the symbol 
$$
E_{\nu,j-1}(t) := \frac{i}{\h} [H_{\nu}, A_{\nu,j-1}(t)]_{\#} - \frac{d}{dt} A_{\nu,j-1}(t).
$$
In the following, we say that $B\sim \sum_{k\geq 0}\h^k B_k$ belongs to $S(\h^j)$ if $B_{k}=0$ for all $k < j$.

We claim that 
\begin{equation}\label{main assumption construction}
E_{\nu,j-1}(t) \in S(\h^j).
\end{equation}
This will be proven below. Due to \eqref{Anuj-1} and \eqref{le symbole Hnu}, we have 
$$
P_{\nu} \#  E_{\nu,j-1}(t) \# P_{\nu} =  E_{\nu,j-1}(t).
$$
By equating the $j$-th terms in both sides using \eqref{main assumption construction} we get \eqref{kix mak1}. Taking into account \eqref{kix mak} and \eqref{kix mak1}, according to Lemma \ref{Lemma remove}, if $\tilde{q}_{\nu,j}(t)$ is a solution of the following problem
\begin{equation}\label{main construction transformation j}
\left\{   
\begin{array}{rcl}
\displaystyle{\frac{d}{dt} }  \tilde{q}_{\nu,j}(t)  &=&  \{\lambda_{\nu}, \tilde{q}_{\nu,j}(t) \} + i [\tilde{H}_{\nu,1}, \tilde{q}_{\nu,j}(t)] + K_{\nu,j-1}(t) \\ 
\tilde{q}_{\nu,j}(t)_{|t=0}&=& \tilde{Q}_{\nu,j},
\end{array} 
\right.
\end{equation}
then 
$$
\tilde{q}_{\nu,j}(t) = P_{\nu,0} \tilde{q}_{\nu,j}(t) P_{\nu,0}, \quad \forall t\in \mathbb R.
$$
To solve \eqref{main construction transformation j}, we apply the result of Appendix \ref{Ann Cauchy problem general} again with $\Lambda= \lambda_{\nu}$, $A= \tilde{H}_{\nu,1}$ and $B(t) = K_{\nu,j-1}(t)$. The solution reads
\begin{align}\label{les solutions tilde}
\tilde{q}_{\nu,j}(t,x,\xi) = T_{\nu}^{-1}(t,x,\xi) \bigg( \tilde{Q}_{\nu,j}\big( \phi_{\nu}^t(x,\xi)\big) + \int_0^t  W_{\nu,j}(t,s,x,\xi)  ds \bigg) T_{\nu}(t,x,\xi),
\end{align}
with 
$$
W_{\nu,j}(t,s,x,\xi):= T_{\nu}^{-1} \big( -s,\phi_{\nu}^{t}(x,\xi) \big) K_{\nu,j-1}\big(s,\phi_{\nu}^{t-s}(x,\xi) \big) T_{\nu}\big(-s,\phi_{\nu}^t(x,\xi) \big),
$$
where $T_{\nu}$ is given by the system \eqref{equation de transport pour Tnu}.

It remains now to prove the claim \eqref{main assumption construction} by induction on $j$. For $j=1$, we have 
\begin{eqnarray*}
\big(E_{\nu,0}(t)\big)_0 &=& i \big([H_{\nu}, A_{\nu,0}(t)]_{\#}\big)_1 - \frac{d}{dt} \big(A_{\nu,0}(t)\big)_0 \\
&=& i \big([H_{\nu}, P_{\nu}\# \tilde{q}_{\nu,0}(t) \# P_{\nu}]_{\#} \big)_1 - \frac{d}{dt} P_{\nu,0} \tilde{q}_{\nu,0}(t) P_{\nu,0} \\
&=& 0,
\end{eqnarray*}
since it is the equation satisfied by $\tilde{q}_{\nu,0}(t)$ (see \eqref{pbm Cauchy tildej}). Thus $E_{\nu,0}(t)\in S(\h)$. 

We assume that $E_{\nu,j-2}(t)\in S(\h^{j-1})$ and let us prove \eqref{main assumption construction}. Using that 
$$
A_{\nu,j-1}(t) = A_{\nu,j-2}(t) + \h^{j-1}\: P_{\nu}\# \tilde{q}_{\nu,j-1}(t) \# P_{\nu}
$$
 we get
\begin{align}
E_{\nu,j-1}(t) &= E_{\nu,j-2}(t) - \: \h^{j-1} \frac{d}{dt} P_{\nu}\# \tilde{q}_{\nu,j-1}(t) \# P_{\nu} + i \: \h^{j-2} \big[ H_{\nu}, P_{\nu}\# \tilde{q}_{\nu,j-1}(t)\# P_{\nu} \big]_{\#} \nonumber  \\
& = E_{\nu,j-2}(t) - \h^{j-1}\: \frac{d}{dt} P_{\nu,0} \tilde{q}_{\nu,j-1}(t) P_{\nu,0} + S(\h^j) + i \h^{j-1} \big( \big[ H_{\nu}, P_{\nu}\# \tilde{q}_{\nu,j-1}(t)\# P_{\nu} \big]_{\#}  \big)_1 + S(\h^j) \nonumber \\
& =  E_{\nu,j-2}(t) - \h^{j-1} \bigg(  \frac{d}{dt} P_{\nu,0} \tilde{q}_{\nu,j-1}(t) P_{\nu,0} - i  \big( \big[ H_{\nu}, P_{\nu}\# \tilde{q}_{\nu,j-1}(t)\# P_{\nu} \big]_{\#}  \big)_1 \bigg) + S(\h^j). \label{inx mari}
\end{align}
Notice that to pass from the first to the second equality, we have used the fact that 
$$
\big[ H_{\nu}, P_{\nu}\# \tilde{q}_{\nu,j-1}(t)\# P_{\nu} \big]_{\#} \in S(\h)
$$
since as it was already point out in \eqref{transf2} its principal symbol vanishes.

On the other hand, combining the definition of $K_{\nu,j-2}(t)$ which is $K_{\nu,j-2}(t) = \big(E_{\nu,j-2}(t) \big)_{j-1}$ and the induction hypothesis $E_{\nu,j-2}(t)\in S(\h^{j-1})$, we get 
$$
E_{\nu,j-2}(t) = \h^{j-1} K_{\nu,j-2}(t) + S(\h^j).
$$
By going back to \eqref{inx mari}, we obtain 
$$
E_{\nu,j-1}(t) = \h^{j-1} \bigg( K_{\nu,j-2}(t) -  \frac{d}{dt} P_{\nu,0} \tilde{q}_{\nu,j-1}(t) P_{\nu,0} + i  \big( \big[ H_{\nu}, P_{\nu}\# \tilde{q}_{\nu,j-1}(t)\# P_{\nu} \big]_{\#}  \big)_1 \bigg) + S(\h^j).
$$  
The first term in the right hand side of the above equation vanishes since it is exactly the equation satisfied by $\tilde{q}_{\nu,j-1}(t)$ (see \eqref{pbm Cauchy tildej}). Thus, we proved that $E_{\nu,j-1}(t)\in S(\h^j)$. This ends the proof of the claim.
\begin{flushright}
$\square$
\end{flushright}

%But according to the equation satisfied by $\tilde{q}_{\nu,j-1}(t)$ (see \eqref{pbm Cauchy tildej}), the first term in the right hand side of the above equation vanishes. Thus, we proved that $E_{\nu,j-1}(t)\in S(\h^j)$.
 Summing up, we hence have solved the Cauchy problems \eqref{pbm Cauchy tildej} for all $j\geq 0$. The solutions $(\tilde{q}_{\nu,j}(t))_{j\geq 0}$ are given by formula \eqref{les solutions tilde}. In particular, $\tilde{q}_{\nu,0}(t)$ is given by \eqref{sol tilde0}. As already mentioned in the begining of this paragraph, the solutions $q_{\nu,j}(t)$ of the Cauchy problems $(\mathcal{C}_{\nu,j})_{j\geq 0}$ can then be computed  using the composition formula \eqref{développement asymptotique} from the general formula \eqref{general formula for the solution}. In particular, the principal symbol $q_{\nu,0}(t)$ is given by \eqref{forme du symbole principal}.

\subsection{Uniform estimates and proofs of Theorem \ref{main T} and Corollary \ref{sec coro}}\label{unpoi}

This section is devoted to the proofs of Theorem \ref{main T} and Corollary \ref{sec coro}. Since the techniques of the proofs are close to those used in the above section, we shall omit some details. 

As in section \ref{Sc}, we start by estimating the derivatives of the constructed symbols $q_{\nu,j}(t)$, $j\geq 0$.

\begin{proposition}\label{est Gen}
Assume \textbf{(A1)} and \textbf{(A2)} and let $1\leq \nu \leq l$. For all $\gamma\in \mathbb N^{2n}$, for all $j\geq 0$, there exists $C_{\gamma,\nu,j}>0$ such that for all $t\in \mathbb R$ and all $(x,\xi)\in \mathbb R^{2n}$, we have 
\begin{equation}\label{AG0}
\big\Vert \partial_{(x,\xi)}^{\gamma} q_{\nu,0}(t,x,\xi) \big\Vert \leq C_{\gamma,\nu,0} \exp \bigg( |\gamma|\Gamma_{\nu}|t| \bigg),
\end{equation}
and for $j\geq 1$,
\begin{equation}\label{AGj}
\big\Vert  \partial_{(x,\xi)}^{\gamma} q_{\nu,j}(t,x,\xi) \big\Vert \leq C_{\gamma,\nu,j} \exp\bigg( \big(2|\gamma|+4j-2\big)\Gamma_{\nu}|t|\bigg),
\end{equation}
where $\Gamma_{\nu}$ is defined by \eqref{upper bounds}.
\end{proposition}

%\begin{remark}
%\begin{itemize}
%\item[(i)] Here, we recover in particluar the result of \cite[Theorem 3.2]{bolte} which states that the operator $Q(t)$ in an $\h$-pseudodifferential operator with symbol $q(t,x,\xi;\h)\sim \sum_{j\geq 0} \h^j q_{j}(t,x,\xi)$ in $S(1;\mathbb R^{2n},M_m(\mathbb C))$, uniformly for $t$ in any bounded interval. This is an immediate consequence of estimate \eqref{semmm} and Beals's characterization of $\h$-pseudodifferential operators.

%\item[(ii)] In \cite[proposition 3.4]{bolte}, it was shown that the class $\mathcal{Q}(1)$ of observables in $S(1;\mathbb R^{2n},M_m(\mathbb C))$ that are block-diagonal with respect to the semiclassical almost projectors $(P_{\nu})_{1\leq \nu \leq l}$ exhaust all symbols in $S(1;\mathbb R^{2n},M_m(\mathbb C))$ such that the corresponding Heisenberg observable $Q(t)$ is an $\h$-pseudodifferential operator with symbol in $S(1;\mathbb R^{2n},M_m(\mathbb C))$. More explicitly, let $H$ be a semiclassical Hamiltonian satisfying assumptions \textbf{(A0)}, \textbf{(A1)} and \textbf{(A2)}. Then, we have 
%$$
%\bigg( \forall |t|\leq T,\; Q(t) = q(t)^w(x,\h D_x;\h), \; \text{with}\; q(t)\sim \sum_{j\geq 0}\h^j q_j(t)\;\ \text{in} \;\;S(1;\mathbb R^{2n},M_m(\mathbb C)) \bigg) \quad \Longleftrightarrow \quad Q \in \mathcal{Q}(1).
%$$
%\end{itemize}
%\end{remark}
%%%%%%%%%%%%%%%%%%%%%%%%%%%%%%%%%%%%%%%%%%%%%%%%%%%%%%%%%%%%%%%%%%%%%%%%%%%%%%%%%%%%%%%%%%%%%%%%%%%%%%%%%%%%%%%%%%%%%%%%%%%%%%%%%%%%%%%%%%%%%%%%%%%%%%%%%%%%%%%%%%%%%%%%%%%%%%%%%%%%%%%%%%%%%%%%%%%%%%%%%%%%%%%%%%%%%%%%%%%%%%%%%%%%%%%%%%%%%%%%%%%%%%%%%%%%%%%%%%

Similarly to the proof of Proposition \ref{u5}, the proof of the above proposition is based on the following lemmas which give estimates on the derivatives of the Hamiltonian flows $\phi_{\nu}^t$ generated by the eigenvalues $\lambda_{\nu}$ and the matrix-valued function $T_{\nu}$ defined in (\ref{equation de transport pour Tnu}).

From now on we fix $\nu\in \{1,...,l\}$.

\begin{lemma}\label{lemma flownu}
We assume that 
\begin{equation}\label{cond flownu}
\partial^{\gamma}_{(x,\xi)}H_0 \in L^{\infty}(\mathbb R^{2n}), \quad \text{for}\;\; |\gamma|\geq 2.
\end{equation}
Then, for all $\gamma \in \mathbb N^{2n}\setminus \{0\}$, there exists $C_{\nu,\gamma}>0$ such that for all $t\in \mathbb R$ and all $(x,\xi)\in \mathbb R^{2n}$,
\begin{equation}\label{flowwnu}
\Vert \partial_{(x,\xi)}^{\gamma}\phi_{\nu}^t(x,\xi) \Vert \leq C_{\nu,\gamma} \exp\big( |\gamma|\Gamma_{\nu}|t|\big).
\end{equation}
\end{lemma}
\begin{proof}
According to inequaltiy \eqref{mm classe de symbole}, \eqref{cond flownu} implies that $\partial_{(x,\xi)}^{\gamma}\lambda_{\nu}\in L^{\infty}(\mathbb R^{2n})$, for $|\gamma|\geq 2$. Thus estimate \eqref{flowwnu} can be proved in the same manner as in Lemma \ref{flot1} (see \cite[Lemma 2.2]{Rob}).
\end{proof}
We turn now to the estimation of the derivatives of $T_{\nu}$ solution of the system \eqref{equation de transport pour Tnu}.
\begin{lemma}\label{lemma est transport numu}
Let assumptions \textbf{(A1)} and \textbf{(A2)} be satisfied. For all $\gamma\in \mathbb N^{2n}\setminus \{0\}$ there exists a constant $C_{\nu,\gamma}>0$ (independent of $t\in \mathbb R$ and $(x,\xi)\in \mathbb R^{2n}$) such that
\begin{equation}\label{est transport nu}
\Vert \partial_{(x,\xi)}^{\gamma} T_{\nu}(t,x,\xi) \Vert \leq C_{\nu,\gamma} \exp\big( |\gamma|\Gamma_{\nu}|t|\big).
\end{equation}
Furthermore, the same estimate holds for $T_{\nu}^{-1}(t,x,\xi)$.
\end{lemma}
\begin{proof} We recall the expression of the $(m\times m)$ hermitian-valued function $\tilde{H}_{\nu,1}$ defined in \eqref{equ79}
$$
\tilde{H}_{\nu,1}= P_{\nu,0} H_{\nu,1} P_{\nu,0} - i [P_{\nu,0},\{\lambda_{\nu},P_{\nu,0}\}] - \frac{i}{2} \lambda_{\nu} P_{\nu,0} \{P_{\nu,0},P_{\nu,0}\}P_{\nu,0} := I_{\nu}^{(1)} + I_{\nu}^{(2)}+ I_{\nu}^{(3)}.
$$
We claim that under assumptions \textbf{(A1)} and \textbf{(A2)}, we have $\tilde{H}_{\nu,1}\in S(1)$. Then, estimate \eqref{est transport nu} can be proved by applying exactly the same method as in the proof of Lemma \ref{Transport1}. 

To prove the claim let us start by computing $H_{\nu,1}$. From \eqref{tildediag1}, we have $H_{\nu,1}:=(P_{\nu}\# H \# P_{\nu})_1 = (P_{\nu}\# H)_1$. Then, using formula (\ref{principal and sub-principal symbol}), we obtain 
\begin{equation}\label{comutation1}
H_{\nu,1} = \frac{1}{2i} \{P_{\nu,0},H_0\} + P_{\nu,0}H_1 + P_{\nu,1} H_0.
\end{equation}

It follows that
$$
I_{\nu}^{(1)} = \frac{1}{2i} P_{\nu,0} \{P_{\nu,0},H_0\}P_{\nu,0} + P_{\nu,0}H_1 P_{\nu,0} + \lambda_{\nu}P_{\nu,0} P_{\nu,1} P_{\nu,0}.
$$
Computing $P_{\nu,1}$ using formula (\ref{principal and sub-principal symbol}) and multiplying from both sides by $P_{\nu,0}$, we get
$$
\lambda_{\nu} P_{\nu,0} P_{\nu,1} P_{\nu,0} = \frac{i}{2} \lambda_{\nu} P_{\nu,0} \{P_{\nu,0},P_{\nu,0}\} P_{\nu,0} = - I_{\nu}^{(3)}.
$$
Consequently, 
\begin{equation}\label{terme à controler1}
\tilde{H}_{\nu,1}= \frac{1}{2i} P_{\nu,0} \{P_{\nu,0},H_0\}P_{\nu,0} + P_{\nu,0}H_1 P_{\nu,0} - i [P_{\nu,0},\{\lambda_{\nu},P_{\nu,0}\}].
\end{equation}
Using assumption \textbf{(A2)} and Lemma \ref{smoothness}, we clearly see that $\tilde{H}_{\nu,1}\in S(1)$. This ends the proof of the lemma.   

\end{proof}

\begin{remark}
As in Lemma \ref{transport avec le flot}, combining \eqref{flowwnu} and \eqref{est transport nu} and using the Fa\'a Di Bruno formula \eqref{Faa de Bruno}, we get the following estimate on the derivatives of $T_{\nu}\big( s,\phi_{\nu}^t(x,\xi) \big)$ : for all $\gamma\in \mathbb N^{2n}$, there exists $C_{\nu,\gamma}>0$ such that 
\begin{equation}\label{est transport avec le flot nu}
\big\Vert \partial_{(x,\xi)}^{\gamma} \big( T_{\nu}(s,\phi_{\lambda}^t(x,\xi)) \big) \big\Vert \leq C_{\nu,\gamma} \exp\bigg( |\gamma|\Gamma_{\nu}(|t|+|s|)  \bigg),
\end{equation}
uniformly for $t,s\in \mathbb R$ and $(x,\xi)\in \mathbb R^{2n}$. The same estimate remains valid for $T^{-1}_{\nu}(s,\phi_{\lambda}^t(x,\xi))$.
\end{remark}

We end our series of Lemmas by the following one where we control the derivatives of the symbols $(H_{\nu,j})_{j\geq 0}$.
\begin{lemma}\label{needed lemma}
Under assumptions \textbf{(A1)} and \textbf{(A2)}, for all $j\geq 0$ and $\gamma\in \mathbb N^{2n}$ with $|\gamma|+j\geq 1$, we have 
\begin{equation}\label{bornitude sur les Hnuj}
\partial_{(x,\xi)}^{\gamma} H_{\nu,j} \in L^{\infty}(\mathbb R^{2n}). 
\end{equation}
\end{lemma}
\begin{proof}
From the proof of Lemma \ref{smoothness}, one verify that by combining condition \eqref{Gap} and assumption \textbf{(A2)}, we get 
\begin{equation}\label{import step}
\Vert \partial_{(x,\xi)}^{\gamma} P_{\nu,0}(x,\xi) \Vert \leq C_{\gamma} g^{-1}(x,\xi), \quad \forall |\gamma|\geq 1.
\end{equation}
Thus, since $H_{\nu,0}=\lambda_{\nu}P_{\nu,0}$, then \eqref{bornitude sur les Hnuj} for $j=0$ follows immediately from \eqref{import step} and inequality \eqref{mm classe de symbole}.

Now, for $j\geq 1$, from the composition formula \eqref{développement asymptotique} we have 
\begin{eqnarray*}
H_{\nu,j} = (P_{\nu}\#H)_j&=& \sum_{|\alpha|+|\beta|+k+p=j} \gamma(\alpha,\beta) {P_{\nu,k}}^{(\beta)}_{(\alpha)} {H_p}_{(\beta)}^{(\alpha)} \\
&=& \sum_{|\alpha|+|\beta|+k=j} \gamma(\alpha,\beta) {P_{\nu,k}}^{(\beta)}_{(\alpha)} {H_0}_{(\beta)}^{(\alpha)} + \sum_{|\alpha|+|\beta|+k=j-1} \gamma(\alpha,\beta) {P_{\nu,k}}^{(\beta)}_{(\alpha)} {H_1}_{(\beta)}^{(\alpha)}.
\end{eqnarray*}
According to Lemma \ref{tilde pnuj smoo}, we have $P_{\nu,k}\in S(g^{-k})$, for all $k\geq 1$. Then, using \textbf{(A2)}, we obtain \eqref{bornitude sur les Hnuj} for all $j\geq 1$.
\end{proof}

\noindent 
Now, we are in position to prove Proposition \ref{est Gen}.

\noindent \textbf{Proof of Proposition \ref{est Gen} :} 

For $j=0$, estimate (\ref{AG0}) is a direct consequence of estimates (\ref{flowwnu}) and (\ref{est transport nu}).

Let us prove estimate \eqref{AGj}. In the following, when it is not precised, all constants $C_{\gamma}>0$ are uniform with respect to $t\in \mathbb R$ and $(x,\xi)\in \mathbb R^{2n}$.

We start by proving \eqref{AGj} for the derivatives of $\tilde{q}_{\nu,j}(t)$, $j\geq 1$, i.e.
\begin{equation}\label{AGj tilde}
\big\Vert  \partial_{(x,\xi)}^{\gamma} \tilde{q}_{\nu,j}(t,x,\xi) \big\Vert \leq C_{\gamma,\nu,j} \exp\bigg( \big(2|\gamma|+4j-2\big)\Gamma_{\nu}|t|\bigg), \quad \forall \gamma\in \mathbb N^{2n}.
\end{equation} 
We proceed by induction with respect to $j$. Recall the expression of $\tilde{q}_{\nu,1}(t)$

$$
\tilde{q}_{\nu,1}(t,x,\xi) = T_{\nu}^{-1}(t,x,\xi) \bigg( \tilde{Q}_{\nu,1}\big( \phi_{\nu}^t(x,\xi)\big) +  \int_0^t W_{\nu,1}(t,s,x,\xi) ds \bigg) T_{\nu}(t,x,\xi),
$$
where 
$$
W_{\nu,1}(t,s,x,\xi) = T_{\nu}^{-1} \big( -s,\phi_{\nu}^{t}(x,\xi) \big) K_{\nu,0}\big(s,\phi_{\nu}^{t-s}(x,\xi) \big) T_{\nu}\big(-s,\phi_{\nu}^t(x,\xi) \big)
$$
$$
K_{\nu,0}(t,x,\xi) = i \bigg( \big[ H_{\nu}(x,\xi;\h), A_{\nu,0}(t,x,\xi;\h) \big]_{\#}  \bigg)_{2} - \frac{d}{dt} \big( A_{\nu,0}(t,x,\xi;\h) \big)_1
$$
and
\begin{equation}\label{Anu0 est}
A_{\nu,0}(t,x,\xi;\h) = \big(P_{\nu}\# \tilde{q}_{\nu,0}(t) \# P_{\nu}\big)(x,\xi;\h).
\end{equation}

Let us estimating the derivatives of $K_{\nu,0}(t,x,\xi)$. Since $H_{\nu}\# P_{\nu} = P_{\nu}\#H_{\nu} = H_{\nu}$, it follows that 
$$
\big[ H_{\nu}, A_{\nu,0}(t) \big]_{\#} = P_{\nu}\# \big[ H_{\nu}, \tilde{q}_{\nu,0}(t) \big]_{\#}\# P_{\nu} .
$$
From this equation, using the composition formula \eqref{développement asymptotique}, we see that $\bigg(\big[ H_{\nu}, A_{\nu,0}(t) \big]_{\#}\bigg)_2$ is a finite linear combination of terms depending on the symbols $P_{\nu,k},H_{\nu,j}, \tilde{q}_{\nu,0}(t)$ and theirs derivatives with at most a derivative of order $2$ (with respect to $(x,\xi)$) of $\tilde{q}_{\nu,0}(t,x,\xi)$. The term $H_{\nu,0}$ appears only in the commutator $[H_{\nu,0},\tilde{q}_{\nu,0}(t)]$ which vanishes since $P_{\nu,0} \tilde{q}_{\nu,0}(t) P_{\nu,0} = \tilde{q}_{\nu,0}(t)$ (see \eqref{en tt temps}). Consequently, using estimate \eqref{AG0}, the fact that $P_{\nu,k}\in S(1)$ and Lemma \ref{needed lemma}, we obtain  
\begin{equation}\label{equ gty1}
\bigg\Vert\partial_{(x,\xi)}^{\gamma}  \bigg( [H_{\nu}, A_{\nu,0}(t)]_{\#}  \bigg)_{2} (x,\xi) \bigg\Vert  \leq  C_{\gamma} \exp\big((|\gamma|+2)\Gamma_{\nu}|t|\big), \quad \forall \gamma\in \mathbb N^{2n}.
\end{equation}
On the other hand, from \eqref{Anu0 est} we have 
$$
\frac{d}{dt} \big(A_{\nu,0}(t)\big)_1 =  \bigg(P_{\nu}\# \frac{d}{dt} \tilde{q}_{\nu,0}(t) \# P_{\nu} \bigg)_1.
$$
Since $\tilde{q}_{\nu,0}(t)$ satisfies equation \eqref{main construction transformation0}, i.e.
\begin{equation*}
\frac{d}{dt} \tilde{q}_{\nu,0}(t) = \{\lambda_{\nu}, \tilde{q}_{\nu,0}(t)\} + i [\tilde{H}_{\nu,1}, \tilde{q}_{\nu,0}(t)],
\end{equation*}
it follows from estimate \eqref{AG0} again, assumption \textbf{(A2)} and the fact that $\tilde{H}_{\nu,1}\in S(1)$ (see the proof of Lemma \ref{lemma est transport numu}) that for all $\gamma \in \mathbb N^{2n}$, there exists $C_{\gamma}>0$ independent of $t\in \mathbb R$ and $(x,\xi)\in \mathbb R^{2n}$ such that
\begin{equation}\label{equ gty2}
\bigg\Vert\partial_{(x,\xi)}^{\gamma}\bigg(\frac{d}{dt} \big(A_{\nu,0}(t)\big)_1\bigg) (x,\xi) \bigg\Vert  \leq  C_{\gamma} \exp\big((|\gamma|+2)\Gamma_{\nu}|t|\big).
\end{equation}
Putting together \eqref{equ gty1} and \eqref{equ gty2}, we obtain 
\begin{equation*}
\bigg\Vert\partial_{(x,\xi)}^{\gamma} K_{\nu,0}(t,x,\xi) \bigg\Vert  \leq  C_{\gamma} \exp\big((|\gamma|+2)\Gamma_{\nu}|t|\big), \quad \forall \gamma \in \mathbb N^{2n}.
\end{equation*}
As in the proof of Proposition \ref{u5}, using the above estimate, estimate \eqref{flowwnu} on the derivatives of the flow $\phi_{\nu}^t$, estimate \eqref{est transport nu} on the derivatives of $T_{\nu}(t,x,\xi)$ and the Fa\'a Di Bruno formula \eqref{Faa de Bruno}, we get
$$
\bigg\Vert \partial_{(x,\xi)}^{\gamma} \tilde{q}_{\nu,1}(t,x,\xi) \bigg\Vert \leq C_{\gamma} \exp\big((2|\gamma|+2)\Gamma_{\nu}|t|\big), \quad \forall \gamma \in \mathbb N^{2n}.
$$
Thus we proved that \eqref{AGj tilde} for $j=1$.

Let us now assume that $\tilde{q}_{\nu,k}(t,x,\xi)$ satisfies \eqref{AGj tilde} for $k\in \{1,...,r-1\}$. Recall the expression of $\tilde{q}_{\nu,r}(t)$
$$
\tilde{q}_{\nu,r}(t,x,\xi) = T_{\nu}^{-1}(t,x,\xi) \bigg( \tilde{Q}_{\nu,r}\big( \phi_{\nu}^t(x,\xi)\big) +  \int_0^t W_{\nu,r}(t,s,x,\xi) ds \bigg) T_{\nu}(t,x,\xi),
$$
where 
$$
W_{\nu,r}(t,s,x,\xi) = T_{\nu}^{-1} \big( -s,\phi_{\nu}^{t}(x,\xi) \big) K_{\nu,r-1}\big(s,\phi_{\nu}^{t-s}(x,\xi) \big) T_{\nu}\big(-s,\phi_{\nu}^t(x,\xi) \big)
$$
$$
K_{\nu,r-1}(t,x,\xi) = i \bigg( \big[ H_{\nu}(x,\xi;\h), A_{\nu,r-1}(t,x,\xi;\h) \big]_{\#}  \bigg)_{r+1} - \frac{d}{dt} \big( A_{\nu,r-1}(t,x,\xi;\h) \big)_r
$$
and
$$
A_{\nu,r-1}(t,x,\xi;\h) = \big(P_{\nu}\# \sum_{k=0}^{r-1}\h^k \tilde{q}_{\nu,k}(t) \# P_{\nu}\big)(x,\xi;\h).
$$
As above, we have 
$$
\big[ H_{\nu}, A_{\nu,r-1}(t) \big]_{\#} = P_{\nu}\# \big[ H_{\nu}, \sum_{k=0}^{r-1} \h^k \tilde{q}_{\nu,k}(t) \big]_{\#} \# P_{\nu}
$$
which yields 
\begin{align*}
\bigg( \big[ H_{\nu}, A_{\nu,r-1}(t) \big]_{\#}  \bigg)_{r+1} = \sum_{k=0}^{r-1} \bigg( P_{\nu}\# &\big[ H_{\nu}, \tilde{q}_{\nu,k}(t) \big]_{\#} \# P_{\nu} \bigg)_{r+1-k} .
\end{align*}
Again, using the composition formula \eqref{développement asymptotique}, we see that for all $k\in \{0,...,r-1\}$, $\bigg( P_{\nu}\# \big[ H_{\nu}, \tilde{q}_{\nu,k}(t) \big]_{\#} \# P_{\nu} \bigg)_{r+1-k}$ depends at most on a derivative of order $r+1-k$ of $\tilde{q}_{\nu,k}(t)$ (and on the derivatives of $H_{\nu,j}$ and $P_{\nu,l}$). Consequently, using the induction hypothesis, we get 
\begin{equation}\label{lasoma}
\bigg\Vert \partial_{(x,\xi)}^{\gamma} \bigg( \big[ H_{\nu}, A_{\nu,r-1}(t) \big]_{\#}  \bigg)_{r+1}(x,\xi) \bigg\Vert \leq C_{\gamma,r} \exp\big( (2|\gamma|+4r-2)\Gamma_{\nu}|t|  \big), \quad \forall \gamma\in \mathbb N^{2n}.
\end{equation}
Since $\displaystyle{\frac{d}{dt}} A_{\nu,r-1}(t)$ depends on $\displaystyle{\frac{d}{dt}} \tilde{q}_{\nu,k}(t)$, $k\in \{0,...,r-1\}$, which satisfy equations \eqref{main construction transformation j}, it follows that to estimate the derivatives with respect to $(x,\xi)$ of $\big( \displaystyle{\frac{d}{dt}} A_{\nu,r-1}(t) \big)_{r}$, one first needs estimates on the derivatives of $K_{\nu,k}(t)$ with $k\in \{0,...,r-2\}$. This can be made by induction on $k$ and we get that $\big(\displaystyle{\frac{d}{dt}} A_{\nu,r-1}(t)\big)_r$ satisfies estimate \eqref{lasoma}. Consequently, we obtain 
$$
\bigg\Vert \partial_{(x,\xi)}^{\gamma} K_{\nu,r-1}(t,x,\xi) \bigg\Vert \leq  C_{\gamma,r} \exp\big( (2|\gamma|+4r-2)\Gamma_{\nu}|t|  \big), \quad \forall \gamma\in \mathbb N^{2n}.
$$
We conclude as in the proof of Proposition \ref{u5} using estimates \eqref{flowwnu}, \eqref{est transport avec le flot nu} and Leibniz formula. Hence, 
\begin{equation}\label{estimations des symb tilde}
\big\Vert  \partial_{(x,\xi)}^{\gamma} \tilde{q}_{\nu,j}(t,x,\xi) \big\Vert \leq C_{\gamma,\nu,j} \exp \bigg( \big( 2|\gamma| + 4j -2 \big) \Gamma_{\nu} |t| \bigg), \quad \forall \gamma \in \mathbb R^{2n}, \forall j\geq 1,
\end{equation}
uniformly for $t\in \mathbb R$ and $(x,\xi)\in \mathbb R^{2n}$. This ends the proof of \eqref{AGj tilde}.

Turn now to the proof of estimate \eqref{AGj}. Let $j\geq 1$. According to the general form of the solution \eqref{general formula for the solution}, we have 
\begin{align*}
q_{\nu,j}(t,x,\xi) = \bigg( P_{\nu} \# \tilde{q}_{\nu,0}(t) \#P_{\nu} \bigg)_j(x,\xi) + \bigg( P_{\nu} \# \tilde{q}_{\nu,1}(t) \#P_{\nu} \bigg)_{j-1}(x,\xi) + \cdots + \bigg( P_{\nu} \# \tilde{q}_{\nu,j}(t) \#P_{\nu} \bigg)_0(x,\xi).
\end{align*}
By the composition formula \eqref{développement asymptotique}, each term $\bigg( P_{\nu} \# \tilde{q}_{\nu,k}(t) \#P_{\nu} \bigg)_{j-k}(x,\xi)$, $k\in \{0,...,j\}$, in the  above sum is a finite linear combination of terms depending on $P_{\nu,l}(x,\xi), \tilde{q}_{\nu,k}(t,x,\xi)$ and theirs derivatives (with respect to $(x,\xi)$) with at most a derivative of order $j-k$ of $\tilde{q}_{\nu,k}(t,x,\xi)$. Then, using \eqref{AGj tilde} and the fact that $P_{\nu,l}\in S(1)$ for all $l\geq 0$, we deduce that for all $1\leq k\leq j$ and $\gamma\in \mathbb N^{2n}$, we have
\begin{equation}
\big\Vert \partial_{(x,\xi)}^{\gamma}  \bigg( P_{\nu} \# \tilde{q}_{\nu,k}(t) \#P_{\nu} \bigg)_{j-k}(x,\xi) \big\Vert \leq C_{j,k,\gamma,\nu} \exp \bigg( \big( 2|\gamma| + 2(j+k) +2 \big) \Gamma_{\nu} |t|  \bigg).
\end{equation}
Taking the supremum over $k\in \{1,...,j\}$, we get \eqref{AGj}. This ends the proof of Proposition \ref{est Gen}.

\begin{flushright}
$\square$
\end{flushright}

\subsubsection{Proofs of Theorem \ref{main T} and Corollary \ref{sec coro}}
\textbf{Proof of Theorem \ref{main T} :}

The starting point is the same as in the proof of Theorem \ref{main T1}. Set 
$$
U_{H_{\nu}}(t) := e^{-\frac{it}{\h}H_{\nu}^w} = e^{-\frac{it}{\h}P_{\nu}^w H^w P_{\nu}^w}, \quad t\in \mathbb R.
$$
For $N\in \mathbb N$, let $Q_{\nu}^{(N)}(t)$ be the remainder term of order $N$ in the asymptotic expansion of $Q_{\nu}(t)$, i.e. 
$$
Q_{\nu}^{(N)}(t):=Q_{\nu}(t)-\sum_{j=0}^N \h^j \big(q_{\nu,j}(t)\big)^w(x,\h D_x).
$$
\begin{lemma}\label{rem est geni}
Fix $1\leq \nu \leq l$. For all $N\in \mathbb N$, the following estimate holds
\begin{eqnarray*}
\bigg\Vert Q_{\nu}^{(N)}(t)\bigg\Vert_{\mathcal{L}(L^2(\mathbb R^n)\otimes \mathbb C^m)} \leq \h^{N+1}  \bigg\Vert  \int_0^t U_{H_{\nu}}(-s) \big(R_{\nu}^{(N+1)}(t-s)\big)^w U_{H_{\nu}}(s) ds \bigg\Vert_{\mathcal{L}(L^2(\mathbb R^n)\otimes \mathbb C^m)} + \mathcal{O}(\h^{N+1}),
\end{eqnarray*}
uniformly for $t\in \mathbb R$, where 
\begin{equation}
R_{\nu}^{(N+1)}(t):= \tilde{R}_{N+1}(H_{\nu},q_{\nu,0}(t))+\tilde{R}_{N}(H_{\nu},q_{\nu,1}(t))+\cdots+\tilde{R}_{1}(H_{\nu},q_{\nu,N}(t)).
\end{equation}
We recall that the notation $\tilde{R}_{k}(A,B)$ is introduced in \eqref{la fonction tilde R}.
\end{lemma}

For $N\in \mathbb N$, we set 
$$
Q^{(N)}(t) := Q(t) - \sum_{j=0}^N \h^j \sum_{\nu=1}^l \big(q_{\nu,j}(t)\big)^w(x,\h D_x).
$$
Using Lemma \ref{rem est geni} and Proposition \ref{proposition réduction} (i), we obtain 
\begin{align}
\big\Vert Q^{(N)}(t) \big\Vert_{\mathcal{L}(L^2(\mathbb R^n)\otimes \mathbb C^m)} &\leq   \sum_{\nu=1}^l \big\Vert Q_{\nu}^{(N)}(t) \big\Vert_{\mathcal{L}(L^2(\mathbb R^n)\otimes \mathbb C^m)} + \bigg\Vert Q(t) - \sum_{\nu= 1}^l Q_{\nu}(t) \bigg\Vert_{\mathcal{L}(L^2(\mathbb R^n)\otimes \mathbb C^m)}\label{go back}\\
& \leq   l\: \h^{N+1}  \sup_{1\leq \nu \leq l} \bigg\Vert \int_0^t U_{H_{\nu}}(-s) \bigg(R_{\nu}^{(N+1)}(t-s)\bigg)^w U_{H_{\nu}}(s) ds \bigg\Vert_{\mathcal{L}(L^2(\mathbb R^n)\otimes \mathbb C^m)} + \mathcal{O}(\h^{N+1}) \nonumber\\
& + \mathcal{O}\big( (1+|t|)\h^{\infty}\big), \nonumber
\end{align}
uniformly for $t\in \mathbb R$.

As in the end of the proof of Theorem \ref{main T1}, using the estimates on the symbols $q_{\nu,j}(t)$ given by Proposition \ref{est Gen}, Theorem \ref{est ress} and the Calder\'on-Vaillancourt theorem (Theorem \ref{Cal}), we prove the following estimate 
$$
\big\Vert  \big(R_{\nu}^{(N+1)}(t) \big)^w  (x,\h D_x;\h) \big\Vert_{\mathcal{L}(L^2(\mathbb R^n)\otimes \mathbb C^m)} \leq C_{\nu,n,N} \exp\bigg(  \big(4N + \tilde{\delta}_n \big) \Gamma_{\nu} |t|  \bigg),
$$
uniformly for $t\in \mathbb R$, where $\tilde{\delta}_n$ is an integer depending only on the dimension $n$. We conclude as in the end of the proof of Theorem \ref{main T1}.
\begin{flushright}
$\square$
\end{flushright}

\noindent
\textbf{Proof of Corollary \ref{sec coro} :}

Let $Q(x,\xi)\sim \sum_{j\geq 0}\h^j Q_j(x,\xi)$ in $S(1)$ and assume that there exists $\tilde{Q}\in S(1)$ such that 
$$
Q_0(x,\xi) = \sum_{\nu=1}^l P_{\nu,0}(x,\xi) \tilde{Q}(x,\xi) P_{\nu,0}(x,\xi).
$$
According to Proposition \ref{proposition réduction} (ii), we have 
$$
\bigg\Vert Q(t) - \sum_{\nu= 1}^l Q_{\nu}(t) \bigg\Vert_{\mathcal{L}(L^2(\mathbb R^n)\otimes \mathbb C^m)} = \mathcal{O}\big( (1+|t|) \h  \big), \quad \text{uniformly for}\; t\in \mathbb R.
$$
Thus by rewriting \eqref{go back} for $N=0$ and using Lemma \ref{rem est geni}, we get
\begin{align*}
\big\Vert Q^{(0)}(t) \big\Vert_{\mathcal{L}(L^2(\mathbb R^n)\otimes \mathbb C^m)} &\leq   \sum_{\nu=1}^l \big\Vert Q_{\nu}^{(0)}(t) \big\Vert_{\mathcal{L}(L^2(\mathbb R^n)\otimes \mathbb C^m)} +  \bigg\Vert Q(t) - \sum_{\nu= 1}^l Q_{\nu}(t) \bigg\Vert_{\mathcal{L}(L^2(\mathbb R^n)\otimes \mathbb C^m)}\\
&\leq   l\: \h  \sup_{1\leq \nu \leq l} \bigg\Vert \int_0^t U_{H_{\nu}}(-s) \bigg(R_{\nu}^{(1)}(t-s)\bigg)^w U_{H_{\nu}}(s) ds \bigg\Vert_{\mathcal{L}(L^2(\mathbb R^n)\otimes \mathbb C^m)} + \mathcal{O}(\h) + \mathcal{O}\big( (1+|t|) \h \big),
\end{align*}
uniformly for $t\in \mathbb R$. We conclude as above.

\begin{flushright}
$\square$
\end{flushright}
We end this section by the following remark concerning an application of the results of this paper.
\begin{remark}[Application] Consider the matrix semiclassical Schr\"odinger operator in $L^2(\mathbb R^n)\otimes \mathbb C^m$
\begin{equation}\label{schrodinger without crossings}
P(\h):=-\h^{2}\Delta \otimes I_m + V(x),
\end{equation}
where $V$ is a $(m\times m)$ hermitian-valued potential satisfying the following long-range assumption

\textbf{(S1).} There exists an hermitian matrix $V_{\infty}\in M_m(\mathbb C)$ and a constant $\delta>0$ such that for all $\alpha \in \mathbb N^n$,
$$
\big\Vert \partial_x^{\alpha} \big( V(x)-V_{\infty} \big) \big\Vert \leq C_{\alpha} \langle x \rangle^{-\delta-|\alpha|}, \quad \forall x\in \mathbb R^n. 
$$
The limiting absorption principle (see e.g. \cite{Gerard0}) ensures that the boundary values of the resolvent of $P(\h)$, 
$$
\big( P(\h)-(E\pm i0)\big)^{-1}:= \lim_{\varepsilon \searrow 0} \big( P(\h)-(E\pm i \varepsilon)\big)^{-1}
$$ 
exists as bounded operators form $L^2_s(\mathbb R^n)\otimes \mathbb C^m$ to $L^2_{-s}(\mathbb R^n)\otimes \mathbb C^m$ for any $s>\frac{1}{2}$ and $E$ outside the pure point spectrum of $P(\h)$. Here $L^2_s(\mathbb R^n)\otimes \mathbb C^m$ denotes the space of $\mathbb C^m$-valued functions defined on $\mathbb R^n$ such that $x\mapsto \langle x \rangle^{s}f(x)$ belongs to $L^2(\mathbb R^n)\otimes \mathbb C^m$. 

In the scalar case, i.e. when $m=1$, (resp. the matrix-valued case without crossings eigenvalues), a well known result is a bounds $\mathcal{O}(\h^{-1})$ on these boundary values near non-trapping energies for the Hamiltonian $p(x,\xi):=|\xi|^2+V(x)$ (resp. for the eigenvalues of $p(x,\xi):=|\xi|^2 I_m + V(x)$). We refer to \cite{Rob2} (resp. \cite{Jecko}) for the proofs of these results. In the case of trapped energies, using the results of Bouzouina-Robert \cite{Rob}, Bony, Burq and Ramond \cite{bony} proved a lower bound of the type $\h^{-1}\log(\h^{-1})$ on the boundary values of the resolvent of scalar Schr\"odinger operators. According to the remark after Theorem 2 in \cite{bony} and our main results (Theorem \ref{main T} and Corollary \ref{second coro}), we can obtain the same lower bound for the boundary values of the operator (\ref{schrodinger without crossings}). The detailed proof will appear elsewhere. 

\end{remark}

%%%%%%%%%%%%%%%%%%%%%%%%%%%%%%%%%%%%%%%%%%%%%%%%%%%%%%%%%%%%%%%%%%%%%%%%%%%%%%%%%%%%%%%%%%%%%%%%%%%%%%%%%%%%%%%%%%%%%%%%%%%%%%%%%

\appendix

\section{Review of semiclassical pseudodifferential calculus for matrix valued symbols}\label{semi-classical background}
In this section we recall some notions and results about the semiclassical pseudodifferential calculus in the context of operators with matrix-valued symbols. These results are well known in the case of scalar-valued symbols and we refer to \cite[ch. 7-9]{dim} and \cite[ch. 4]{Zwo} for more details.

The set of Weyl operators with symbols in the classes $S(g)$ introduced in section \ref{sec} is stable under the operator multiplication. Let $\sigma$ be the canonical symplectic form on $\mathbb R^{2n}$
\begin{equation}\label{symplectic form}
\sigma(x,\xi;y,\zeta):= \langle J(x,\xi),(y,\zeta)\rangle, \quad J:=\begin{pmatrix}
0 & I_n \\
-I_n & 0
\end{pmatrix}, \quad \forall (x,\xi,y,\zeta) \in \mathbb R^{4n}.
\end{equation}
More precisely, we have the following well known result (see \cite{Rob1, Zwo})
\begin{theorem}
Let $g_1,g_2$ be two order functions on $\mathbb R^{2n}$. The map 
$$
\begin{array}{rcl}
S(g_1)\times S(g_2) & \longrightarrow &  S(g_1 g_2) \\
(P,Q)  &\longmapsto & P\# Q  \\
\end{array}
$$ 
where $P \# Q$ is defined by :
\begin{equation}\label{definition}
P \# Q (x,\xi):=e^{\frac{i\h}{2}\sigma(D_x,D_{\xi};D_{y},D_{\eta})}\big(P(x,\xi)Q(y,\eta)\big)_{|(x,\xi)=(y,\eta)},
\end{equation}
is a bilinear continuous map in the topology generated by the semi-norms associated to (\ref{seminorms}) and we have 
$$
(P \# Q)^w(x,\h D_x)=P^w(x,\h D_x) \circ Q^w(x,\h D_x),
$$
as operators mapping $\mathscr{S}(\mathbb R^{n})\otimes \mathbb C^m$ to $\mathscr{S}(\mathbb R^{n})\otimes \mathbb C^m$. The symbol $P\#Q$, called the Moyal product of $P,Q$, admits the following asymptotic expansion in powers of $\h$
\begin{equation}\label{dév asy annexe}
P\# Q(x,\xi) \sim \sum_{j\geq 0} \frac{h^j}{\fact{j}} \big(\frac{i}{2}\sigma(D_x,D_{\xi};D_y,D_{\eta})\big)^j \big( P(x,\xi)Q(y,\eta)\big)_{|(x,\xi)=(y,\eta)} \quad \text{in}\;\; S(g_1g_2).
\end{equation}
Furthermore, if $P(x,\xi;\h)\sim \sum_{j\geq 0}\h^j P_j(x,\xi)$ in $S(g_1)$ and $Q(x,\xi;\h)\sim \sum_{j \geq 0}\h^j Q_j(x,\xi)$ in $S(g_2)$ are two semiclassical symbols, then $P\#Q$ is again a semiclassical symbol and we have    
\begin{equation}\label{product rule}
P\#Q(x,\xi;\h) \sim \sum_{j\geq 0} h^j (P\#Q)_j(x,\xi) \quad \text{in}\;\; S(g_1g_2),
\end{equation}
where for all $j\geq 0$,
\begin{equation}\label{développement asymptotique}
(P\#Q)_j(x,\xi) :=  \sum_{| \alpha | + | \beta | + k + l = j} \gamma(\alpha,\beta) {P_k}_{(\alpha)}^{(\beta)}(x,\xi) {Q_l}_{(\beta)}^{(\alpha)}(x,\xi),
\end{equation}
with $\gamma(\alpha,\beta):=\displaystyle{\frac{(-1)^{| \beta |}}{(2i)^{| \alpha | + | \beta |}\fact{\alpha}\fact{\beta}} }$. In particular, the principal symbol and the sub-principal symbol of $P\#Q$ are respectively given by 
\begin{equation}\label{principal and sub-principal symbol}
(P\#Q)_0=P_0Q_0, \quad (P\#Q)_1= \frac{1}{2i}\{P_0,Q_0\}+P_0Q_1+P_1Q_0.
\end{equation}
\end{theorem}

In the following remark we collect some useful identities which can be easily computed using \eqref{principal and sub-principal symbol}.
\begin{remark}
We recall that the Moyal commutator $[P,Q]_{\#}$ of $P$ and $Q$ is defined as $[P,Q]_{\#}:= P \# Q - Q\# P$. For $P\sim \sum_{j\geq 0}\h^j P_j$, $Q\sim \sum_{j\geq 0} \h^j Q_j$ and $C\sim \sum_{j\geq 0}\h^j C_j$ three semiclassical matrix-valued symbols, we have 
\begin{equation}\label{deduc1}
\big( [P,Q]_{\#} \big)_0 = [P_0,Q_0], \quad \big([P,Q]_{\#}\big)_1 = \frac{1}{2i} \big( \{P_0,Q_0\} - \{Q_0, P_0\} \big) + [P_0,Q_1] + [P_1,Q_0].
\end{equation}
\begin{equation}\label{deduc2}
(P\#Q\#C)_0 = P_0 Q_0 C_0, \quad (P\#Q\#C)_1 = \frac{1}{2i} \{P_0Q_0,C_0\} + P_0Q_0 C_1 + \frac{1}{2i} \{P_0,Q_0\} C_0 + P_0 Q_1 C_0 + P_1 Q_0 C_0.
\end{equation}
\end{remark}
\subsection*{The Moyal bracket}
Let $P\sim \sum_{j\geq 0}\h^j P_j$ in $S(g_1)$ and $Q\sim \sum_{j\geq 0} \h^j Q_j$ in $S(g_2)$ be two matrix-valued semiclassical symbols. The Moyal bracket of $P,Q$ denoted $\{P,Q\}^*$ is defined as the Weyl symbol of $\frac{i}{\h}[P^w,Q^w]$. By means of the Moyal product it can be written as 
$$
\{P,Q\}^*:=\frac{i}{\h}[P,Q]_{\#}=\frac{i}{\h}(P\#Q-Q\#P).
$$
If the principal symbols $P_0$ and $Q_0$ commute, i.e. if $[P_0,Q_0]=0$, then using the rule of asymptotic expansion of the Moyal product of symbols (formula \ref{product rule}), one can expand $\{P,Q\}^*$ in a power series of $\h$ and gets
\begin{equation}\label{rMp}
\{P,Q\}^*\sim \sum_{j\geq 0} \h^j \{P,Q\}^*_j \quad \text{in}\;\; S(g_1g_2),
\end{equation}
with $\{P,Q\}^*_j=i \big( [P,Q]_{\#}\big)_{j+1}=i \big( (P\#Q)_{j+1}-(Q\#P)_{j+1} \big)$, for all $j\geq 0$. 

Let $N\geq 1$. The remainder term of order $N-1$ in the asymptotic expansion (\ref{rMp}) can be expressed by means of the remainder terms in the asymptotic expansions of $P\#Q$ and $Q\#P$. More precisely, we have 
\begin{equation}\label{asy Moyal bracket} 
\{P,Q\}^*-\sum_{j=0}^{N-1}\h^j \{P,Q\}^*_{j}=i\h^{-1}(R_N(P,Q)-R_N(Q,P)),
\end{equation}
where 
\begin{equation}\label{remainder termcv}
R_{N}(P,Q;x,\xi;\h):=P \# Q(x,\xi)-\sum_{j=0}^N \h^j (P\#Q)_j(x,\xi)
\end{equation}
denotes the remainder term of order $N$ in the asymptotic expansion of $P\#Q$.

\subsection*{Remainder estimate in the composition formula} In \cite[Theorem A.1]{Rob}, Bouzouina and Robert established the following estimate on the derivatives of $R_N(P,Q)$ in the case of scalar-valued symbols. This result remains true without any change in the case of matrix-valued symbols.
\begin{theorem}\label{est ress}
There exists a constant $K_n>0$ such that for every integer $\kappa\geq 4n$ and every $s>4n$, there exists $\tau_{n,\kappa,s}>0$ such that for every $P,Q\in \mathscr{S}(\mathbb R^{2n})\otimes M_m(\mathbb C)$ we have :

For every $N\geq 1$ and every $\gamma \in \mathbb N^{2n}$, the following estimate holds for every $u\in \mathbb R^{2n}$
 \begin{align}\label{ress}
\Vert \partial_{u}^{\gamma}R_{N}(P,Q;u;\h) \Vert \leq \h^{N+1} &\tau_{n,\kappa,s} K_{n}^{N+|\gamma|}(\fact{N})^{-1} \nonumber \\
& \times \sup_{\substack{v,w \in \mathbb R^{2n}\\\mu,\nu\in \mathbb N^{2n};|\mu|+|\nu|\leq \kappa +|\gamma|\\\alpha,\beta \in \mathbb N^n;|\alpha|+|\beta|=N+1}} \bigg( \langle (v,w) \rangle^{s-\kappa} \big\Vert \partial_{v}^{(\alpha,\beta)+\mu}P(v+u)\big\Vert  \big\Vert \partial_{w}^{(\beta,\alpha)+\nu}Q(w+u) \big\Vert \bigg).
\end{align}
\end{theorem}
\begin{remark}\label{extension}
As it was shown in \cite{Rob}, using the fact that $\mathscr{S}(\mathbb R^{2n})\otimes M_m(\mathbb C)$ is dense in $S(\langle u \rangle^a; \mathbb R^{2n}, M_m(\mathbb C) )$, $a\in \mathbb R$, for the topology of the Fr\'echet spaces $S(\langle u \rangle^{a+\varepsilon}; \mathbb R^{2n}, M_m(\mathbb C))$, for all $\varepsilon>0$, Theorem \ref{est ress} can be extended to symbols $P\in S(\langle u \rangle^a; \mathbb R^{2n}, M_m(\mathbb C))$ and $Q\in S(\langle u \rangle^b; \mathbb R^{2n}, M_m(\mathbb C))$, with $a,b\in \mathbb R$ such that $\kappa-s\geq a+b$ to get a finite right hand side in (\ref{ress}).
\end{remark}
We end this background by the following well known result (see \cite[ch. 4]{Zwo}).
\begin{theorem}[Calder\'on-Vaillancourt]\label{Cal}
There exists an integer $k_n$ and a constant $C_n>0$ such that if $Q\in S(1)$ then $Q^w(x,hD_x;\h) : L^{2}(\mathbb R^n)\otimes \mathbb C^m \rightarrow L^{2}(\mathbb R^n)\otimes \mathbb C^m$ is bounded and we have 
\begin{equation}\label{improv boulk}
\big\Vert Q^w(x,hD_x;\h) \big\Vert_{\mathcal{L}(L^{2}(\mathbb R^n)\otimes \mathbb C^m)}\leq C_n \sup_{|\alpha|+|\beta|\leq k_n} \h^{\frac{|\alpha|+|\beta|}{2}} \big\Vert\partial_{\xi}^{\alpha}\partial_x^{\beta}Q \big\Vert_{L^{\infty}(\mathbb R^{2n})}.
\end{equation}
\end{theorem}
%\begin{remark}
%The optimal integer $k_n$ for which the above result holds was obtained by Boulkhemair \cite{Bou} which proved the following estimate
%$$
%\Vert Q^w(x,hD_x;\h)\Vert_{\mathcal{L}(L^{2}(\mathbb R^n)\otimes \mathbb C^m)}\leq C_n \sup_{|\alpha|,|\beta|\leq [\frac{n}{2}]+1}\Vert\partial_{\xi}^{\alpha}\partial_x^{\beta}Q\Vert_{L^{\infty}(\mathbb R^{2n})}.
%$$ 
%\end{remark}

\section{Cauchy problem}\label{Ann Cauchy problem general}

Let $\Lambda\in C^{\infty}(\mathbb R^{2n};\mathbb R)$, $A\in C^{\infty}(\mathbb R^{2n}) \otimes M_m(\mathbb C)$ hermitian-valued and $B\in C^{\infty}(\mathbb R\times \mathbb R^{2n}) \otimes M_m(\mathbb C)$. In this paragraph, we give the general solution of the following Cauchy problem
\begin{equation}\label{Cauchy problem general}
\left\{   
\begin{array}{rcl}
\displaystyle{\frac{d}{dt}} \psi(t,x,\xi) &=& \big\{ \Lambda,\psi(t,\cdot,\cdot)\big\}(x,\xi) + i [A(x,\xi),\psi(t,x,\xi)]+B(t,x,\xi) \\
 \psi(t,x,\xi)_{|t=0} &=& \psi_0(x,\xi),
\end{array}
\right.
\end{equation}
which arises when we solve the Cauchy problems \eqref{Cj} and \eqref{pbm Cauchy tildej} in sections \ref{Sc} and \ref{Generalization}, respectively. We assume that the flow $\phi_{\Lambda}^t(x,\xi)$ exists globally on $\mathbb R$ for all $(x,\xi)\in \mathbb R^{2n}$ since it is the case for $\phi_{\lambda}^t$ and $\phi_{\nu}^t$ (see section \ref{sec}). 

We introduce the $(m\times m)$ matrix-valued function $T$ solution of the following system
\begin{equation}\label{equ1tr}
\frac{d}{dt} T(t,x,\xi)= -i A\big(\phi_{\Lambda}^t(x,\xi)\big) T(t,x,\xi), \quad T(0,x,\xi)=I_m.
\end{equation}
The following lemma was proved in \cite[Proposition 4]{Brum}.
\begin{lemma}\label{trbrumm}
The matrix $T(t,x,\xi)$ is unitary and we have 
\begin{equation}\label{property unitary mat}
T\big( -t,\phi_{\Lambda}^t(x,\xi)\big) = T^{-1}(t,x,\xi), \quad \forall t\in \mathbb R, (x,\xi)\in \mathbb R^{2n}.
\end{equation}
\end{lemma}
Notice that in \cite{Brum}, the quantity $\Gamma(t,x,\xi)=T(-t,\phi^t_{\Lambda}(x,\xi))$ was considered instead of $T$. The equation satisfied by $T^{-1}$ reads 
\begin{equation}\label{equ2tr}
\frac{d}{dt} T^{-1}(t,x,\xi) = i T^{-1}(t,x,\xi) A\big(\phi_{\Lambda}^t(x,\xi)\big).
\end{equation}
A simple computation using \eqref{equ1tr} and \eqref{equ2tr} yields 
\begin{align*}
\frac{d}{dt}\bigg(  T^{-1}(-t,x,\xi) \psi\big(t,&\phi_{\Lambda}^{-t}(x,\xi)\big)T(-t,x,\xi) \bigg) = \label{hid22}\\
& T^{-1}(-t,x,\xi)  \bigg( \frac{d}{dt} \psi\big(t,\phi_{\Lambda}^{-t}(x,\xi) \big) - \{\Lambda,\psi(t)\} \circ \phi_{\Lambda}^{-t}(x,\xi) -i [A,\psi(t)]\circ \phi_{\Lambda}^{-t}(x,\xi) \bigg) T(-t,x,\xi) . \nonumber
\end{align*}
Consequently, equation (\ref{Cauchy problem general}) is equivalent to the following one 
\begin{equation*}
\frac{d}{dt}\bigg(  T^{-1}(-t,x,\xi) \psi\big(t,\phi_{\Lambda}^{-t}(x,\xi)\big)T(-t,x,\xi) \bigg) = T^{-1}(-t,x,\xi) B\big(t,\phi_{\Lambda}^{-t}(x,\xi)\big) T(-t,x,\xi).
\end{equation*}
Therefore
\begin{equation}\label{turn}
\psi\big(t,\phi_{\Lambda}^{-t}(x,\xi)\big)= T(-t,x,\xi)\bigg( \psi_0(x,\xi)+ \int_0^t T^{-1}(-s,x,\xi) B\big(s,\phi_{\Lambda}^{-s}(x,\xi)\big) T(-s,x,\xi) \; ds \bigg) T^{-1}(-t,x,\xi).
\end{equation}
Using Lemma \ref{trbrumm}, we obtain the solution of (\ref{Cauchy problem general}) which reads
\begin{equation*}
\psi(t,x,\xi)= T^{-1}(t,x,\xi)\bigg( \psi_0(\phi_{\Lambda}^t(x,\xi))+ \int_0^t T^{-1}\big(-s,\phi_{\Lambda}^t(x,\xi)\big) B\big(s,\phi_{\Lambda}^{t-s}(x,\xi)\big) T\big(-s,\phi_{\Lambda}^t(x,\xi)\big)  \; ds \bigg) T(t,x,\xi).
\end{equation*} 
\begin{flushright}
$\square$
\end{flushright}

The following lemma is used in the proof of Proposition \ref{main proposition construction}. Similar result was announced in the appendix of \cite{Spo} (see equation (A.22) therein).

%we consider the Cauchy problem \eqref{Cauchy problem general} with $\Lambda=\lambda_{\nu}$ and $A= \tilde{H}_{\nu,1}$ defined by \eqref{equ79}. We prove that if the initial condition $\psi_0$ satisfies $\psi_0=P_{\nu,0} \psi_0 P_{\nu,0}$ and $B(t)=P_{\nu,0} B(t) P_{\nu,0}$, then at all times $t$, the solution $\psi(t)$ satisfies $\psi(t)=P_{\nu,0}\psi(t)P_{\nu,0}$. We used this result in the proof of 
\begin{lemma}\label{Lemma remove}
Consider the Cauchy problem \eqref{Cauchy problem general} with $\Lambda=\lambda_{\nu}$ and $A= \tilde{H}_{\nu,1}$ defined by \eqref{equ79}. We assume that $\psi_0$ and $B(t)$ satisfy
$$
\psi_0 = P_{\nu,0} \psi_0 P_{\nu,0} \quad \text{and} \quad B(t)=P_{\nu,0}B(t)P_{\nu,0}, \quad \forall t \in \mathbb R.
$$
Then the solution $\psi(t)$ satisfies  
$$
\psi(t) = P_{\nu,0} \psi(t) P_{\nu,0}  , \quad \forall t\in \mathbb R.
$$
\end{lemma}
\begin{proof}
Put $\overline{P}_{\nu,0} := I_m -P_{\nu,0}$. We shall prove that
$$
\overline{P}_{\nu,0} \psi(t) = 0 \quad \text{and} \quad \psi(t) \overline{P}_{\nu,0}  = 0  , \quad \forall t\in  \mathbb R.
$$

We have
\begin{eqnarray*}
\frac{d}{dt} \overline{P}_{\nu,0} \psi(t) &=& \overline{P}_{\nu,0}  \{\lambda_{\nu},\psi(t)\} - \overline{P}_{\nu,0} [\psi(t), i\tilde{H}_{\nu,1}] \\
&=& \{\lambda_{\nu}, \overline{P}_{\nu,0} \psi(t)\} - \{\lambda_{\nu}, \overline{P}_{\nu,0}\} \psi(t) - \overline{P}_{\nu,0} [\psi(t), i\tilde{H}_{\nu,1}] \\
&=& \{\lambda_{\nu}, \overline{P}_{\nu,0} \psi(t)\} - \{\lambda_{\nu}, \overline{P}_{\nu,0}\} \psi(t) + i \overline{P}_{\nu,0} \tilde{H}_{\nu,1} \psi(t) - i\overline{P}_{\nu,0} \psi(t) \tilde{H}_{\nu,1} \\
&=& \{\lambda_{\nu}, \overline{P}_{\nu,0} \psi(t)\} - \{\lambda_{\nu}, \overline{P}_{\nu,0}\} \psi(t) + i \overline{P}_{\nu,0} \tilde{H}_{\nu,1} \psi(t) + \mathcal{O}(\overline{P}_{\nu,0} \psi(t)),
\end{eqnarray*}
where we used the fact that $\tilde{H}_{\nu,1}\in S(1)$ (see the proof of Lemma \ref{lemma est transport numu}). %A simple computation yields
According to the definition of  $\tilde{H}_{\nu,1}$ we have 
\begin{equation}\label{new4}
i\overline{P}_{\nu,0} \tilde{H}_{\nu,1} \psi(t) = \overline{P}_{\nu,0}  \big[ P_{\nu,0}, \{\lambda_{\nu},P_{\nu,0}\} \big] \psi(t)= - \overline{P}_{\nu,0} \{\lambda_{\nu}, P_{\nu,0}\} P_{\nu,0} \psi(t) .
\end{equation}
Next, multiplying  the obvious equality  $\{\lambda_{\nu},P_{\nu,0}\} = \{\lambda_{\nu},P_{\nu,0}^2\} = \{\lambda_{\nu},P_{\nu,0}\} P_{\nu,0} + P_{\nu,0} \{\lambda_{\nu},P_{\nu,0}\}$
%$ \{\lambda_\nu, P_{\nu,0}\}=
on the left and right by $ P_{\nu,0},$ gives  $P_{\nu,0} \{\lambda_{\nu},P_{\nu,0}\} P_{\nu,0} =2 P_{\nu,0}\{\lambda_{\nu},P_{\nu,0}\} P_{\nu,0} $  and then
$P_{\nu,0} \{\lambda_{\nu},P_{\nu,0}\} P_{\nu,0}=0.$ Combining this with \eqref{new4}, we obtain
$$
i\overline{P}_{\nu,0} \tilde{H}_{\nu,1} \psi(t) 
=
 \{\lambda_{\nu}, \overline{P}_{\nu,0}\}P_{\nu,0} \psi(t).$$
%Here to pass from the second to the third equality we have used that $P_{\nu,0}\{\lambda_{\nu}, P_{\nu,0}\}P_{\nu,0} = 0$ which follows from $\{\lambda_{\nu},P_{\nu,0}\} = \{\lambda_{\nu},P_{\nu,0}^2\} = \{\lambda_{\nu},P_{\nu,0}\} P_{\nu,0} + P_{\nu,0} \{\lambda_{\nu},P_{\nu,0}\}$. 

Therefore, we have 
\begin{eqnarray*}
\frac{d}{dt} \overline{P}_{\nu,0} \psi(t) &=& \{\lambda_{\nu}, \overline{P}_{\nu,0} \psi(t)\} - \{\lambda_{\nu}, \overline{P}_{\nu,0}\} \psi(t) + \{\lambda_{\nu}, \overline{P}_{\nu,0}\}P_{\nu,0} \psi(t) +  \mathcal{O}(\overline{P}_{\nu,0} \psi(t)) \\
&=& \{\lambda_{\nu}, \overline{P}_{\nu,0} \psi(t)\} - \{\lambda_{\nu}, \overline{P}_{\nu,0}\} \overline{P}_{\nu,0} \psi(t) + \mathcal{O}(\overline{P}_{\nu,0} \psi(t)), 
\end{eqnarray*}
which by using the fact that $\{\lambda_{\nu},\overline{P}_{\nu,0}\} = \mathcal{O}(1)$ (which follows from assumption \textbf{(A2)} and Lemma \ref{smoothness}) gives 
$$
\frac{d}{dt} \overline{P}_{\nu,0} \psi(t) = \{\lambda_{\nu}, \overline{P}_{\nu,0} \psi(t)\} + \mathcal{O}(\overline{P}_{\nu,0} \psi(t)).
$$
Put $g(t,x,\xi) := \overline{P}_{\nu,0}(x,\xi) \psi(t,x,\xi)$ and $f(t,x,\xi):= g(t,\phi_{\nu}^{-t}(x,\xi))$. Taking into account the fact that $f(0)=g(0)=\overline{P}_{\nu,0}\psi_0 =0$ (since $\psi_0 = P_{\nu,0}\psi_0 P_{\nu,0}$ by hypothesis), we have 
%Taking into account the fact that $\overline{P}_{\nu,0}\psi_0 =0$ (since $\psi_0 = P_{\nu,0}\psi_0 P_{\nu,0}$ by hypothesis), we consider the equation 
%\begin{equation}\label{for equ}
%\frac{d}{dt} g(t) = \{\lambda_{\nu},g(t)\} + \mathcal{O}(g(t)), \quad g(0)=0.
%\end{equation}
%For $f(t,x,\xi):= g(t,\phi_{\nu}^{-t}(x,\xi))$, we have 
$$
\frac{d}{dt} f(t,x,\xi) = \mathcal{O} \big( f(t,x,\xi) \big), \quad f(0)=0.
$$
Consequently, using Gronwall Lemma, we get 
$$
 f(t) = 0, \quad \forall t\in \mathbb R. 
$$
Hence 
$$
\overline{P}_{\nu,0} \psi(t) = 0, \quad \forall t\in \mathbb R.
$$
The same  arguments show that $\psi(t)\overline{P}_{\nu,0} = 0$, for all $t\in \mathbb R$. This ends the proof of the lemma.

\end{proof}
\section{Semiclassical projections}\label{projjj}

In this appendix, we prove that under assumption \textbf{(A1)}, $\lambda_{\nu}$ and $P_{\nu,0}$ belong to nice classes of symbols, for all $1\leq \nu \leq l$, and we give an idea of the proof of  Theorem \ref{projections}. For more details we refer to \cite{Sor} and the original paper \cite{Hel} (see also \cite{bolte}).

\begin{lemma}\label{smoothness}
Fix $1\leq \nu \leq l$. Under assumption \textbf{(A1)}, $P_{\nu,0}\in S(1)$ and for all $\gamma\in \mathbb N^{2n}$, there exists $C_{\gamma}>0$ such that 
\begin{equation}\label{mm classe de symbole}
\big|\partial_{(x,\xi)}^{\gamma}\lambda_{\nu}(x,\xi)\big| \leq C_{\gamma} \big\Vert \partial_{(x,\xi)}^{\gamma} H_0(x,\xi) \big\Vert, \quad \forall (x,\xi)\in \mathbb R^{2n}.
\end{equation}
In particular, $\lambda_{\nu}\in S(g)$.
\end{lemma}
\begin{proof}
Let $\nu\in \{1,...,l\}$. Let $\varepsilon(x,\xi)>0$ be such that
\begin{equation}\label{propriété du rayon}
0<\frac{\rho}{2} g(x,\xi) \leq \varepsilon(x,\xi) \leq  \frac{1}{2} \min_{1\leq \mu \neq \nu \leq l} |\lambda_{\mu}(x,\xi) - \lambda_{\nu}(x,\xi)|. 
\end{equation}
Put
\begin{equation}\label{contour}
\gamma_{\nu}(x,\xi):= \big\{z\in \mathbb C;\; |z-\lambda_{\nu}(x,\xi)|=\varepsilon(x,\xi) \big\},
\end{equation}
and
$$
P_{\nu,0}(x,\xi) = \frac{i}{2\pi} \int_{\gamma_{\nu}(x,\xi)} (H_0(x,\xi)-z)^{-1} dz.
$$
By the Cauchy theorem, we see that a small variation of the contour $\gamma_{\nu}(x,\xi)$ does not change $P_{\nu,0}(x,\xi)$. Let $z\in \gamma_{\nu}(x,\xi)$. According to \eqref{propriété du rayon}, $(H_0(x,\xi)-z)^{-1}$ exists for all $(x,\xi)\in \mathbb R^{2n}$ and since $H_0(x,\xi)$ is hermitian it follows that
\begin{equation}\label{chmoj}
\Vert (H_0(x,\xi) - z)^{-1} \Vert \leq \frac{1}{\text{dist}\big(z,\sigma(H_0(x,\xi))\big)} \leq \frac{2}{\rho} g^{-1}(x,\xi), 
\end{equation}
where $\sigma(H_0(x,\xi)):= \{ \lambda_1(x,\xi), ...,\lambda_l(x,\xi)\}$. Combining \eqref{chmoj} and the fact that $H_0\in S(g)$, one sees that $P_{\nu,0}\in S(1)$. For $\gamma =0$, \eqref{mm classe de symbole} is obvious. Taking the derivatives of the equation $\big(P_{\nu,0}(x,\xi)\big)^2=P_{\nu,0}(x,\xi)$, we obtain 
\begin{equation}\label{proj prop}
P_{\nu,0}(x,\xi) \partial_{(x,\xi)}^{\gamma} P_{\nu,0}(x,\xi)  P_{\nu,0}(x,\xi) = 0, \quad \forall \gamma \in \mathbb N^{2n}\setminus \{0\}.
\end{equation}
Now, by differentiating successively the equation $H_0(x,\xi)P_{\nu,0}(x,\xi)=\lambda_{\nu}(x,\xi) P_{\nu,0}(x,\xi)$ using \eqref{proj prop} and the fact that $P_{\nu,0}\in S(1)$, one gets \eqref{mm classe de symbole} for all $\gamma \in \mathbb N^{2n}\setminus \{0\}$.
\end{proof}

\textbf{Outline of the proof of Theorem \ref{projections} :}

\noindent
%\textbf{Proof of Theorem \ref{projections} :} 

Fix $1\leq \nu \leq l$ and let $\gamma_{\nu}(x,\xi)$ be the contour defined in \eqref{contour}. According to \eqref{chmoj}, for all $z\in \gamma_{\nu}(x,\xi)$, $(H_0(x,\xi)-z)$ is elliptic, i.e. $(H_0(x,\xi)-z)^{-1}\in S(g^{-1})$. By the composition formula \eqref{dév asy annexe}, we have 
\begin{eqnarray}\label{rdupar}
(H(x,\xi;\h) - z) \# (H_0(x,\xi)-z)^{-1} &=& (H_0(x,\xi) - z) \#  (H_0(x,\xi)-z)^{-1} + \h  H_1(x,\xi) \# (H_0(x,\xi)-z)^{-1} \nonumber \\
&= & I_m - \h r(x,\xi,z;\h),
\end{eqnarray}
with $r\in S(1)$, uniformly for $z\in \gamma_{\nu}(x,\xi)$. Consequently, using the symbolic calculus of $\hbar$-pseudodifferential operators (see \cite[ch. 8]{dim}), we can construct a parametrix $B\in S(g^{-1})$  such that for $z\in \gamma_{\nu}(x,\xi)$, 
\begin{equation}\label{parametrix}
B(x,\xi,z;\h)\sim\sum_{j\geq 0} \h^j B_j(x,\xi,z) \,\,\, \; \text{in}\;\; S(g^{-1}), \;\; \; \text {with } \, \, B_0(x,\xi,z)=(H_0(x,\xi)-z)^{-1},
\end{equation}
and 
%$$ B(x,\xi,z;\h) \sim (H_0(x,\xi)-z)^{-1} + \h (H_0(x,\xi)-z)^{-1} \# r + \h^2 (H_0(x,\xi)-z)^{-1} \# r \# r + \cdots \quad \text{in} \;\; S(g^{-1}) $$
%such that for all $(x,\xi)\in \mathbb R^{2n}$ and all $z\in \gamma_{\nu}(x,\xi)$,
\begin{equation}\label{new1}
B(x,\xi,z;\h)\#(H(x,\xi;\h)-z)\sim (H(x,\xi;\h)-z)\# B(x,\xi,z;\h)\sim I_m,
\end{equation}
in $S(1)$.  The above formula implies  that for  $z,\tilde z\in \gamma_{\nu}(x,\xi)$ 
$$(H(x,\xi;\h)-z)\# \Big[B(x,\xi,z;h)-B(x,\xi,\tilde z;h)\Big] \#(H(x,\xi;\h)-\tilde z)\sim (z-\tilde z)I_m,$$
$$(H(x,\xi;\h)-z)\#  B(x,\xi,z;h)\#  B(x,\xi,\tilde z;h) \# (H(x,\xi;\h)-\tilde z )\sim I_m,$$
which yields 
\begin{equation}\label{new2}
B(x,\xi,z;\h) - B(x,\xi,\tilde{z};\h) \sim (z-\tilde{z}) B(x,\xi,z;\h) \# B(x,\xi,\tilde{z};\h).
\end{equation}

Put 
\begin{equation}\label{almost projector}
\tilde{P}_{\nu}(x,\xi;\h):=\frac{i}{2\pi}\int_{\gamma_{\nu}(x,\xi)} B(x,\xi,z;\h)dz \sim \frac{i}{2\pi} \sum_{j\geq 0} \h^j \int_{\gamma_{\nu}(x,\xi)} B_j(x,\xi,z)dz.
\end{equation}
By construction  of  $ \gamma_{\nu}(x,\xi)$ and $ B(x,\xi,z;h)$, we easily see that $\tilde{P}_{\nu}(x,\xi;\h) \in S(1)$. 
%where $B_0(x,\xi,z) = (H_0(x,\xi)-z)^{-1}$ and for $j\geq 1$, $B_j(x,\xi,z) = (H_0(x,\xi)-z)^{-1}\# r\#\cdots\#r$, with $\#$ repeated $j$-times. 

Let us start by proving \eqref{diag}. As we already pointed out in the above proof, by the Cauchy theorem, a small variation of the contour $\gamma_{\nu}(x,\xi)$ does not change $\tilde{P}_{\nu}(x,\xi;h)$. Let $\tilde{\gamma}_{\nu}(x,\xi)$ be a simple closed contour with the same properties than $\gamma_{\nu}(x,\xi)$ contained inside $\gamma_{\nu}(x,\xi)$.  Clearly, \eqref{new2} remains true for $z\in \gamma_\nu(x,\xi)$ and $\tilde z\in  \tilde\gamma_\nu(x,\xi)$.

%According to \eqref{parametrix}, for all $z\in \gamma_{\nu}(x,\xi)$, $B^w(x,\h D_x,z;\h)$ coincides with $(H^w(x,\h D_x;\h)-z)^{-1}$ modulo $\mathcal{O}(\h^{\infty})$. Then, from the resolvent identity, at symbols level we have 
%$$
%B(x,\xi,z;\h) - B(x,\xi,\tilde{z};\h) \sim (z-\tilde{z}) B(x,\xi,z;\h) \# B(x,\xi,\tilde{z};\h), \quad \forall z\in \gamma_{\nu}(x,\xi), \tilde{z}\in \tilde{\gamma}_{\nu}(x,\xi).
%$$
Using \eqref{new2}, we obtain
\begin{eqnarray}
\tilde{P}_{\nu}(x,\xi;\h) \# \tilde{P}_{\nu}(x,\xi;\h) &=& \left(\frac{i}{2\pi} \right)^2 \int_{\gamma_{\nu}(x,\xi)} \int_{\tilde{\gamma}_{\nu}(x,\xi)} B(x,\xi,z;\h) \# B(x,\xi,\tilde{z};\h) dzd\tilde{z} \nonumber \\
&\sim & \left( \frac{i}{2\pi} \right)^2 \int_{\gamma_{\nu}(x,\xi)} \int_{\tilde{\gamma}_{\nu}(x,\xi)} \bigg( \frac{1}{z-\tilde{z}}B(x,\xi,z;\h)   + \frac{1}{\tilde{z}-z}  B(x,\xi,\tilde{z};\h) \bigg) dzd\tilde{z} \nonumber \\
&=:& I_1 + I_2\label{inst of}
\end{eqnarray}
where 
$$
I_1 := \left( \frac{i}{2\pi} \right)^2 \int_{\gamma_{\nu}(x,\xi)} \bigg( \int_{\tilde{\gamma}_{\nu}(x,\xi)} \frac{1}{z-\tilde{z}} d\tilde{z} \bigg) B(x,\xi,z;\h)   dz = 0
$$
$$
I_2 := \left( \frac{i}{2\pi} \right)^2 \int_{\tilde{\gamma}_{\nu}(x,\xi)} \bigg( \int_{\gamma_{\nu}(x,\xi)} \frac{1}{\tilde{z}-z}   dz \bigg) B(x,\xi,\tilde{z};\h) d\tilde{z} = \frac{i}{2\pi}\int_{\tilde{\gamma}_{\nu}(x,\xi)}  B(x,\xi,\tilde{z};\h) d\tilde{z}.
$$
This gives  \eqref{diag}. Property (\ref{diag2}) follows immediately from the selfadjointness of $H^w(x,\h D_x;\h)$ while \eqref{diag3} is a consequence of \eqref{new1}.
%To see (\ref{diag3}), notice that from \eqref{parametrix}, we deduce 
%$$
%B(x,\xi,z;\h) \# H(x,\xi;\h)  \sim H(x,\xi;\h) \# B(x,\xi,z;\h), \quad \forall z\in \gamma_{\nu}(x,\xi).
%$$
%Therefore
%\begin{eqnarray*}
%\tilde{P}_{\nu}(x,\xi;\h) \# H(x,\xi;\h) &=& \frac{i}{2\pi} \int_{\gamma_{\nu}(x,\xi)} B(x,\xi,z;\h) \# H(x,\xi;\h) dz \\
%&\sim& \frac{i}{2\pi} \int_{\gamma_{\nu}(x,\xi)} H(x,\xi;\h)\#B(x,\xi,z;\h) dz \\
%&\sim& H(x,\xi;\h) \# \tilde{P}_{\nu}(x,\xi;\h). 
%\end{eqnarray*}

In order to prove \eqref{différents indices}, we consider two contours $\gamma_{\nu}(x,\xi)$ and $\gamma_{\mu}(x,\xi)$ such that $\text{dist}(\gamma_{\nu}(x,\xi),\gamma_{\mu}(x,\xi))\geq c > 0$ and we repeat the same computation as in \ref{inst of} with $\gamma_{\mu}(x,\xi)$ instead of $\tilde{\gamma}_{\nu}(x,\xi)$. In this case $I_1=I_2 =0$.

Formula (\ref{identité}), follows from the construction of $\tilde P_\nu(x,\xi;h)$  which yields
$$\sum_{\nu=1}^l\tilde P_\nu(x,\xi;h)\sim I_m.$$
%Again using the fact that $B^w(x,\h D_x,z;\h)= (H^w(x,\h D_x;\h)-z)^{-1}$ modulo $\mathcal{O}(\h^{\infty})$, for $z\in \gamma_{\nu}(x,\xi)$, we get 
%$$\sum_{\nu=1}^l \tilde{P}_{\nu}^w(x,\h D_x;\h) = \frac{i}{2\pi} \int_{\bigcup_{\nu=1}^l \gamma_{\nu}} (H^w(x,\h D_x;\h)-z)^{-1} dz + \mathcal{O}(\h^{\infty}) = \text{id}_{\mathcal{L}(L^2(\mathbb R^n)\otimes \mathbb C^m)} + \mathcal{O}(\h^{\infty}).$$ 
%We have 
%$$
%\sum_{\nu=1}^l \tilde{P}_{\nu}(x,\xi;\h) = \frac{i}{2\pi} \int_{\bigcup_{\nu=1}^l \Gamma_{\nu}(x,\xi)} B(x,\xi,z;\h) dz \sim \sum_{j\geq 0} \h^j \frac{i}{2\pi} \int_{\bigcup_{\nu=1}^l \Gamma_{\nu}(x,\xi)} B_j(x,\xi,z) dz 
%$$
%Using the composition formula \eqref{dév asy annexe}, one can see that for all $j\geq 0$, $B_j(x,\xi,z)$ is a finite linear combination of terms of the form 
%$$
%(H_0(x,\xi)-z)^{-1} b_1(x,\xi,z) (H_0(x,\xi)-z)^{-1} ... b_{k-1}(x,\xi,z) (H_0(x,\xi)-z)^{-1}
%$$
%with $k\leq 2j+1$ and $\Gamma_{\nu}(x,\xi) \ni z\mapsto b_{\eta}(x,\xi,z)$ is analytic, $1\leq \eta \leq k-1$ (see also \cite{dim}). Then, using the Cauchy theorem, one can replace the contour $\bigcup_{\nu=1}^l \Gamma_{\nu}(x,\xi)$ by $\mathcal{C}(0,R):= \{z\in \mathbb C;\; |z|= R\}$, with $R>0$ large enough, enclosing all the eigenvalues of $H_0(x,\xi)$. Thus, we write 
%$$
%\sum_{\nu=1}^l \tilde{P}_{\nu}(x,\xi;\h) \sim \frac{i}{2\pi} \int_{\mathcal{C}(0,R)} (H_0(x,\xi)-z)^{-1} dz + \sum_{j\geq 1} \h^j \frac{i}{2\pi} \int_{\mathcal{C}(0,R)} B_j(x,\xi,z) dz
%$$
\begin{flushright}
$\square$
\end{flushright}

The following lemma is needed in the proof of Lemma \ref{needed lemma}. Put
$$
\tilde{P}_{\nu,j}(x,\xi) := \frac{i}{2\pi} \int_{\gamma_{\nu}(x,\xi)} B_j(x,\xi,z) dz, \quad  j\geq 0.
$$
\begin{lemma}\label{tilde pnuj smoo}
Under assumptions \textbf{(A1)} and \textbf{(A2)}, we have 
$$
\tilde{P}_{\nu,j} \in S(g^{-j}), \quad \forall j\geq 0.
$$
\end{lemma}
\begin{proof} For all $j\geq 0$, $B_j(x,\xi,z)$ is given by (see equation (8.11) in \cite{dim})
$$
B_j(x,\xi,z) = (H_0(x,\xi)-z)^{-1}\#^j r:= (H_0(x,\xi)-z)^{-1} \# r \#r \cdots \# r,
$$
with $\#$ repeated $j$-times. The symbol $r$ is defined in \eqref{rdupar}, more precisely 
\begin{eqnarray*}
r(x,\xi,z;\h) &=& \frac{1}{\h} \big( I_m - (H(x,\xi;\h)-z) \# (H_0(x,\xi)-z)^{-1} \big)  \\
&=& \frac{1}{\h} (I_m - (H_0(x,\xi)-z) \# (H_0(x,\xi)-z)^{-1}) - H_1(x,\xi) \# (H_0(x,\xi)-z)^{-1}.
\end{eqnarray*}
Since $(H_0(x,\xi)-z)^{-1}\in S(g^{-1})$ according to \eqref{chmoj}, it follows from assumption \textbf{(A2)} and the composition formula \eqref{dév asy annexe} that $r\in S(g^{-1})$. Then, for all $j\geq 0$
$$
(H_0(x,\xi)-z)^{-1} \#^j r \in S(g^{-(j+1)}).
$$
Consequently, $\tilde{P}_{\nu,j} \in S(g^{-j})$, for all $j\geq 0$.
\end{proof}

\textbf{Acknowledgement.} The author wishes to express his gratitude to Mouez Dimassi for suggesting the problem and many stimulating conversations. The author also acknowledges helpful discussions with Jean-Fran\c cois Bony. The author is grateful to the referee for his stimulating questions and recommandations which help to improve the paper. This research was partially supported by the program of the European Commission Erasmus Mundus Green IT.

\end{document}